\documentclass[12pt,letterpaper]{article}
\usepackage[utf8]{inputenc}
\usepackage[english]{babel}
\usepackage[T1]{fontenc}
\usepackage{mathtools}
\usepackage{booktabs}
\usepackage{diagbox}
\usepackage{float}
\usepackage{amsfonts}
\usepackage{apxproof}
\usepackage{enumitem}
\usepackage{bm}
\usepackage{setspace}
\usepackage{mathrsfs}

\usepackage[authoryear]{natbib}
\usepackage{hyperref}

\usepackage[margin=1.25in]{geometry}
\usepackage{amsthm}
\newtheoremrep{theorem}{Theorem}
\newtheorem*{remark}{Remark}

\newtheoremrep{proposition}{Proposition}
\newtheorem{corollary}{Corollary}
\newtheoremrep{lemma}{Lemma}
\newtheoremrep{assumption}{Assumption}

\newcommand{\fnorm}[1]{{\left\| #1 \right\|_F}}

\DeclareMathOperator*{\vecr}{vec}

\newcommand{\argmin}{\mathop{\rm argmin}\limits}

\title{Determining the Structure of Dynamic Factor Models}
\author{Sangmyung Ha\footnote{Department of Economics, Indiana University (Email: ha1@iu.edu). I am grateful to Joon Y. Park, Yoosoon Chang, Keli Xu, and Laura Liu for excellent advice and suggestions. I also thank the participants of the 2023 MEG Meeting, especially Bin Chen, who gave numerous useful suggestions, and the participants of 2024 SETA. All errors are my own.}}
\date{August 12, 2026}

\begin{document}

\maketitle
\begin{abstract}
We propose two procedures for determining the number of dynamic factors, extending \citet{Bai02} and \citet{Ahn-Horenstein-2013} to dynamic factor models where lagged factors may directly influence the observed variables. As an intermediate step, we develop a simple and computationally efficient alternating least squares algorithm that directly estimates the dynamic factors, rather than their static representations. By working with these direct estimates, our approach enables joint determination of the number of factors and the filter length. Our approach does not require the exact finite-order VAR specification maintained by \citet{Bai2007} and \citet{Amengual-Watson-2007}. We apply our procedures to estimate the number of primitive shocks in a large panel of U.S. macroeconomic time series.
\end{abstract}
\vfill


\newpage
\section{Introduction}
Factor models are a central tool in empirical macroeconomics and finance for summarizing co-movements in large-dimensional time series. Their appeal lies in the idea that a small number of latent factors can explain the dynamics of many observed variables.

A fundamental distinction is between \emph{static} and \emph{dynamic} factor models. In a static factor model, the observables $x_t$ depend only contemporaneously on a low-dimensional vector of factors $f_t$, e.g., $x_t = \lambda f_t + \varepsilon_t$, whereas a dynamic factor model allows lags of the factors to affect the observables, e.g., $x_t = \sum_{k=0}^{m-1} \lambda_k f_{t-k} + \varepsilon_t$. The dynamic specification is more general and economically informative, as it captures delayed responses and propagation mechanisms that are central to macroeconomic and financial dynamics.\footnote{The dynamic factor model studied in this paper allows the factor process $(f_t)$ to exhibit arbitrary serial dynamics, provided that it satisfies the Assumption \ref{pd} to be introduced in Section \ref{sec:model}. This is in contrast to the generalized dynamic factor model of \cite{Forni2000}.} 

While a large body of work has developed methods for estimating the number of static factors—see, for example, \citet{Bai02}, \citet{Onatski-2010}, and \citet{Ahn-Horenstein-2013}—the literature on estimating the number of dynamic factors remains relatively limited. Notable exceptions are \citet{Bai2007} and \citet{Amengual-Watson-2007}. These papers approach the problem by constructing a static representation of the dynamic factor model. Under the assumption that the static representations of the dynamic factors follow a VAR process with a lower-dimensional innovation, \citet{Bai2007} and \citet{Amengual-Watson-2007} develop consistent procedures for estimating the number of dynamic factors. 

Although VAR models provide a flexible and widely used framework for capturing time series dynamics, the literature on the study of the static factor models shows that a VAR specification for the static factors is not necessary for either estimating the factors or determining their number. This raises the question of whether it is possible to consistently estimate the number of dynamic factors without imposing a VAR model on the factors.

This paper proposes a new approach to estimating the number of dynamic factors that does not rely on the assumption that the factors follow a VAR process. Instead, our method requires that the second moment matrix of the static representations of the dynamic factors be positive definite, a substantially weaker condition than imposing a parametric VAR structure. Our approach is enabled by directly estimating the dynamic factors using an alternating least squares algorithm. The procedure iteratively updates factors and loadings via simple least squares steps, making it computationally efficient and scalable to large datasets. When applied to static factor models, the alternating least squares estimator coincides with the principal components estimator commonly used in static factor models. We establish consistency under double asymptotics when both the time-series dimension and the dimension of the observed variable diverge, obtaining the standard convergence rates typically associated with the estimation of static factors. By leveraging consistent estimators of the dynamic factors, we propose two procedures that jointly estimate the number of dynamic factors and the filter length, defined as the length of the horizon over which current and lagged factors exert a direct effect on the observed variables.  

Our approach extends the criteria of \citet{Ahn-Horenstein-2013} and \citet{Bai02} to dynamic factor models. Under general Assumptions to be introduced in Sections \ref{sec:model} and \ref{sec:test}, the number of dynamic factors and the filter length of the dynamic factor model are jointly identified asymptotically independently of any parametric specification for the stochastic process governing the dynamic factors. In particular, if the dynamic factors follow a VAR\((p)\) process, the identification of the filter length is invariant to the choice of the lag order \( p \). We also address the issue of VAR lag order selection using estimated factors when the dynamic factors follow a VAR process.

The eigenvalue ratio test of \citet{Ahn-Horenstein-2013} identifies the number of factors by detecting the point at which the ratio of consecutive singular values of the data matrix diverges. We extend this idea to dynamic factor models by introducing a notion of 'dynamic' singular values, defined as the spectral norm of the residual matrix after removing a given number of estimated dynamic factors. Building on this concept, we propose the dynamic singular value ratio test, which generalizes the eigenvalue ratio test of \citet{Ahn-Horenstein-2013} and coincides with it when applied to a static factor model. Moreover, the proposed framework can be used to determine the filter length of the dynamic factor model by examining the behavior of ratios of consecutive dynamic singular values as the filter length varies for a fixed number of dynamic factors. 

The information criterion of \citet{Bai02} selects the number of static factors by augmenting the residual variance with a penalty term that increases with the number of static factors. Our second procedure extends this approach to dynamic factor models by employing a penalty term that depends on both the number of dynamic factors and the filter length. When we restrict the filter length to one, the proposed criterion coincides with that of \citet{Bai02}, up to a constant factor of two in the penalty term. 

\subsection{Related literature and contributions of this paper}
There is a large literature on dynamic factor models, which can be broadly divided into two classes. The first class comprises generalized dynamic factor models with infinite filter length, which do not admit a finite-dimensional static representation. The second class consists of models with a finite filter length, for which a static representation exists. Our work falls into the second category: the dynamic factor models with a finite filter length. What distinguishes our work from many existing ones is that we work directly with the dynamic factors themselves, rather than the static representations.

The generalized dynamic factor model with infinite filter length was introduced by \citet{Forni2000}, who proposed estimation based on frequency-domain principal components. Subsequent work has developed one-sided representations, estimation methods, and asymptotic theory for these models: \citet{Forni2015}, \citet{Forni-Hallin-Lippi-Zaffaroni-2017}, and \citet{Barigozzi-Hallin-Luciani-Zaffaroni-2023}. Methods for determining the number of dynamic factors in the generalized dynamic factor framework have also been proposed by \citet{Hallin2007} and \citet{Onatski-2010}. However, the greater generality afforded by infinite-length filters typically requires the dynamic factors to be serially uncorrelated processes for identification. We focus our attention on dynamic factor models with finite-length filter in this paper.

A separate strand of the literature studies the static representations of dynamic factor models. \citet{Bai-2003} and \citet{Bai-Li-2016} provide inferential theory for static factors estimated by principal components and quasi–maximum likelihood, respectively. Procedures for determining the number of static factors include information criteria \citep{Bai02} and eigenvalue-ratio tests \citep{Ahn-Horenstein-2013}, among others. Despite this extensive literature, comparatively little attention has been paid to the direct estimation of dynamic factors. An exception is \citet{Bai2015}, who studies Bayesian estimation of dynamic factor models. This paper is most closely related to \citet{Bai2007} and \citet{Amengual-Watson-2007}, who study the estimation of the number of dynamic factors in models with finite filter length. 

To summarize, our contributions are as follows. 

First, we generalize the procedures of \citet{Bai02} and \citet{Ahn-Horenstein-2013} to estimate the number of dynamic factors in dynamic factor models. Our approach is more robust than those of \citet{Bai2007} and \citet{Amengual-Watson-2007}, as it does not require imposing a VAR structure on the static factors. 

Second, our procedure provides a consistent estimator of the filter length, jointly with the number of dynamic factors, which, to the best of our knowledge, is the first in the literature. Identifying the extent to which primitive shocks directly affect the economy is of independent interest. For instance, using the FRED--MD dataset of \citet{McCracken-Ng-2015} over the sample period March 1973 to March 2025, the filter length is estimated to be two.\footnote{The associated number of dynamic factors is found to be four, which increases to seven after the onset of COVID.}

Third, we examine the finite-sample performance of our selection procedures and show that they outperform those of \citet{Bai2007} and \citet{Amengual-Watson-2007} across a range of designs in determining $q$. 

Finally, we establish the consistency of the proposed alternating least squares estimator in estimating the dynamic factors. Furthermore, when the dynamic factors follow a VAR$(p)$ process, we show that the lag order selection can also be done using the information criteria methods, with unobserved factors replaced with their estimates.

The remainder of the paper is organized as follows. Section \ref{sec:model} presents the model and discusses representations of the dynamic factors. Section \ref{sec:estimation} describes the alternating least squares estimation procedure and the consistency of the proposed method. Section \ref{sec:test} develops two methods for determining the number of dynamic factors and the filter length. Section \ref{sec:simulations} reports simulation evidence on the finite-sample performance of the proposed procedures. Section \ref{sec:empirical} provides an empirical application. Section \ref{sec:conclusion} concludes. All technical proofs are collected in the Appendix.

\paragraph{Notations.}
We denote by $\|\cdot\|$ the Euclidean norm when applied to vectors and the associated spectral (operator) norm when applied to matrices. The Frobenius norm is denoted by $\fnorm{\cdot}$. For a symmetric matrix $A$, the notation $A>0$ and $A\geq0$ indicates that $A$ is positive definite and positive semi-definite, respectively. We denote $P_A=A(A'A)^{-1}A'$ and $M_A=I-P_A$ for full column matrix $A$. We write $\lceil x \rceil$ for the smallest integer greater than or equal to $x$ and $\lfloor x\rfloor$ for the largest integer less than or equal to $x$. Throughout, $q$ denotes the number of dynamic factors, $m$ the filter length, and $r = qm$ the corresponding number of static factors. Finally, we denote $\mathbb{Z}$, $\mathbb{N}_0$, and $\mathbb{R}$ the set of integers, nonnegative integers, and real numbers, respectively.

\section{Model, Assumptions, and Representations}\label{sec:model}

\subsection{Model}
We denote the pair $(q,m)$ representing the number of dynamic factors and the filter length as the structure of a dynamic factor model. For a given structure $(q,m)$, we consider the finite-order dynamic factor model
\begin{equation}\label{eq:dynamic_factor_model}
x_{it}=\chi_{it}+\varepsilon_{it},\qquad
\chi_{it}=\sum_{k=0}^{m-1}\lambda_{ik}'f_{t-k},\qquad
i=1,\ldots,N,\quad t=1,\ldots,T,
\end{equation}
where $f_t\in\mathbb R^q$ is the vector of dynamic factors, $\lambda_{ik}\in\mathbb R^q$ is the loading vector associated with lag $k$, and $m$ is the filter length. We assume that the factor process is defined for all $t\in\mathbb Z$, so the required presample lags are well defined.

When $m=1$, \eqref{eq:dynamic_factor_model} coincides with a static factor model. If the finite-order restriction is relaxed and the loading filter is allowed to be infinite under appropriate summability conditions, the model is related to the generalized dynamic factor framework of \cite{Forni2000}. We restrict attention to the finite-filter case $m<\infty$, which can be viewed as a finite-order dynamic extension of the approximate static factor models studied in \cite{Bai-2003} and \cite{Stock02}.

Let $r=qm$ and define the stacked static factor and loading vectors by
\[
g_t=
\begin{pmatrix}
f_t'&f_{t-1}'&\cdots&f_{t-m+1}'
\end{pmatrix}'
\in\mathbb R^r,
\qquad
\gamma_i=
\begin{pmatrix}
\lambda_{i0}'&\lambda_{i1}'&\cdots&\lambda_{i,m-1}'
\end{pmatrix}'
\in\mathbb R^r.
\]
Then \eqref{eq:dynamic_factor_model} has the algebraically equivalent static representation
\begin{equation}\label{eq:static_representation_individual}
x_{it}=\gamma_i'g_t+\varepsilon_{it}.
\end{equation}
For each $k=0,\ldots,m-1$, define the $N\times q$ loading matrix
\[
\lambda_k=
\begin{pmatrix}
\lambda_{1k}'\\
\vdots\\
\lambda_{Nk}'
\end{pmatrix},
\]
and let
\[
x_t=
\begin{pmatrix}
x_{1t}\\
\vdots\\
x_{Nt}
\end{pmatrix},
\qquad
\varepsilon_t=
\begin{pmatrix}
\varepsilon_{1t}\\
\vdots\\
\varepsilon_{Nt}
\end{pmatrix},
\qquad
\Gamma=
\begin{pmatrix}
\lambda_0&\lambda_1&\cdots&\lambda_{m-1}
\end{pmatrix}
\in\mathbb R^{N\times r}.
\]
Equivalently, the $i$-th row of $\Gamma$ is $\gamma_i'$. The vector form of the model is therefore
\begin{equation}\label{eq:static_representation_cross_section}
x_t=\Gamma g_t+\varepsilon_t.
\end{equation}
Stacking the observations over time, define
\[
X=
\begin{pmatrix}
x_1'\\
\vdots\\
x_T'
\end{pmatrix}
\in\mathbb R^{T\times N},
\qquad
E=
\begin{pmatrix}
\varepsilon_1'\\
\vdots\\
\varepsilon_T'
\end{pmatrix}
\in\mathbb R^{T\times N},
\qquad
G=
\begin{pmatrix}
g_1'\\
\vdots\\
g_T'
\end{pmatrix}
\in\mathbb R^{T\times r}.
\]
The model can then be written in static matrix form as
\begin{equation}\label{eq:static_representation_matrix}
X=G\Gamma'+E.
\end{equation}
For $k=0,\ldots,m-1$, let
\[
F_{-k}=
\begin{pmatrix}
f_{1-k}'\\
\vdots\\
f_{T-k}'
\end{pmatrix}
\in\mathbb R^{T\times q}
\]
collect the $k$-th lag of the dynamic factors. It follows that
\[
G=
\begin{pmatrix}
F_0&F_{-1}&\cdots&F_{-(m-1)}
\end{pmatrix},
\]
and hence the dynamic matrix representation is
\begin{equation}\label{eq:dynamic_representation_matrix}
X=\sum_{k=0}^{m-1}F_{-k}\lambda_k'+E.
\end{equation}
For convenience, we write $F=F_0$ for the $T\times q$ matrix collecting the contemporaneous dynamic factors.

Although \eqref{eq:static_representation_matrix} has the algebraic form of a static factor model with $r=qm$ factor coordinates, the columns of $G$ are not unrestricted. In particular, the $m$ blocks of $g_t$ are lagged copies of the same $q$-dimensional process. If $g_{jt}$ denotes the $j$th $q$-dimensional block of $g_t$, then
\[
g_{1t}=f_t,\qquad g_{j+1,t}=g_{j,t-1},\qquad j=1,\ldots,m-1.
\]
Relative to an unrestricted static factor model with $r$ factors, the dynamic parameterization therefore uses only $q$ underlying factor series and imposes cross-time shift restrictions on their stacked representation. This restriction is the source of its parsimony; the number of loading coefficients remains $Nqm$ in both representations.

The common component may admit multiple dynamic representations. To distinguish the true structure from generic candidate structures, we let $q_0$ denote the smallest factor dimension among all exact representations of the common component. Among representations attaining $q_0$, let $m_0$ denote the smallest filter length, and define $r_0=q_0m_0$. We refer to $(q_0,m_0)$ as the true structure. We use $(q,m)$ and $r=qm$ to denote a generic candidate structure. The existence and characterization of alternative observationally equivalent representations are discussed in Section \ref{sec:dynamic_representations}.

\subsection{Assumptions}

Throughout, $q_0$ and $m_0$ are fixed, $r_0=q_0m_0$, and $N,T\to\infty$ jointly. Recall that $G\in\mathbb R^{T\times r_0}$ and $\Gamma\in\mathbb R^{N\times r_0}$ denote the stacked factor and loading matrices associated with the true structure $(q_0,m_0)$. The following two assumptions are used to establish consistency of the estimated factor and loading spaces. Additional assumptions used to establish consistency of the procedures for determining $(q_0, m_0)$ are introduced in Section \ref{sec:test}.

\begin{assumption}[Nondegeneracy and pervasiveness]\label{pd}
The stacked factor and loading matrices associated with the true structure satisfy
\[
\frac{1}{T}G'G\to_p\Sigma_g,
\qquad
\frac{1}{N}\Gamma'\Gamma\to_p\Sigma_\gamma,
\]
as $N,T\to\infty$, where $\Sigma_g$ and $\Sigma_\gamma$ are finite $r_0\times r_0$ positive-definite matrices.
\end{assumption}

\begin{assumption}[Weak idiosyncratic component]
\label{error_norm}
The idiosyncratic-error matrix satisfies
\[
\|E\|
=
O_p\left(\sqrt{\max\{N,T\}}\right).
\]
\end{assumption}

The two parts of Assumption \ref{pd} have distinct roles. The first requires the lag-stacked factor vector
\[
g_{t}
=
\begin{pmatrix}
f_t'&f_{t-1}'&\cdots&f_{t-m_0+1}'
\end{pmatrix}'
\]
to have a nonsingular asymptotic second-moment matrix. The second is the usual pervasiveness condition, requiring each true static-factor direction to have nonvanishing cross-sectional strength.

A convenient sufficient condition for the factor-side restriction is that $(f_t)$ is a covariance-stationary and ergodic process with finite second moments whose spectral density matrix is positive definite on a set of frequencies with positive Lebesgue measure. Under this condition,
\[
\mathbb E(g_{t}g_{t}')>0,
\]
and an ergodic law of large numbers yields
\[
\frac{1}{T}G'G\to_p \mathbb E(g_{t}g_{t}').
\]
This condition includes, for example, stable and invertible VARMA processes driven by innovations with positive-definite covariance. In particular, a stable VAR driven by full-rank innovations has a spectral density matrix that is positive definite at every frequency; see, for example, \cite[Sections 11.3 and 11.8]{BrockwellDavis1991}. The factor-side condition in Assumption \ref{pd} is therefore weaker than requiring a finite-order VAR specification.\footnote{The factor-side condition in Assumption \ref{pd} is discussed in greater detail in Supplemental Appendix \ref{apx:discussion_assumption}.}
Analogous nondegeneracy and pervasiveness conditions are maintained in \cite{Bai2007} and \cite{Amengual-Watson-2007} when estimating the number of static factors. Unlike their fixed-order VAR formulations, however, Assumption \ref{pd} does not require the factors to possess an exact finite-order autoregressive representation. For example, a generic invertible VMA$(1)$ process has a VAR$(\infty)$ representation and can satisfy Assumption \ref{pd}, even though it does not generally admit an exact finite-order VAR representation. \cite{Bai2007} discuss finite-order VAR approximations to such processes but note that their fixed-order theory does not directly cover the MA case.

Assumption \ref{error_norm} places a global bound on the strength of the idiosyncratic component. It permits weak cross-sectional and serial dependence in the errors but rules out an additional idiosyncratic component with singular values of the same order as those generated by the strong common component. Spectral-norm and related eigenvalue restrictions of this type are standard in the factor-model literature; see, for example, \cite{Onatski-2010}, \cite{Ahn-Horenstein-2013}, \cite{Moon_Weidner_2017}, and \cite{Bai-Ng-2023}.

\subsection{Equivalent Dynamic Representations}
\label{sec:dynamic_representations}

Recall that the true common component is
\[
\chi_t=\sum_{k=0}^{m_0-1}\lambda_kf_{t-k},
\qquad
f_t\in\mathbb R^{q_0}.
\]
Throughout this subsection, we restrict attention to exact finite-order causal linear reparameterizations of the true factor process. Specifically, an alternative representation with structure $(q,m)$ is of the form
\[
h_t=H(L)f_t=\sum_{\ell=0}^{d}H_\ell f_{t-\ell},
\qquad
\chi_t=\sum_{k=0}^{m-1}\widetilde \Gamma_kh_{t-k},
\]
for some finite $d$, where $H_\ell\in\mathbb R^{q\times q_0}$ and $\widetilde \Gamma_k\in\mathbb R^{N\times q}$. We interpret $m$ as the maximum candidate filter length, so some $\widetilde \Gamma_k$ may be zero, and the coordinates of $h_t$ may be redundant.

The true representation is minimal within this class: no exact representation exists with factor dimension $q<q_0$, and, among representations with $q=q_0$, no exact representation exists with filter length $m<m_0$. Two dynamic factor representations are called equivalent if they generate the same common component $\chi_t$ for every $t$. When combined with the same idiosyncratic component $\varepsilon_t$, the resulting models for $x_t$ are observationally equivalent.

The true representation $(q_0,m_0)$ always admits the static structure $(r_0,1)$ obtained by stacking the current and lagged values of $f_t$. There may exist other admissible dynamic representations that are observationally equivalent. To illustrate, consider a true model with structure $(q_0,m_0)=(2,3)$:
\[
\chi_t=\lambda_0f_t+\lambda_1f_{t-1}+\lambda_2f_{t-2},
\qquad
f_t\in\mathbb R^2.
\]
Its static representation is
\begin{equation}
\label{eq:example_static}
\chi_t=
\begin{pmatrix}
\lambda_0&\lambda_1&\lambda_2
\end{pmatrix}
g_t,
\qquad
g_t=
\begin{pmatrix}
f_t\\
f_{t-1}\\
f_{t-2}
\end{pmatrix}
\in\mathbb R^6.
\end{equation}
Thus, $(q,m)=(6,1)$ is an admissible structure. The same common component also admits a representation with structure $(q,m)=(3,2)$. Define
\[
h_t=
\begin{pmatrix}
f_{t,1}\\
f_{t,2}+f_{t-1,1}\\
f_{t-1,2}
\end{pmatrix}
\in\mathbb R^3
\]
and let
\[
\widetilde\Gamma_0=
\begin{pmatrix}
\lambda_{0,1}&
\lambda_{0,2}&
\lambda_{1,2}-\lambda_{2,1}
\end{pmatrix},
\qquad
\widetilde\Gamma_1=
\begin{pmatrix}
\lambda_{1,1}-\lambda_{0,2}&
\lambda_{2,1}&
\lambda_{2,2}
\end{pmatrix},
\]
where $\lambda_{k,j}$ denotes the $j$-th column of $\lambda_k$. Then
\begin{equation}
\label{eq:example_dynamic}
\chi_t=\widetilde\Gamma_0h_t+\widetilde\Gamma_1h_{t-1}.
\end{equation}
Hence, $(q,m)=(3,2)$ is also admissible.

The example illustrates the tradeoff between the number of factor coordinates and the filter length. Reducing the filter length below $m_0$ requires additional factor coordinates, while increasing the factor dimension weakly reduces the minimum admissible filter length. The following proposition shows that admissibility is governed by the factor dimension $q$ and the comparison between $r=qm$ and the true static rank $r_0=q_0m_0$.

\begin{proposition}[Characterization of admissible dynamic structures]
\label{alt_spec}
Let $r_0=q_0m_0$. Every pair $(q,m)\in\mathbb N^2$ satisfying
\[
q\geq q_0
\qquad\text{and}\qquad
qm\geq r_0
\]
is an admissible dynamic structure.

Conversely, every admissible dynamic structure satisfies $q\geq q_0$. Moreover, under Assumption \ref{pd}, an exact representation with $qm<r_0$ is ruled out with probability approaching one as $N,T\to\infty$. Thus, admissible dynamic structures are asymptotically characterized by
\[
q\geq q_0
\qquad\text{and}\qquad
qm\geq q_0m_0.
\]
\end{proposition}
Proposition \ref{alt_spec} implies that structures satisfying $qm=r_0$ lie on the boundary of the admissible region and may provide nonredundant representations of the common component. However, the true structure $(q_0,m_0)$ is the most parsimonious representation because its factor dimension is the smallest. Structures satisfying $qm>r_0$ are also admissible, but necessarily contain redundant static factor coordinates or loading directions under Assumption \ref{pd}.

\begin{corollary}[Lower bounds on admissible structures]
\label{cor:lower_bounds}
Suppose that $q_0,m_0\geq1$ and let $r_0=q_0m_0$. Under Assumption
\ref{pd}, the minimum admissible factor dimension for a fixed
$m\geq1$ is
\[
q_{\min}(m)
=
\max\left\{
q_0,
\left\lceil\frac{r_0}{m}\right\rceil
\right\}
=
\begin{cases}
\displaystyle
\left\lceil\frac{q_0m_0}{m}\right\rceil,
&1\leq m<m_0,\\[2mm]
q_0,
&m\geq m_0.
\end{cases}
\]
For a fixed $q\geq q_0$, the minimum admissible filter length is
\[
m_{\min}(q)
=
\left\lceil\frac{r_0}{q}\right\rceil
=
\left\lceil\frac{q_0m_0}{q}\right\rceil.
\]
These lower bounds are asymptotically sharp.
\end{corollary}
For $1\leq m<m_0$, the minimum number of additional factor coordinates required to reduce the candidate filter length to $m$ is
\[
q_{\min}(m)-q_0
=
\left\lceil\frac{q_0m_0}{m}\right\rceil-q_0.
\]
Moreover,
\[
m_{\min}(q_0)=m_0,
\]
and $m_{\min}(q)$ is weakly decreasing in $q$. Because of the ceiling operator, increasing $q$ need not strictly reduce the minimum filter length at every value of $q$.

Let $q_{max}\geq q_0$ and $m_{max}\geq m_0$ be given, and let
\[
S
=
\{0,\ldots,q_{\max}\}
\times
\{0,\ldots,m_{\max}\}.
\]
Consider the disjoint partition of $S$ given by
\begin{equation}
\begin{split}
\label{partition_set}
S_1
&=
\left\{
(q,m)\in S:
q\geq q_0,\quad
qm\geq r_0
\right\},\\
S_2
&=
\left\{
(q,m)\in S:
qm<r_0
\right\},\\
S_3
&=
\left\{
(q,m)\in S:
q<q_0,\quad
qm\geq r_0
\right\}.
\end{split}
\end{equation}
Proposition \ref{alt_spec} shows that every structure in $S_1$ admits an exact representation of the true common component. Structures in $S_2$ have static dimension smaller than the true static rank $r_0=q_0m_0$. Under Assumption \ref{pd}, the matrix of true common components has rank $r_0$ with probability approaching one, whereas a representation with structure $(q,m)$ has rank at most $qm$. Hence, structures in $S_2$ cannot reproduce the true common component with probability approaching one. Structures in $S_3$ have sufficient static dimension, $qm\geq q_0m_0$, but fewer than $q_0$ dynamic factor coordinates. The static-rank argument therefore does not distinguish these structures from the true model. Their misspecification instead follows from the minimality of $q_0$.

\section{Estimation}\label{sec:estimation}
In the dynamic factor model,
\[
x_t = \sum_{k=0}^{m-1} \lambda_k f_{t-k} + \varepsilon_t, \quad t=1,\ldots,T,
\]
the dynamic factors $(f_t)$ are indexed from $t = 2-m$ to $t = T$ to account for the lagged factors entering the equation. We jointly estimate $(f_t)_{t=2-m}^{T}$ and $(\lambda_k)_{k=0}^{m-1}$ by minimizing the least squares objective
\begin{align}\label{obj_fn}
    \frac{1}{NT} \fnorm{X - \sum_{k=0}^{m-1} F_{-k} \lambda_k'}^2 = \frac{1}{NT}\sum_{t=1}^T\left\|x_t-\sum_{k=0}^{m-1} \lambda_k f_{t-k}\right\|^2
\end{align}
over $(f_t)$ and $(\lambda_k)$. The first-order conditions for the factor loadings $(\lambda_j)$ are
\begin{align}\label{foc_lambda}
    \sum_{k=0}^{m-1} \hat \lambda_k \hat F_{-k}' \hat F_{-j} = X' \hat F_{-j}, \quad j=0,\ldots,m-1.
\end{align}
The first-order conditions for $(f_t)$ are
\begin{align}\label{foc_f}
    \sum_{j=\max(0,1-t)}^{\min(m-1,T-t)}\lambda_j'\left(x_{t+j}-\sum_{k=0}^{m-1}\lambda_kf_{t+j-k}\right) = 0
\end{align}
for $t=2-m,\ldots,T$.

We minimize \eqref{obj_fn} using an alternating least squares procedure. We start with an initial guess for $(\hat \lambda_k)$. Given $(\hat \lambda_k)$, we may update $(f_t)$ by solving \eqref{foc_f}. Given $(\hat f_t)$, we may update $(\lambda_k)$ by solving \eqref{foc_lambda}. We repeat the updating steps until convergence. The first-order conditions \eqref{foc_lambda} define a multivariate regression with rank $qm$, which has a closed-form solution. The first-order condition \eqref{foc_f} can be efficiently solved using the discrete Fourier transform, requiring inversion of only $q \times q$ matrices.\footnote{Derivation of the discrete Fourier transform needed to solve the first order condition is given in the Supplemental Materials \ref{sec:DFF}.} Since each step minimizes a least squares objective, the algorithm guarantees a monotone decrease of the objective function. The minimization problem is not convex, and the alternating least squares algorithm generally does not converge to the global minimum. We try multiple initial values and choose the solution that obtains the smallest objective function value in practice.

\subsection{Consistency}\label{sec:consistency}
Before introducing the selection criteria for the structure $(q,m)$ of the dynamic factor models, we first show that the estimated dynamic factors are consistent.

\begin{proposition}[Factor-space consistency under admissible structures]
\label{consistency_1}
Let Assumptions \ref{pd} and \ref{error_norm} hold. For any fixed $(q,m)\in S_1$, let $(\hat G,\hat\Gamma)$ be a global minimizer of the least-squares objective \eqref{obj_fn} over representations with structure $(q,m)$. Then
\[
\frac{1}{\sqrt T}\fnorm{M_{\hat G}G}
=
O_p\left(
\frac{1}{\sqrt N}+\frac{1}{\sqrt T}
\right),
\]
where $G$ is the true stacked factor matrix.
\end{proposition}
\paragraph{Remark}
Proposition \ref{consistency_1} establishes one-sided factor-space consistency for every fixed admissible structure $(q,m)\in S_1$. Thus, the estimated stacked factor space asymptotically contains the true factor space. When $qm=r_0=q_0m_0$, this result yields the usual consistency of the estimated factor space. When $qm>r_0$, the estimated space may contain additional redundant directions.

Under the restriction $m=1$, the admissible structures are
\[
S_1\cap\{(q,m):m=1\}
=
\{(q,1):q\geq r_0\}.
\]
For this unrestricted static-factor problem, Theorem 1 of \cite{Bai02} establishes average consistency, up to a possibly rank-deficient transformation, of the rescaled principal-component factor estimator for any fixed number of estimated factors. Proposition \ref{consistency_1} provides a dynamic analogue of this static result by allowing every admissible structure $(q,m)\in S_1$, including structures with $m>1$.

For $m>1$, the dynamically constrained least-squares estimator imposes the lag-stacking restriction
\[
\hat g_t=
\begin{pmatrix}
\hat f_t\\
\hat f_{t-1}\\
\vdots\\
\hat f_{t-m+1}
\end{pmatrix}.
\]
Equivalently, the $m$ blocks of $\hat g_t$ are shifted copies of the same estimated $q$-dimensional factor process. An unrestricted principal-component estimator of the static factors does not impose this temporal restriction, so its estimated factor coordinates need not satisfy the displayed relation.

This distinction is relevant for determining the dynamic structure of the model. Unrestricted principal components identifies the static factor space and its dimension $r_0=q_0m_0$, but it does not separately impose a decomposition into the dynamic factor dimension $q_0$ and the filter length $m_0$. The dynamically constrained estimator preserves the temporal structure used by the subsequent procedures for determining these two dimensions.

\section{Determining the Number of Dynamic Factors and Filter Length}\label{sec:test}

\subsection{Selection Based on the Dynamic Singular Value Ratio}
Throughout this section,  $\sigma_j(X)$ denotes the $j$-th largest singular value of $X$. We define a dynamic analogue of the singular-value approximation error by imposing the lag-stacking restrictions of a dynamic factor model on the low-rank approximating matrix. Formally, let $\delta_{q+1,m}(X)$ denote the $(q+1)$-th dynamic singular value of a matrix $X$ with filter length $m$, defined as
\[
\delta_{q+1,m}(X) = \inf_{C\in \mathcal{C}_{q,m}} 
\left\| X - C \right\|,
\]
where
\[
\mathcal{C}_{q,m}=\left\{\sum_{k=0}^{m-1}A_{-k}B_k':B_k\in\mathbb{R}^{N\times q}, A_{-k}=\begin{pmatrix}
    a_{1-k}'\\
    \vdots\\
    a_{T-k}'
\end{pmatrix},a_{2-m},\ldots,a_T\in\mathbb{R}^q\right\}.
\]
with $q \ge 0$ and $m \ge 0$. For $q=0$ or $m=0$, we define $\mathcal C_{q,m}=\{0\}$ so that $\delta_{1,m}(X)$ or $\delta_{q+1,0}(X)$ denotes spectral norm of $X$. 

When $m=1$, $\delta_{q+1,1}(X)$ coincides with the usual $(q+1)$-th singular value of $X$ because $\mathcal{C}_{q,1}=\{C:\text{rank}(C)\leq q\}$. For general $m$, every $C\in\mathcal{C}_{q,m}$ has rank at most $qm$. Therefore,
\begin{align}\label{eq:dynamic_static_singular_bound}
    \delta_{q+1,m}(X) \geq \delta_{qm+1,1}(X)
\end{align}
The inequality reflects the additional temporal restrictions imposed by the dynamic approximation relative to an unrestricted rank-$qm$ static approximation. More generally, given $(q_0,m_0)$ and its associated set $S_1$, we have
\begin{align*}
    \delta_{q+1,m}(X) \leq \delta_{q_0+1,m_0}(X)
\end{align*}
for any integer $(q,m)\in S_1$ because $\mathcal{C}_{q_0,m_0}\subseteq \mathcal{C}_{q,m}$ due to Proposition \ref{alt_spec}.

Let
\[
C_0
=
\sum_{k=0}^{m_0-1}F_{-k}\lambda_k'
=
G\Gamma'
\]
denote the $T\times N$ matrix of true common components.

\begin{assumption}\label{lower_bound_E}
For any $\epsilon>0$, there exists a constant $c_0>0$ such that
\[
\mathbb{P}\!\left( \frac{1}{\sqrt{N+T}} \, \delta_{k_{max}+1,1}(E) < c_0 \right) \le \epsilon
\]
for all sufficiently large $N$ and $T$, where $k_{\max}$ is fixed.
\end{assumption}

\begin{assumption}\label{lower_bound_common}
For any $\epsilon>0$, there exists a constant $c_2>0$ such that
\[
\mathbb P\left(
\min_{\substack{0\leq q<q_0\\1\leq m\leq m_{\max}}}
\frac{1}{\sqrt{NT}}\,
\delta_{q+1,m}(C_0)
<c_2
\right)
\leq\epsilon
\]
and for all sufficiently large $N$ and $T$, where $m_{max}\geq m_0$ is fixed.
\end{assumption}

Assumption \ref{lower_bound_E} requires that the leading singular values of $\frac{1}{\sqrt{\max(N,T)}}E$ be bounded away from zero in probability. Under conditions such as those in \citet{Ahn-Horenstein-2013}, a positive fraction of the singular values of $\frac{E}{\sqrt{\max\{N,T\}}}$ remain bounded away from zero. Assumption \ref{lower_bound_E}, which requires this property only for a fixed number of leading singular values, is correspondingly weaker.

Assumption \ref{lower_bound_common} imposes a corresponding lower bound on the dynamic approximation error of the true common component under every underfactored structure. Part of this condition follows directly from Assumption \ref{pd}. For every $(q,m)$ satisfying $0\leq q<q_0$ and $1\leq m\leq m_0$, we have
\[
qm
\leq
(q_0-1)m_0
<
r_0.
\]
Therefore, by \eqref{eq:dynamic_static_singular_bound},
\[
\delta_{q+1,m}(C_0)
\geq
\sigma_{qm+1}(C_0)
\geq
\sigma_{r_0}(C_0).
\]
The smallest nonzero singular value of $C_0=G\Gamma'$ satisfies
\[
\frac{\sigma_{r_0}^2(C_0)}{NT}
\geq \lambda_{\min}
\left(
\frac{G'G}{T}
\right)
\lambda_{\min}
\left(
\frac{\Gamma'\Gamma}{N}
\right).
\]
Assumption \ref{pd} implies that the right-hand side is bounded away from zero in probability. It follows that
\[
\min_{\substack{0\leq q<q_0\\1\leq m\leq m_0}}
\frac{1}{\sqrt{NT}}\,
\delta_{q+1,m}(C_0)
\]
is bounded away from zero in probability.

Thus, Assumption \ref{pd} already provides the required separation for underfactored structures whose candidate filter length does not exceed $m_0$. For $m>m_0$, however, it is possible to have
\[
qm\geq r_0
\]
even when $q<q_0$. In this case, static rank alone does not provide a lower bound on the approximation error. Assumption \ref{lower_bound_common} strengthens the minimality of $q_0$ by requiring underfactored representations to remain uniformly separated from the true common component at the $\sqrt{NT}$ scale over the fixed candidate range
\[
1\leq m\leq m_{\max}.
\]

Supplemental Appendix \ref{apx:discussion_assumption} provides sufficient conditions for Assumption \ref{lower_bound_common}. In particular, under the assumption that the optimized sample mean-squared approximation criterion converges to its population counterpart, a stable and invertible VARMA factor process driven by innovations with positive-definite covariance satisfies the spectral component of these sufficient conditions. The resulting frequency-domain condition is related to the dynamic-eigenvalue restrictions used in the generalized dynamic factor literature; see, for example, \cite{Forni2000}, \cite{Forni2015}, and \cite{Forni-Hallin-Lippi-Zaffaroni-2017}.

We construct ratio-based selection criteria obtained from residuals obtained from the Frobenius norm least-squares estimator. For each $(q,m)$ with $q,m\geq1$, let $(\hat F_{-k}^{(q,m)},\hat\lambda_k^{(q,m)})_{k=0}^{m-1}$ form any global minimizer of the least-squares objective \eqref{obj_fn}. Define the fitted common component by
\begin{equation*}
\hat C_{q,m}
=
\sum_{k=0}^{m-1}
\hat F_{-k}^{(q,m)}
\hat\lambda_k^{(q,m)\prime}.
\end{equation*}
For $q=0$ or $m=0$, set $\hat C_{q,m}=0$. We define $\hat\delta_{q+1,m}(X)$ as
\[
\hat\delta_{q+1,m}(X)
=
\left\|
X-\hat C_{q,m}
\right\|,
\qquad
q,m\in\mathbb N_0.
\]
It is important to note that the minimizer of the Frobenius norm generally differs from the minimizer of the spectral norm. When $m=1$, the Frobenius minimizer is a best rank-$q$ approximation and its residual spectral norm, i.e., $\hat\delta_{q+1,1}(X)=\delta_{q+1,1}(X)$. For $m>1$, $\hat\delta_{q+1,m}(X)\neq\delta_{q+1,m}(X)$ in general. While dynamic singular values can in principle be computed by directly minimizing the spectral norm, we rely on the Frobenius-norm solution for computational tractability.

It can be shown that under Assumptions \ref{pd}, \ref{error_norm}, \ref{lower_bound_E}, and \ref{lower_bound_common}, there exists $q_{max}$ and $m_{max}$ such that for all $0\leq q\leq q_{max}$ and $0\leq m\leq m_{max}$, $\frac{\hat\delta_{q+1,m}(X)}{\sqrt{N+T}}$ is bounded above and away from zero in probability for $(q,m)\in S_1$, whereas $\frac{\hat\delta_{q+1,m}(X)}{\sqrt{NT}}$ is bounded above and away from zero in probability for $(q,m)\in S_2\cup S_3$. This behavior of $\hat\delta_{q+1,m}(X)$ motivates the dynamic singular value ratio tests given in Theorem \ref{thm:q_consistency} and \ref{thm:m_consistency} below.

In Theorem \ref{thm:q_consistency}, we fix $m$ and estimate the number of dynamic factors using a ratio-type statistic based on adjacent estimated dynamic singular values. 

\begin{theorem}[Consistency of the factor-dimension estimator]
\label{thm:q_consistency}
Suppose that $q_0,m_0\geq1$ and that Assumptions \ref{pd}, \ref{error_norm}, \ref{lower_bound_E}, and \ref{lower_bound_common} hold. Let $q_{\max}$, $m_{\max}$, and $k_{\max}$ be fixed positive integers satisfying
\[
q_0\leq q_{\max},
\qquad
m_0\leq m_{\max},
\qquad
q_0m_0+q_{\max}m_{\max}\leq k_{\max}.
\]
Fix $m\in\{1,\ldots,m_{\max}\}$ and define
\[
q_m^\star
=
\max\left\{
q_0,
\left\lceil\frac{q_0m_0}{m}\right\rceil
\right\}.
\]
Suppose that $q_m^\star\leq q_{\max}$. For $q=1,\ldots,q_{\max}$, define
\[
DR_{q,m}(X)
=
\frac{\hat\delta_{q,m}(X)}
{\hat\delta_{q+1,m}(X)}.
\]
If a denominator is zero, set the corresponding ratio equal to $+\infty$. Using the smallest maximizer as a deterministic tie-breaking rule, define
\[
\hat q(m)
=
\min
\operatorname*{arg\,max}_{1\leq q\leq q_{\max}}
DR_{q,m}(X).
\]
Then
\[
\mathbb P\left\{
\hat q(m)=q_m^\star
\right\}
\to1
\]
as $N,T\to\infty$. Consequently, if $m\geq m_0$, then
\[
\mathbb P\left\{
\hat q(m)=q_0
\right\}
\to1,
\]
whereas if $1\leq m<m_0$, then
\[
\mathbb P\left\{
\hat q(m)
=
\left\lceil\frac{q_0m_0}{m}\right\rceil
\right\}
\to1.
\]
\end{theorem}

For a fixed $m$, the quantity $q_m^\star$ is the smallest factor dimension for which $(q,m)$ belongs to $S_1$. Hence, at $q=q_m^\star$, the numerator of $DR_{q,m}(X)$ corresponds to the inadmissible structure $(q_m^\star-1,m)$ and is of order $\sqrt{NT}$, whereas the denominator corresponds to the admissible structure $(q_m^\star,m)$ and is of order $\sqrt{N+T}$. The ratio therefore diverges at $q_m^\star$ and remains stochastically bounded at every other candidate value.

When $m\geq m_0$, the admissibility boundary satisfies $q_m^\star=q_0$, so the estimator consistently identifies the true number of dynamic factors. When $m<m_0$, it selects the minimum number of factor coordinates required to represent the true common component with candidate filter length $m$, as characterized by Proposition \ref{alt_spec}.

When $m=1$,
\[
DR_{q,1}(X)
=
\frac{\sigma_q(X)}{\sigma_{q+1}(X)}.
\]
The eigenvalue-ratio statistic of \citet{Ahn-Horenstein-2013} is based on
\[
\frac{\sigma_q^2(X)}{\sigma_{q+1}^2(X)}
=
DR_{q,1}^2(X).
\]
The two criteria therefore have the same maximizer. In this case, Theorem \ref{thm:q_consistency} consistently estimates the true static dimension
\[
r_0=q_0m_0.
\]
The theorem extends this static ratio criterion by allowing the candidate filter length to exceed one and by imposing the lag-stacking restrictions of a dynamic factor model.

We next fix the factor dimension and estimate the minimum admissible filter length.

\begin{theorem}[Consistency of the filter-length estimator]
\label{thm:m_consistency}
Suppose that $q_0,m_0\geq1$ and that Assumptions \ref{pd}, \ref{error_norm}, and \ref{lower_bound_E} hold. Let $q_{\max}$, $m_{\max}$, and $k_{\max}$ be fixed positive integers satisfying
\[
q_0\leq q_{\max},
\qquad
m_0\leq m_{\max},
\qquad
q_0m_0+q_{\max}m_{\max}\leq k_{\max}.
\]
Fix $q\in\{q_0,\ldots,q_{\max}\}$ and define
\[
m_q^\star
=
\left\lceil\frac{q_0m_0}{q}\right\rceil.
\]
For $m=1,\ldots,m_{\max}$, define
\[
MR_{q,m}(X)
=
\frac{\hat\delta_{q+1,m-1}(X)}
{\hat\delta_{q+1,m}(X)}.
\]
If a denominator is zero, set the corresponding ratio equal to $+\infty$. Define
\[
\hat m(q)
=
\min
\operatorname*{arg\,max}_{1\leq m\leq m_{\max}}
MR_{q,m}(X).
\]
Then
\[
\mathbb P\left\{
\hat m(q)=m_q^\star
\right\}
\to1
\]
as $N,T\to\infty$. Consequently, if $q=q_0$, then
\[
\mathbb P\left\{
\hat m(q_0)=m_0
\right\}
\to1,
\]
whereas if $q>q_0$, then
\[
\mathbb P\left\{
\hat m(q)
=
\left\lceil\frac{q_0m_0}{q}\right\rceil
\right\}
\to1.
\]
\end{theorem}

For a fixed $q\geq q_0$, the quantity $m_q^\star$ is the smallest filter length for which $(q,m)$ belongs to $S_1$. At $m=m_q^\star$, the numerator of $MR_{q,m}(X)$ corresponds to the inadmissible structure $(q,m_q^\star-1)$, while the denominator corresponds to the admissible structure $(q,m_q^\star)$. The ratio therefore diverges only at $m_q^\star$.

Theorem \ref{thm:m_consistency} consistently estimates the admissibility boundary $m_q^\star$ for a fixed factor dimension $q\geq q_0$. In particular, it consistently estimates the true filter length $m_0$ when $q=q_0$. When $q>q_0$, it selects the smallest candidate filter length required to represent the true common component using $q$ factor coordinates.

To recover $(q_0,m_0)$ using Theorem \ref{thm:q_consistency} and \ref{thm:m_consistency}, we may start with any $m$ sufficiently large so that $m\geq m_0$. Using $m\geq m_0$, we may recover $q_0$ using Theorem \ref{thm:q_consistency}. Then, using Theorem \ref{thm:m_consistency}, we may recover $m_0$ using $q=q_0$.

\subsection{Selection Based on Information Criteria}
Let $V(q,m)$ denote the mean squared errors when the structure $(q,m)$ is used to estimate a dynamic factor model, i.e.,
\begin{align*}
V(q,m) 
= \frac{1}{NT}
\fnorm{X - \hat C_{q,m}}^2,
\end{align*}
where $\hat C_{q,m}=\sum_{k=0}^{m-1} \hat F_{-k}^{(q,m)} \hat\lambda_k^{(q,m)\prime}$ is the least squares estimator obtained under the structure $(q,m)$.

Following \citet{Bai02}, we augment the residual variance with a penalty term. The key is to construct a penalty sequence that vanishes asymptotically, yet dominates the difference in the residual sum of squares between the true model and any overparameterized alternative. In this spirit, we define the following two information criteria:
\begin{align*}
PC(q,m) &= V(q,m) + (r+q) g(N,T),\\
DC(q,m) &= \frac{1}{NT} \hat\delta_{q+1,m}^2(X) + (r+q) g(N,T),
\end{align*}
where $r=qm$ and $g(N,T)$ is a vanishing penalty function. Compared to \citet{Bai02}, the main differences are: (i) model fit is evaluated after estimating the dynamic factor model, not its static representation, and (ii) the penalty accounts for both the number of static factors $r=qm$, and the number of dynamic factors $q$. 

\begin{theorem}\label{test_IC}
Suppose Assumptions \ref{pd}, \ref{error_norm}, and \ref{lower_bound_common} hold.  Let $q_{max}\geq q_0$ be a fixed integer and let $m_{max}\geq m_0$ be given as in Assumption \ref{lower_bound_common}. Let $\hat q$ and $\hat m$ be given by
\[
(\hat q,\hat m) = \argmin_{0\leq q\leq q_{max},0\leq m\leq m_{max}}PC(q,m)
\]
or
\[
(\hat q,\hat m) = \argmin_{0\leq q\leq q_{max},0\leq m\leq m_{max}}DC(q,m).
\]
If $g(N,T)\to0$ and $\left(\frac{NT}{N+T}\right)g(N,T)\to\infty$ as $N,T\to\infty$, we have
\[
\mathbb{P}\{(\hat q,\hat m)= (q_0,m_0)\} \to 1
\]
as $N,T\to\infty$.
\end{theorem}

The key to determining the true structure $(q_0,m_0)$ is penalizing both the number of static factors $r=qm$ and the number of dynamic factors $q$. Penalizing only the number of static factors would select the static representation of the dynamic factor model, since among the structures $(q,m)$ with the same static factors, the structure $(qm,1)$ has the smallest squared error. Including $q$ as a separate term in the penalty ensures recovery of the most parsimonious dynamic structure. In fact, we may use any penalty function of the form $\varphi(q,m)g(N,T)$ as long as $g(N,T)\to0$, $\left(\frac{NT}{N+T}\right)g(N,T)\to\infty$ as $N,T\to\infty$, and 
\begin{align*}
    \varphi(q_0,m_0) \leq \varphi(q,m) 
\end{align*}
for all $(q,m)\in S_1$ with strict inequality when $(q,m)\neq (q_0,m_0)$. For example, any function of the form $\varphi(q,m) = \alpha qm+\beta q$ with $\alpha,\beta>0$ can also be used to consistently recover $(q_0,m_0)$. 

Theorem \ref{test_IC} shows that Assumptions \ref{pd}, \ref{error_norm}, and \ref{lower_bound_common} suffice to asymptotically identify $(q,m)$. The result does not require a particular finite-dimensional parametric model, such as a finite-order VAR, for the factor process. Assumptions \ref{pd} and \ref{error_norm} are commonly used to identify the number of static factors, e.g., \cite{Chamberlain_Rothschild_1983}, \cite{Bai02}, \cite{Bai-Ng-2023}. To recover the number of dynamic factors, along with the corresponding filter length, the only additional assumption required is the Assumption \ref{lower_bound_common}. 

Note that when $q_0=0$ or $m_0=0$, so that there does not exist a common component with factor structure, the information criterion consistently recovers $(q_0,m_0)=(0,0)$.

\begin{assumption}[Nondegenerate average idiosyncratic variance]
\label{average_error_variance}
There exists a constant
\[
\sigma_\varepsilon^2\in(0,\infty)
\]
such that
\[
\frac{1}{NT}\|E\|_F^2
\to_p
\sigma_\varepsilon^2.
\]
\end{assumption}

\begin{corollary}\label{IC}
Suppose Assumptions \ref{pd}, \ref{error_norm}, \ref{lower_bound_common}, and \ref{average_error_variance} hold. Let $q_{max}$ be a fixed integer and let $m_{max}$ be given as in Assumption \ref{lower_bound_common}. Let
\[
(\hat q, \hat m) = \argmin_{0 \le q \le q_{\max}, 0 \le m \le m_{\max}} IC(q,m),
\]
where
\[
IC(q,m) = \log V(q,m) + (r+q) g(N,T).
\]
If $g(N,T) \to 0$ and $\left(\frac{NT}{N+T}\right)g(N,T)\to\infty$, then
\[
\mathbb{P}\{(\hat q,\hat m) = (q_0,m_0)\} \to 1.
\]
\end{corollary}
As in \citet{Bai02}, we use three candidate penalty terms:
\begin{align*}
    g_1(N,T) &= \left(\frac{N+T}{NT}\right)\log\left(\frac{NT}{N+T}\right)\\
    g_2(N,T) &= \left(\frac{N+T}{NT}\right)\log\left(\min(N,T)\right)\\
    g_3(N,T) &= \left(\frac{1}{\min(N,T)}\right)\log\left(\min(N,T)\right).
\end{align*}
We also scale the penalty term with $\hat\sigma_{PC}^2 = V(q_{max},m_{max})$ and $\hat d_{DC}^2 = \frac{1}{NT}\hat\delta_{1,0}^2(X)$ for PC and DC respectively. Under Assumptions \ref{error_norm} and \ref{average_error_variance}, we may show that 
\[
\hat\sigma^2_{PC}\to_p \sigma_\varepsilon^2
\]
as $N,T\to\infty$. Also, we have
\[
\left|\frac{\|X\|}{\sqrt{NT}}-\frac{\|C_0\|}{\sqrt{NT}}\right|\leq \frac{\|E\|}{\sqrt{NT}}
\]
so that 
\[
\hat d_{DC}\to_p\lambda_{max}(\Sigma_\gamma^{1/2}\Sigma_g\Sigma_\gamma^{1/2})
\]
as $N,T\to\infty$ when $q_0,m_0\geq1$. The DC criterion with the scale $\hat d_{DC}^2 = \frac{1}{NT}\hat\delta_{1,0}^2(X)$ does not depend on $(q_{max},m_{max})$ through the scale term. However, it may not recover the true structure under $(q_0,m_0)=(0,0)$. Under the null structure, this scale is only of order $\frac{N+T}{NT}$, making the effective penalty too small to dominate the noise-scale reduction from fitting spurious factors.

Both $PC$ and $IC$ can be viewed as dynamic extensions of the original criteria in \citet{Bai02}: if $m=1$, they reduce numerically to the criteria in \cite{Bai02}, up to a multiplicative factor of two in the penalty term.

\subsection{Determining the Lag Order of the VAR Model}\label{sec:var}
Although our method does not require the dynamic factors \((f_t)\) to follow a VAR process, it may still be useful in practice to model their dynamics using a VAR model. For instance, one may be interested in estimating the dynamic responses of macroeconomic variables to structural shocks identified using residuals of a VAR model for \((f_t)\), as in \cite{Forni2009}.

In this context, modeling \((f_t)\) directly using a VAR is substantially more parsimonious than working with the static representation \((g_t)\). As shown by \citet{Bai2007} and \citet{Forni2009}, if the dynamic factors \((f_t)\) follow a VAR\((p)\) process, then the associated static factors \((g_t)\) follow a VAR\((\tilde p)\) process, where \(\tilde p = 1\) if \(m \ge p\) and \(\tilde p = p - m + 1\) if \(m < p\).

When \(m \ge p\), estimation based on the static representation requires fitting a VAR\((1)\) in dimension \(qm\), involving \(q^2 m^2\) parameters. In contrast, direct estimation of \((f_t)\) requires fitting a VAR\((p)\) in dimension \(q\), involving only \(q^2p\) parameters. Since \(m \ge p\), and \(m > 1\) in genuine dynamic factor models, the reduction in dimensionality can be substantial. When \(m < p\), the estimation of the static representation requires \((p - m + 1) q^2 m^2\) parameters, which again strictly exceed \(p q^2\) whenever \(m > 1\). This highlights substantial reductions in dimension and parameter count by directly modeling the dynamic factors.

In this section, we assume $(f_t)$ follows a VAR$(p)$ process, and show that the lag order $p$ can be selected consistently using the estimated factors.

\begin{assumption}[Nondegeneracy of the augmented factor stack]
\label{pd_extra}
Let
\[
J_0=
\begin{pmatrix}
I_{T-1}&0_{(T-1)\times1}
\end{pmatrix},
\qquad
J_1=
\begin{pmatrix}
0_{(T-1)\times1}&I_{T-1}
\end{pmatrix},
\]
so that $J_0$ and $J_1$ select the first and last $T-1$ rows, respectively. Define
\[
\overline G^{(m+1)}
=
J_1
\begin{pmatrix}
F_0&F_{-1}&\cdots&F_{-m}
\end{pmatrix}
\in
\mathbb R^{(T-1)\times q(m+1)}.
\]
Then
\[
\frac{1}{T-1}
{\overline G^{(m+1)}}'
\overline G^{(m+1)}
\to_p
\Sigma_g^{(m+1)}
\succ0.
\]
\end{assumption}

\begin{assumption}[VAR lag-order regularity]
\label{var_order_regularity}
The dynamic factor process follows a stable VAR$(p_0)$ model,
\[
f_t
=
\sum_{j=1}^{p_0}A_jf_{t-j}
+
u_t,
\qquad
0\leq p_0\leq p_{\max},
\]
where $q$, $p_0$, and $p_{\max}$ are fixed. The order $p_0$ is minimal, and the innovation sequence satisfies
\[
\mathbb E(u_t\mid\mathcal F_{t-1})=0,
\qquad
\mathbb E(u_tu_t')=\Sigma_u\succ0,
\]
together with uniformly bounded fourth moments and the weak-dependence conditions needed for the usual fixed-dimensional VAR laws of large numbers and central limit theorems.

For every $\ell=1,\ldots,p_{\max}$, the population covariance matrix of
\[
z_{\ell,t}
=
\begin{pmatrix}
f_{t-1}'&\cdots&f_{t-\ell}'
\end{pmatrix}'
\]
is positive definite. Moreover, for every $\ell=p_0+1,\ldots,p_{\max}$, the population covariance matrix of the additional lag regressors after projecting them on $f_{t-1},\ldots,f_{t-p_0}$ is positive definite.
\end{assumption}

Assumption \ref{pd_extra} rules out an exact linear relation among $m+1$ consecutive factor vectors. It is used to identify a common $q\times q$ rotation across the $m$ lag blocks of the static factor representation.

Proposition \ref{consistency_1} only gives consistency of static representations $g_t$ of dynamic factors $f_t$, i.e., consistency up to a $qm\times qm$ dimensional invertible matrix. Assumption \ref{pd_extra} is used to identify and consistently estimate the dynamic factors $f_t$ up to a $q\times q$ dimensional matrix. \cite{bai_wang_2013} uses a similar assumption for identification. After establishing consistency of the dynamic factors $f_t$, we use the result to establish consistency of the VAR lag order selection in the following Proposition \ref{test_var}.

\begin{proposition}[Consistency of the dynamic factors and lag loadings]
\label{consistency_2}
Let $(q,m)=(q_0,m_0)$ be the true structure, with $q$ and $m$ fixed. Let
\[
\hat G
=
\begin{pmatrix}
\hat F_0&
\hat F_{-1}&
\cdots&
\hat F_{-(m-1)}
\end{pmatrix},
\qquad
\hat\Gamma
=
\begin{pmatrix}
\hat\lambda_0&
\hat\lambda_1&
\cdots&
\hat\lambda_{m-1}
\end{pmatrix}
\]
be constructed from a global minimizer of \eqref{obj_fn}. Choose a normalized representative satisfying
\[
\frac{1}{N}
\sum_{k=0}^{m-1}
\hat\lambda_k'
\hat\lambda_k
=
I_q, \quad \frac{1}{N}
\sum_{k=0}^{m-1}
\lambda_k'\lambda_k
=
I_q
\]
Suppose that Assumptions \ref{pd}, \ref{error_norm}, and \ref{pd_extra} hold and
\[
\lambda_{\max}
\left(
\frac{1}{T}
\hat G'\hat G
\right)
=
O_p(1).
\]
Define
\[
H_f'
=
\left(
\sum_{k=0}^{m-1}
\lambda_k'\hat\lambda_k
\right)
\left(
\sum_{k=0}^{m-1}
\hat\lambda_k'\hat\lambda_k
\right)^{-1}.
\]
Then $H_f$ is nonsingular with probability approaching one,
\[
\|H_f\|+\|H_f^{-1}\|=O_p(1),
\]
and
\[
\frac{1}{\sqrt T}
\left\|
\hat F-FH_f'
\right\|_F
=
O_p\left(
\frac{1}{\sqrt N}+\frac{1}{\sqrt T}
\right).
\]
Moreover, for each $k=0,\ldots,m-1$,
\[
\frac{1}{\sqrt N}
\left\|
\hat\lambda_k-\lambda_kH_f^{-1}
\right\|_F
=
O_p\left(
\frac{1}{\sqrt N}+\frac{1}{\sqrt T}
\right).
\]
\end{proposition}
Suppose we use the common effective sample
\[
t=p_{\max}+1,\ldots,T
\]
for all candidate lag orders in the VAR regression. Define
\[
Y
=
\begin{pmatrix}
f_{p_{\max}+1}'\\
\vdots\\
f_T'
\end{pmatrix}
\in\mathbb R^{n\times q},
\qquad
\hat Y
=
\begin{pmatrix}
\hat f_{p_{\max}+1}'\\
\vdots\\
\hat f_T'
\end{pmatrix}.
\]
For $\ell\geq1$, define
\[
Z_\ell
=
\begin{pmatrix}
F_{-1}&\cdots&F_{-\ell}
\end{pmatrix}
\in\mathbb R^{n\times q\ell},
\qquad
\hat Z_\ell
=
\begin{pmatrix}
\hat F_{-1}&\cdots&\hat F_{-\ell}
\end{pmatrix},
\]
where all matrices are restricted to the common effective sample. For $\ell=0$, the regressor matrix is empty. For each $\ell=0,\ldots,p_{\max}$, define the minimized trace criterion
\[
\hat v_\ell
=
\frac{1}{T-p_{\max}}
\min_{B\in\mathbb R^{q\times q\ell}}
\left\|
\hat Y-\hat Z_\ell B'
\right\|_F^2,
\]
with
\[
\hat v_0
=
\frac{1}{T-p_{\max}}
\|\hat Y\|_F^2.
\]
Let $g(N,T)>0$ be a deterministic penalty sequence satisfying
\begin{equation}
\label{eq:VAR_penalty_conditions}
g(N,T)\to0,
\qquad
\frac{NT}{N+T}g(N,T)\to\infty.
\end{equation}
For a fixed constant $\kappa>0$, define
\[
\operatorname{PIC}_{\mathrm{VAR}}(\ell)
=
\hat v_\ell
+
\kappa q^2\ell\,g(N,T),
\qquad
\ell=0,\ldots,p_{\max}.
\]
\begin{theorem}[Consistency of trace-based VAR lag-order selection]
\label{test_var}
Suppose that the true dynamic structure $(q_0,m_0)$ is used to estimate the factors and that Assumptions \ref{pd}, \ref{error_norm}, \ref{pd_extra}, and \ref{var_order_regularity} hold. Let $g(N,T)$ satisfy 
\begin{equation*}
g(N,T)\to0,
\qquad
\frac{NT}{N+T}g(N,T)\to\infty
\end{equation*}
and define
\[
\hat p
=
\min
\operatorname*{arg\,min}_{0\leq\ell\leq p_{\max}}
\operatorname{PIC}_{\mathrm{VAR}}(\ell).
\]
Then
\[
\mathbb P\{\hat p=p_0\}\to1
\]
as $N,T\to\infty$. No restriction on the relative rates at which $N$ and $T$ diverge is required.
\end{theorem}
\begin{remark}
Existing structural-factor and FAVAR applications generally select the VAR lag order using conventional AIC or BIC applied to estimated factors. \cite{Forni-Gambetti-2010}, for example, choose two lags as the midpoint between the orders selected by AIC and BIC applied to estimated factors.

For the conventional time-series BIC penalty $g_{\mathrm{BIC}}(N,T)=\frac{\log T}{T}$ to yield consistent model selection, we need an additional rate condition
\[
\frac{T}{N\log T}\to0
\]
to hold so that the squared factor estimation error $O_p\left(\frac{1}{N}+\frac{1}{T}\right)$ is negligible compared to the penalty term.

The panel-adaptive penalties in Theorem \ref{test_var} do not require this relative-growth condition.
\end{remark}

\section{Simulations}\label{sec:simulations}
We conduct Monte Carlo simulations of dynamic factor models to evaluate the finite-sample performance of estimators for the number of dynamic factors and the filter length. Specifically, we simulate the model
\[
x_{t} = \sum_{k=0}^{m-1} \lambda_k f_{t-k} + \varepsilon_t,
\]
where the idiosyncratic errors satisfy
\[
\varepsilon_{it} = \sqrt{\frac{\theta(1-\rho^2)}{1+2J\beta^2}}\, e_{it}, \qquad
e_{it} = \rho e_{i,t-1} + v_{it} + \sum_{1 \le |j| \le J} \beta v_{i-j,t}.
\]
The innovations \(v_{it}\) and the factor loadings \(\lambda_k\) are drawn independently from a standard normal distribution. This data-generating process is similar to those considered in \cite{Bai02} and \cite{Ahn-Horenstein-2013}, except that we allow the filter length \(m\) to exceed one. The common factors \(f_t\) follow a VARMA\((1,1)\) process,
\[
f_t = A f_{t-1} + u_t + \Theta u_{t-1},
\]
where $u_t$ follows independent standard normal under the following four specifications:
\begin{enumerate}
    \item \(\rho = \beta = J = 0\) and \(A = \Theta = 0\).
    \item \(\rho = 0.3\), \(\beta = 0.1\), \(J = 10\), and \(A = \Theta = 0\).
    \item \(\rho = \beta = J = 0\), \(A = \mathrm{Diag}(0.7, 0.5, 0.3)\), and \(\Theta = 0\).
    \item \(\rho = \beta = J = 0\), \(A = 0\), and \(\Theta = \mathrm{Diag}(0.7, 0.5, 0.3)\).
\end{enumerate}

By construction, the variance of \(\varepsilon_{it}\) equals \(\theta\). The variance of the common component \(\sum_{k=0}^{m-1} \lambda_k f_{t-k}\) is given by \(m \cdot \mathrm{tr}(\Sigma_f)\), where \(\Sigma_f = \mathbb{E}[f_t f_t']\). We therefore set \(\theta = m \cdot \mathrm{tr}(\Sigma_f)\), ensuring that the expected signal-to-noise ratio equals one in all specifications.

The first specification corresponds to an exact dynamic factor model, in which both the dynamic factors and idiosyncratic errors are serially uncorrelated. The second specification represents an approximate dynamic factor model, allowing for serial and cross-sectional dependence in the idiosyncratic errors. The third and fourth specifications examine settings in which the dynamic factors follow VAR\((1)\) and VMA\((1)\) processes, respectively.

\subsection{Selection Based on the Dynamic Singular Value Ratio}
We use the four specifications described in the previous section with $(q_0,m_0)=(3,3)$. 

\begin{table}[H]
\centering
\small
\setlength{\tabcolsep}{2pt} 
\renewcommand{\arraystretch}{0.9}

\begin{tabular}{|ll|ccc|ccc|ccc|ccc|}
\hline
\hline
& & \multicolumn{3}{c|}{DGP1} & \multicolumn{3}{c|}{DGP2} & \multicolumn{3}{c|}{DGP3} & \multicolumn{3}{c|}{DGP4}\\ 
\hline
 N&T & $m=2$ & $m=3$ & $m=4$ & $m=2$ & $m=3$ & $m=4$& $m=2$ & $m=3$ & $m=4$& $m=2$ & $m=3$ & $m=4$\\ 
\hline
$100$ & $100$ & 0.92 & 0.99 & 0.96 & 0.17 & 0.28 & 0.23 & 0.00 & 0.79 & 0.47 & 0.12 & 0.93 & 0.81  \\
$100$ & $200$ & 0.99 & 0.98 & 1.00 & 0.17 & 0.40 & 0.24 & 0.04 & 0.98 & 0.86 & 0.36 & 1.00 & 1.00  \\
$100$ & $300$ & 1.00 & 1.00 & 1.00 & 0.20 & 0.55 & 0.29 & 0.09 & 1.00 & 0.98 & 0.58 & 1.00 & 1.00  \\
$200$ & $100$ & 1.00 & 1.00 & 1.00 & 0.52 & 0.88 & 0.61 & 0.10 & 0.88 & 0.66 & 0.43 & 1.00 & 0.88  \\
$200$ & $200$ & 1.00 & 1.00 & 1.00 & 0.89 & 0.99 & 0.93 & 0.43 & 1.00 & 0.99 & 0.90 & 1.00 & 1.00  \\
$200$ & $300$ & 1.00 & 1.00 & 1.00 & 0.93 & 1.00 & 0.95 & 0.82 & 1.00 & 1.00 & 0.98 & 1.00 & 1.00  \\
$300$ & $100$ & 1.00 & 0.99 & 1.00 & 0.85 & 0.99 & 0.86 & 0.18 & 1.00 & 0.89 & 0.72 & 1.00 & 0.98  \\
$300$ & $200$ & 1.00 & 1.00 & 1.00 & 0.99 & 0.99 & 1.00 & 0.81 & 1.00 & 1.00 & 0.99 & 1.00 & 1.00  \\
$300$ & $300$ & 1.00 & 1.00 & 1.00 & 1.00 & 1.00 & 1.00 & 0.93 & 1.00 & 1.00 & 1.00 & 1.00 & 1.00  \\
\hline
\end{tabular}
\caption{Monte Carlo frequency with which $\hat q(m)=q_m^\star$, for each fixed candidate filter length $m$.}
\label{tab:DR_q}
\end{table}

Table \ref{tab:DR_q} reports the Monte Carlo frequency with which $\hat q(m)=q_m^\star$. For $m=2$, we have $q_m^\star=\lceil\frac{q_0m_0}{m}\rceil = 5$. For $m=3,4$, we have $q_m^\star=q_0=3$. The estimator performs very well in DGP1. DGP3 and DGP4 also perform strongly when $m\geq m_0$ as the panel dimensions increase. DGP2 is more challenging, particularly when $N$ is small, reflecting the serial and cross-sectional dependence in the idiosyncratic component. The short-filter case $m=2$ is also difficult under persistent factor dynamics.

\begin{table}[H]
\centering
\small
\setlength{\tabcolsep}{3pt} 
\renewcommand{\arraystretch}{0.9}
\begin{tabular}{|ll|ccc|ccc|ccc|ccc|}
\hline
\hline
& & \multicolumn{3}{c|}{DGP1} & \multicolumn{3}{c|}{DGP2} & \multicolumn{3}{c|}{DGP3} & \multicolumn{3}{c|}{DGP4}\\ 
\hline
 N&T & $q=3$ & $q=4$ & $q=5$ & $q=3$ & $q=4$ & $q=5$ & $q=3$ & $q=4$ & $q=5$ & $q=3$ & $q=4$ & $q=5$\\ 
\hline
$100$ & $100$ & 0.99 & 0.54 & 0.97 & 0.16 & 0.05 & 0.30 & 0.02 & 0.00 & 0.02 & 0.49 & 0.00 & 0.41  \\
$100$ & $200$ & 0.98 & 0.94 & 1.00 & 0.17 & 0.07 & 0.28 & 0.10 & 0.00 & 0.04 & 0.88 & 0.01 & 0.76  \\
$100$ & $300$ & 1.00 & 0.99 & 1.00 & 0.29 & 0.06 & 0.38 & 0.26 & 0.00 & 0.07 & 0.97 & 0.14 & 0.95  \\
$200$ & $100$ & 1.00 & 0.98 & 1.00 & 0.73 & 0.20 & 0.81 & 0.21 & 0.00 & 0.13 & 0.91 & 0.05 & 0.79  \\
$200$ & $200$ & 1.00 & 1.00 & 1.00 & 1.00 & 0.77 & 0.99 & 0.49 & 0.00 & 0.26 & 1.00 & 0.39 & 0.98  \\
$200$ & $300$ & 1.00 & 1.00 & 1.00 & 1.00 & 0.89 & 1.00 & 0.76 & 0.00 & 0.56 & 1.00 & 0.79 & 1.00  \\
$300$ & $100$ & 0.98 & 1.00 & 1.00 & 0.97 & 0.69 & 0.95 & 0.29 & 0.00 & 0.15 & 0.97 & 0.15 & 0.94  \\
$300$ & $200$ & 1.00 & 1.00 & 1.00 & 0.99 & 0.98 & 1.00 & 0.83 & 0.02 & 0.61 & 1.00 & 0.70 & 1.00  \\
$300$ & $300$ & 1.00 & 1.00 & 1.00 & 1.00 & 1.00 & 1.00 & 0.95 & 0.12 & 0.90 & 1.00 & 0.91 & 1.00  \\
\hline
\end{tabular}
\caption{Monte Carlo frequency with which $\hat m(q)=m_q^\star$, for each fixed factor dimension $q$.}
\label{tab:DR_m}
\end{table}
In Table \ref{tab:DR_m}, we use DR test to estimate the filter length $m$ given $q$. For $q=3,4$, we have $m_q^\star=\lceil\frac{q_0m_0}{q}\rceil = 3$, whereas for $q=5$, we have $m_q^\star=\lceil\frac{q_0m_0}{q}\rceil = 2$.


The particularly difficult case is $q=4$ in DGP3. The candidate structure $(4,2)$ has static dimension eight, only one below the true static rank $q_0m_0=9$. When the factors follow a VAR($1$) process, factor persistence can make the final required static direction relatively weak. Consequently, the dynamic approximation error under $(4,2)$ can be small in finite samples, limiting the separation between $m=2$ and $m=3$.

\subsection{Selection Based on Information Criteria}
We see finite sample performance of the PC, DC, IC criteria in selecting $(q,m)$. We compare the results with the two existing methods given by \cite{Bai2007} and \cite{Amengual-Watson-2007}\footnote{In both \cite{Bai2007} and \cite{Amengual-Watson-2007}, we used IC test whenever it was needed to test for the number of static factors. When the lag order of the VAR($p$) model needed to be specified, we used $p=1$ for the first three specifications and $p=2$ for the fourth specification. For \cite{Bai2007}, we used the default parameters $(\delta,m)=(0.1,1)$ with covariance matrix method. For \cite{Amengual-Watson-2007}, we used the residuals from direct regression, i.e., $\hat Y_t^B$ in \cite{Amengual-Watson-2007}.}. 

In Table \ref{tab:dgp12} and \ref{tab:dgp34}, PC, DC, and IC denote the information criterion methods proposed in this paper. BN and AW denote methods proposed in \cite{Bai2007} and \cite{Amengual-Watson-2007}, respectively. To illustrate the attainable performance of each procedure, the tables report the best-performing penalty variant among $g_1,g_2,g_3$ for each method. In particular, PC and AW use $g_1(N,T) = (\frac{N+T}{NT})\log(\frac{NT}{N+T})$, $BN$ uses $g_2(N,T) = (\frac{N+T}{NT})\log(\min(N,T))$, and $DC, IC$ use $g_3(N,T) = (\frac{1}{\min(N,T)})\log(\min(N,T))$. Results for selecting filter length are missing in  BN and AW, as they only select the number of dynamic factors.
{
\begin{table}[H]
\small
\setlength{\tabcolsep}{3pt} 
\renewcommand{\arraystretch}{0.9}
\begin{tabular}{|ll|cc|cc|cc|c|c|cc|cc|cc|c|c|}
\hline
\hline
& & \multicolumn{8}{c|}{$DGP1$}& \multicolumn{8}{c|}{$DGP2$}\\ 
\hline
& & \multicolumn{2}{c|}{$PC$} & \multicolumn{2}{c|}{$DC$} & \multicolumn{2}{c|}{$IC$} & $BN$ & $AW$&\multicolumn{2}{c|}{$PC$} & \multicolumn{2}{c|}{$DC$} & \multicolumn{2}{c|}{$IC$} & $BN$ & $AW$\\ 
\hline
N&T & $q$ & $m$ & $q$ & $m$& $q$ & $m$ & $q$ & $q$ & $q$ & $m$ & $q$ & $m$& $q$ & $m$ & $q$ & $q$ \\
\hline
$100$ & $100$ & 0.97 & 0.99 & 0.98 & 0.98 & 0.97 & 0.98 & 0.34 & 0.23 & 0.47 & 0.50 & 0.15 & 0.18 & 0.26 & 0.26 & 0.04 & 0.45  \\
$100$ & $200$ & 0.99 & 0.99 & 0.98 & 0.99 & 0.98 & 0.99 & 0.55 & 0.89 & 0.91 & 0.91 & 0.18 & 0.18 & 0.93 & 0.93 & 0.00 & 0.01  \\
$100$ & $300$ & 1.00 & 1.00 & 1.00 & 1.00 & 1.00 & 1.00 & 0.89 & 0.99 & 0.98 & 0.98 & 0.20 & 0.22 & 0.99 & 0.99 & 0.00 & 0.00  \\
$200$ & $100$ & 1.00 & 1.00 & 1.00 & 1.00 & 1.00 & 1.00 & 0.82 & 0.98 & 0.97 & 0.99 & 0.57 & 0.66 & 0.98 & 0.99 & 0.45 & 0.94  \\
$200$ & $200$ & 1.00 & 1.00 & 1.00 & 1.00 & 1.00 & 1.00 & 1.00 & 1.00 & 1.00 & 1.00 & 0.98 & 0.98 & 0.75 & 0.75 & 0.33 & 0.99  \\
$200$ & $300$ & 1.00 & 1.00 & 1.00 & 1.00 & 1.00 & 1.00 & 1.00 & 1.00 & 1.00 & 1.00 & 1.00 & 1.00 & 1.00 & 1.00 & 0.01 & 0.68  \\
$300$ & $100$ & 0.99 & 1.00 & 0.98 & 0.98 & 0.99 & 1.00 & 0.97 & 0.98 & 1.00 & 0.99 & 0.91 & 0.94 & 1.00 & 1.00 & 0.98 & 1.00  \\
$300$ & $200$ & 1.00 & 1.00 & 1.00 & 1.00 & 1.00 & 1.00 & 1.00 & 1.00 & 0.99 & 0.99 & 1.00 & 1.00 & 0.99 & 0.99 & 0.93 & 1.00  \\
$300$ & $300$ & 1.00 & 1.00 & 1.00 & 1.00 & 1.00 & 1.00 & 1.00 & 1.00 & 1.00 & 1.00 & 0.99 & 0.99 & 1.00 & 1.00 & 0.79 & 1.00  \\
\hline
\end{tabular}
\caption{Marginal Monte Carlo selection frequencies under DGP1 and DGP2. For PC, DC, and IC, the columns labeled $q$ and $m$ report $\Pr(\hat q=q_0)$ and $\Pr(\hat m=m_0)$, respectively. BN and AW estimate only $q_0$.}
\label{tab:dgp12}
\end{table}}
Table \ref{tab:dgp12} reports marginal selection frequencies under DGP1 and DGP2. All three proposed selectors perform very well in DGP1, even in the smallest panel, and their factor-dimension selection frequencies are substantially higher than those of BN and AW in the smaller samples. DGP2 is more challenging because the idiosyncratic component is serially and locally cross-sectionally dependent. PC is the most robust criterion in this design and improves rapidly with the panel dimensions. DC and IC are more conservative in small panels but also become accurate as both $N$ and $T$ increase.

{
\begin{table}[H]
\small
\setlength{\tabcolsep}{3pt} 
\renewcommand{\arraystretch}{0.9}
\begin{tabular}{|ll|cc|cc|cc|c|c|cc|cc|cc|c|c|}
\hline
\hline
& & \multicolumn{8}{c|}{$DGP3$}& \multicolumn{8}{c|}{$DGP4$}\\
\hline
& & \multicolumn{2}{c|}{$PC$} & \multicolumn{2}{c|}{$DC$} & \multicolumn{2}{c|}{$IC$} & $BN$ & $AW$&\multicolumn{2}{c|}{$PC$} & \multicolumn{2}{c|}{$DC$} & \multicolumn{2}{c|}{$IC$} & $BN$ & $AW$\\ 
\hline
N&T & $q$ & $m$ & $q$ & $m$& $q$ & $m$ & $q$ & $q$ & $q$ & $m$ & $q$ & $m$& $q$ & $m$ & $q$ & $q$ \\
\hline
$100$ & $100$ & 0.90 & 0.34 & 0.31 & 0.08 & 0.91 & 0.49 & 0.51 & 0.01 & 0.94 & 0.82 & 0.72 & 0.55 & 0.99 & 0.92 & 0.62 & 0.02  \\
$100$ & $200$ & 0.96 & 0.18 & 0.39 & 0.11 & 0.92 & 0.21 & 0.10 & 0.58 & 0.98 & 0.66 & 0.94 & 0.86 & 0.97 & 0.71 & 0.67 & 0.80  \\
$100$ & $300$ & 0.99 & 0.07 & 0.51 & 0.22 & 0.93 & 0.02 & 0.08 & 0.88 & 1.00 & 0.70 & 1.00 & 0.98 & 0.99 & 0.63 & 0.84 & 0.99  \\
$200$ & $100$ & 0.96 & 0.82 & 0.41 & 0.29 & 0.82 & 0.61 & 0.60 & 0.44 & 1.00 & 0.97 & 0.87 & 0.82 & 0.97 & 0.94 & 0.96 & 0.68  \\
$200$ & $200$ & 1.00 & 0.97 & 0.99 & 0.98 & 1.00 & 1.00 & 0.32 & 0.98 & 1.00 & 1.00 & 1.00 & 1.00 & 1.00 & 1.00 & 0.99 & 0.99  \\
$200$ & $300$ & 1.00 & 0.99 & 1.00 & 1.00 & 1.00 & 1.00 & 0.54 & 1.00 & 1.00 & 1.00 & 1.00 & 1.00 & 1.00 & 1.00 & 1.00 & 1.00  \\
$300$ & $100$ & 0.99 & 0.98 & 0.34 & 0.18 & 0.86 & 0.70 & 0.86 & 0.88 & 1.00 & 1.00 & 0.95 & 0.93 & 0.98 & 0.99 & 1.00 & 0.94  \\
$300$ & $200$ & 1.00 & 1.00 & 1.00 & 1.00 & 1.00 & 1.00 & 0.87 & 0.99 & 1.00 & 1.00 & 1.00 & 1.00 & 1.00 & 1.00 & 1.00 & 1.00  \\
$300$ & $300$ & 1.00 & 1.00 & 1.00 & 1.00 & 1.00 & 1.00 & 0.82 & 1.00 & 1.00 & 1.00 & 1.00 & 1.00 & 1.00 & 1.00 & 1.00 & 1.00  \\
\hline
\end{tabular}
\caption{Marginal Monte Carlo selection frequencies under DGP3 and DGP4. For PC, DC, and IC, the columns labeled $q$ and $m$ report $\Pr(\hat q=q_0)$ and $\Pr(\hat m=m_0)$, respectively. BN and AW estimate only $q_0$.}
\label{tab:dgp34}
\end{table}}

Table \ref{tab:dgp34} reports results under serially dependent factor processes. For the common task of estimating $q_0$, PC has a higher or equal selection frequency than BN and AW in all reported designs. DC and IC also perform strongly in many moderately large panels, although DC does not uniformly dominate the benchmark procedures in unbalanced panels. Filter-length selection is generally more difficult than factor-dimension selection, particularly for PC under DGP3 when $N$ is small relative to $T$. Because BN and AW do not estimate $m_0$, comparisons with these methods concern only the selection of $q_0$.

\section{Empirical Analysis: Shocks in the United States}\label{sec:empirical}
In this section, we estimate the factor structure $(q,m)$ of dynamic factor models estimated using a large panel of U.S. macroeconomic variables. The Supplemental Materials~\ref{sec:apd-mp} additionally reports the estimated impulse responses of these variables to monetary policy shocks identified within the dynamic factor model framework.

The empirical analysis is based on the FRED--MD dataset of \cite{McCracken-Ng-2015}, covering the period from March 1959 to March 2025. The sample comprises $T = 793$ observations on $N = 126$ macroeconomic time series. All variables are transformed in accordance with the procedures recommended by \cite{McCracken-Ng-2015}.

To determine the factor structure, we implement the PC2, DC2, and IC2 criteria using a rolling window of fixed length equal to 10 years. The evaluation period begins in March 1969 and extends through March 2025. For the March of each year in this period, the criteria are computed using the information set consisting of the preceding 10 years of observations.

\begin{figure}[H]
    \centering
    \includegraphics[width=1\linewidth]{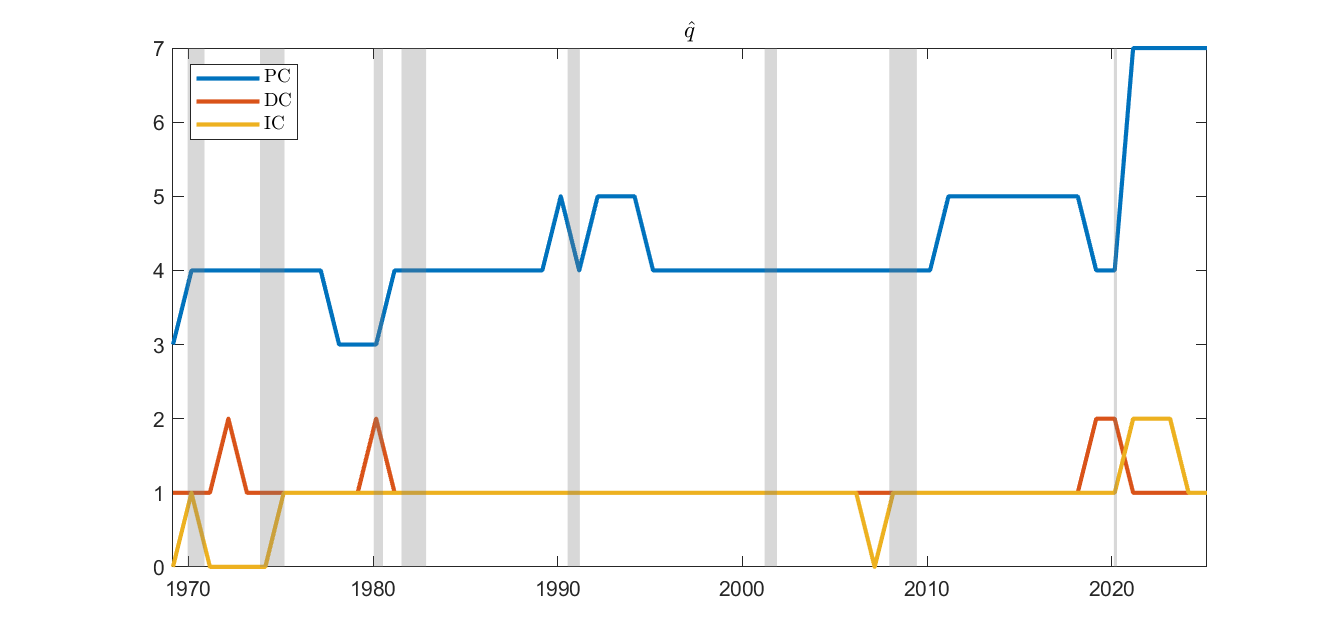}
    \caption{Estimated $q$ using PC2, DC2, and IC2.}
    \label{fig:qhat}
\end{figure}
The PC2 criterion tends to select a larger number of factors, whereas the DC2 and IC2 criteria typically yield a more parsimonious specification. Based on the PC2 results, we find evidence in favor of four dynamic factors prior to the COVID-19 period. In contrast, the post-COVID-19 period is characterized by an increase in the estimated number of factors, with the criterion indicating the presence of seven dynamic factors.

\begin{figure}[H]
    \centering
    \includegraphics[width=1\linewidth]{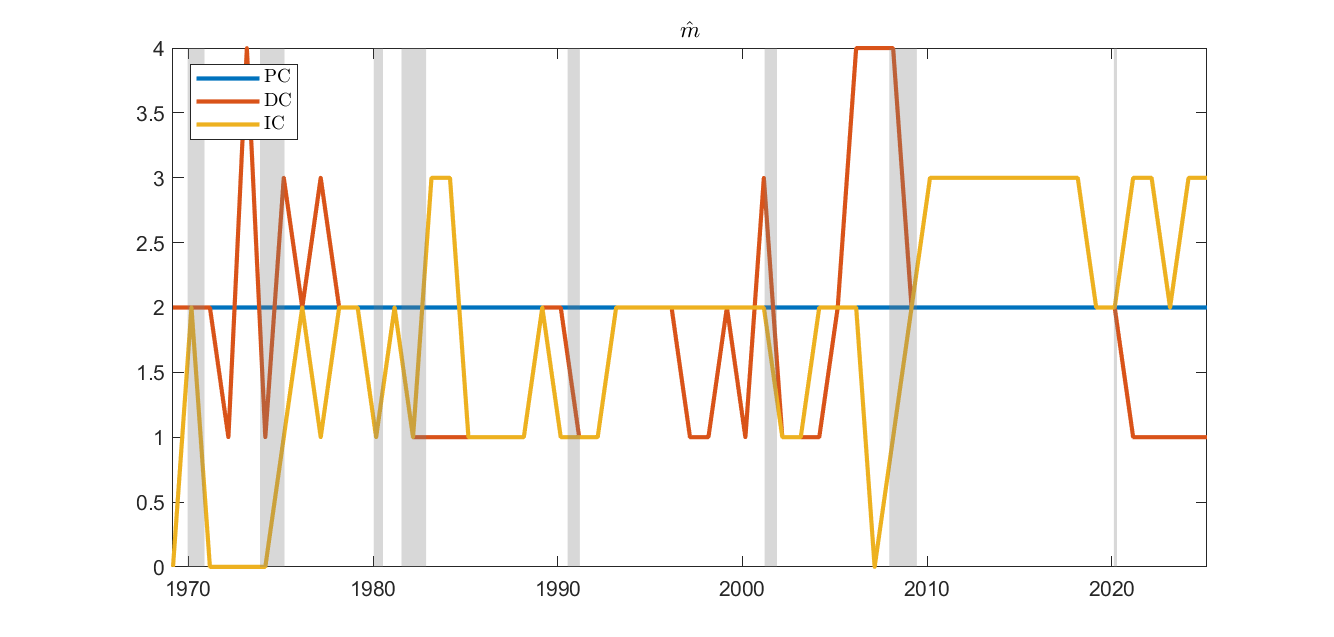}
    \caption{Estimated $m$ Using PC2, DC2, and IC2.}
    \label{fig:mhat}
\end{figure}
The PC2 criterion consistently selects $m = 2$ across the sample. The DC2 and IC2 criteria fluctuates around selecting $m = 1$ to $m=3$.

\begin{figure}[H]
    \centering
    \includegraphics[width=1\linewidth]{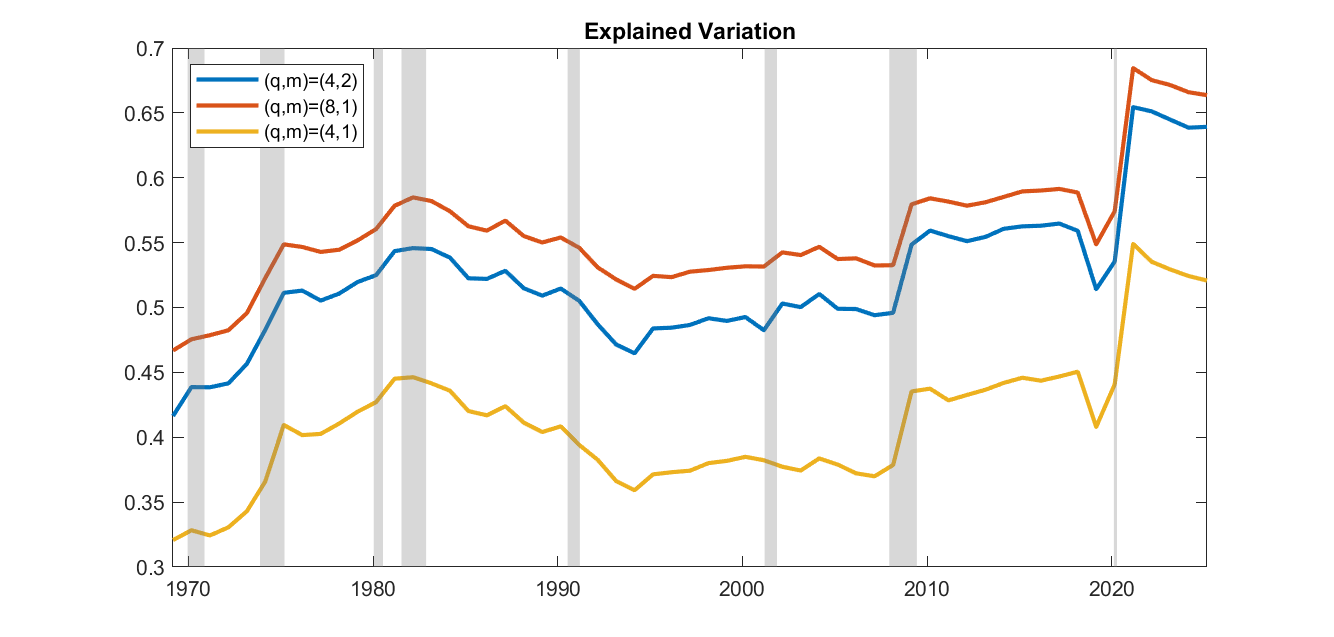}
    \caption{Explained variation of the common component.}
    \label{fig:EV}
\end{figure}
Based on the PC2 criterion, we find evidence supporting the specification $(q,m) = (4,2)$ for the majority of the sample period. The corresponding common component explains approximately 52\% of the total variation in average, with this share increasing to about 65\% in the post-COVID-19 period. A similar, but less drastic, increase in explained variation is also observed during the 2008 global financial crisis.

When adopting the static specification, $(q,m) = (8,1)$, the explained variation increases by approximately 3.6 percentage points. This gain can be interpreted as reflecting improved in-sample fit due to the greater flexibility of the model. 

In contrast, holding the number of factors fixed while reducing the filter length to $(q,m) = (4,1)$ leads to a substantial decline in explained variation, on the order of 11.1 percentage points. This deterioration can be interpreted as evidence of misspecification.

\section{Conclusion}\label{sec:conclusion}
This paper develops methods for determining the structure of dynamic factor models with finite filter length. As an intermediate step, we propose a computationally efficient alternating least squares (ALS) estimator that directly targets the dynamic factors and establish its consistency.

We introduce two procedures for jointly selecting the number of dynamic factors and filter length: a dynamic singular value ratio (DR) test extending \citet{Ahn-Horenstein-2013}, and an information criterion method similar to \cite{Bai02} that penalizes both the number of dynamic factors and its implied static dimensions. Both procedures are shown to be consistent.

Empirical results using the FRED-MD dataset indicate evidence of filter length $m=2$ with four dynamic factors prior to COVID-19 period, which increase to seven post-COVID-19.

\newpage
\bibliography{references.bib}

@article{Eichenbaum-Evans-1995,
    author = {Eichenbaum, Martin and Evans, Charles L.},
    title = {Some Empirical Evidence on the Effects of Shocks to Monetary Policy on Exchange Rates*},
    journal = {The Quarterly Journal of Economics},
    volume = {110},
    number = {4},
    pages = {975-1009},
    year = {1995},
    month = {11},
    issn = {0033-5533},
    doi = {10.2307/2946646},
    url = {https://doi.org/10.2307/2946646},
    eprint = {https://academic.oup.com/qje/article-pdf/110/4/975/5290735/110-4-975.pdf},
}

@article{Forni-Gambetti-2010,
title = {The dynamic effects of monetary policy: A structural factor model approach},
journal = {Journal of Monetary Economics},
volume = {57},
number = {2},
pages = {203-216},
year = {2010},
issn = {0304-3932},
doi = {https://doi.org/10.1016/j.jmoneco.2009.11.009},
url = {https://www.sciencedirect.com/science/article/pii/S0304393209001597},
author = {Mario Forni and Luca Gambetti}
}

@article{Bai-Ng-2023,
title = {Approximate factor models with weaker loadings},
journal = {Journal of Econometrics},
volume = {235},
number = {2},
pages = {1893-1916},
year = {2023},
issn = {0304-4076},
doi = {https://doi.org/10.1016/j.jeconom.2023.01.027},
url = {https://www.sciencedirect.com/science/article/pii/S030440762300060X},
author = {Jushan Bai and Serena Ng}
}

@article{Bai-Li-2016,
    author = {Bai, Jushan and Li, Kunpeng},
    title = {Maximum Likelihood Estimation and Inference for Approximate Factor Models of High Dimension},
    journal = {The Review of Economics and Statistics},
    volume = {98},
    number = {2},
    pages = {298-309},
    year = {2016},
    month = {05},
    issn = {0034-6535},
    doi = {10.1162/REST_a_00519},
    url = {https://doi.org/10.1162/REST_a_00519},
    eprint = {https://direct.mit.edu/rest/article-pdf/98/2/298/1918134/rest_a_00519.pdf},
}

@article{Bai-2003,
author = {Bai, Jushan},
title = {Inferential Theory for Factor Models of Large Dimensions},
journal = {Econometrica},
volume = {71},
number = {1},
pages = {135-171},
doi = {https://doi.org/10.1111/1468-0262.00392},
url = {https://onlinelibrary.wiley.com/doi/abs/10.1111/1468-0262.00392},
eprint = {https://onlinelibrary.wiley.com/doi/pdf/10.1111/1468-0262.00392},
year = {2003}
}

@article{Onatski-2010,
 ISSN = {00346535, 15309142},
 URL = {http://www.jstor.org/stable/40985808},
 author = {Alexei Onatski},
 journal = {The Review of Economics and Statistics},
 number = {4},
 pages = {1004--1016},
 publisher = {The MIT Press},
 title = {DETERMINING THE NUMBER OF FACTORS FROM EMPIRICAL DISTRIBUTION OF EIGENVALUES},
 urldate = {2025-12-17},
 volume = {92},
 year = {2010}
}

@article{Barigozzi-Hallin-Luciani-Zaffaroni-2023,
title = {Inferential theory for generalized dynamic factor models},
journal = {Journal of Econometrics},
volume = {239},
number = {2},
pages = {105422},
year = {2024},
issn = {0304-4076},
doi = {https://doi.org/10.1016/j.jeconom.2023.02.003},
url = {https://www.sciencedirect.com/science/article/pii/S0304407623000593},
author = {Matteo Barigozzi and Marc Hallin and Matteo Luciani and Paolo Zaffaroni}
}

@article{Forni-Hallin-Lippi-Zaffaroni-2017,
title = {Dynamic factor models with infinite-dimensional factor space: Asymptotic analysis},
journal = {Journal of Econometrics},
volume = {199},
number = {1},
pages = {74-92},
year = {2017},
issn = {0304-4076},
doi = {https://doi.org/10.1016/j.jeconom.2017.04.002},
url = {https://www.sciencedirect.com/science/article/pii/S0304407617300477},
author = {Mario Forni and Marc Hallin and Marco Lippi and Paolo Zaffaroni}
}

@TechReport{McCracken-Ng-2015,
type={Working Papers},
institution={Federal Reserve Bank of St. Louis},
author={Michael W. McCracken and Serena Ng},
title={FRED-MD: A Monthly Database for Macroeconomic Research},
year={2015},
month={Jun},
number={2015-12},
doi={10.20955/wp.2015.012},
url={https://ideas.repec.org/p/fip/fedlwp/2015-012.html},
}

@article{Ahn-Horenstein-2013,
author = {Ahn, Seung C. and Horenstein, Alex R.},
title = {Eigenvalue Ratio Test for the Number of Factors},
journal = {Econometrica},
volume = {81},
number = {3},
pages = {1203-1227},
doi = {https://doi.org/10.3982/ECTA8968},
url = {https://onlinelibrary.wiley.com/doi/abs/10.3982/ECTA8968},
eprint = {https://onlinelibrary.wiley.com/doi/pdf/10.3982/ECTA8968},
year = {2013}
}

@article{Chamberlain_Rothschild_1983,
 ISSN = {00129682, 14680262},
 URL = {http://www.jstor.org/stable/1912275},
 author = {Gary Chamberlain and Michael Rothschild},
 journal = {Econometrica},
 number = {5},
 pages = {1281--1304},
 publisher = {[Wiley, Econometric Society]},
 title = {Arbitrage, Factor Structure, and Mean-Variance Analysis on Large Asset Markets},
 urldate = {2025-10-17},
 volume = {51},
 year = {1983}
}

@article{Amengual-Watson-2007,
 ISSN = {07350015},
 URL = {http://www.jstor.org/stable/27638909},
 author = {Dante Amengual and Mark W. Watson},
 journal = {Journal of Business \& Economic Statistics},
 number = {1},
 pages = {91--96},
 publisher = {[American Statistical Association, Taylor & Francis, Ltd.]},
 title = {Consistent Estimation of the Number of Dynamic Factors in a Large N and T Panel},
 urldate = {2025-10-14},
 volume = {25},
 year = {2007}
}

@article{Forni2000,
    author = {Forni and Reichlin, Lucrezia and Hallin, Marc and Lippi, Marco},
    year = {2000},
    month = {02},
    pages = {540-554},
    title = {The Generalized Dynamic-Factor Model: Identification And Estimation},
    volume = {82},
    journal = {The Review of Economics and Statistics},
    doi = {10.1162/003465300559037}
}

@article{Stock02,
 ISSN = {01621459},
 URL = {http://www.jstor.org/stable/3085839},
 author = {James H. Stock and Mark W. Watson},
 journal = {Journal of the American Statistical Association},
 number = {460},
 pages = {1167--1179},
 publisher = {[American Statistical Association, Taylor & Francis, Ltd.]},
 title = {Forecasting Using Principal Components from a Large Number of Predictors},
 volume = {97},
 year = {2002}
}

@article{Bai02,
author = {Bai, Jushan and Ng, Serena},
title = {Determining the Number of Factors in Approximate Factor Models},
journal = {Econometrica},
volume = {70},
number = {1},
pages = {191-221},
doi = {https://doi.org/10.1111/1468-0262.00273},
url = {https://onlinelibrary.wiley.com/doi/abs/10.1111/1468-0262.00273},
eprint = {https://onlinelibrary.wiley.com/doi/pdf/10.1111/1468-0262.00273},
year = {2002}
}

@article{Forni2009,
 ISSN = {02664666, 14694360},
 URL = {http://www.jstor.org/stable/40388590},
 author = {Mario Forni and Domenico Giannone and Marco Lippi and Lucrezia Reichlin},
 journal = {Econometric Theory},
 number = {5},
 pages = {1319--1347},
 publisher = {Cambridge University Press},
 title = {Opening the Black Box: Structural Factor Models with Large Cross Sections},
 urldate = {2023-03-29},
 volume = {25},
 year = {2009}
}

@article{Hallin2007,
 ISSN = {01621459},
 URL = {http://www.jstor.org/stable/27639890},
 author = {Marc Hallin and Roman Liška},
 journal = {Journal of the American Statistical Association},
 number = {478},
 pages = {603--617},
 publisher = {[American Statistical Association, Taylor \& Francis, Ltd.]},
 title = {Determining the Number of Factors in the General Dynamic Factor Model},
 urldate = {2023-03-29},
 volume = {102},
 year = {2007}
}

@article{Bai2007,
 ISSN = {07350015},
 URL = {http://www.jstor.org/stable/27638906},
 author = {Jushan Bai and Serena Ng},
 journal = {Journal of Business \& Economic Statistics},
 number = {1},
 pages = {52--60},
 publisher = {[American Statistical Association, Taylor \& Francis, Ltd.]},
 title = {Determining the Number of Primitive Shocks in Factor Models},
 urldate = {2023-03-29},
 volume = {25},
 year = {2007}
}

@article{Bai2015,
author = {Jushan Bai and Peng Wang},
title = {Identification and Bayesian Estimation of Dynamic Factor Models},
journal = {Journal of Business \& Economic Statistics},
volume = {33},
number = {2},
pages = {221-240},
year  = {2015},
publisher = {Taylor & Francis},
doi = {10.1080/07350015.2014.941467},
URL = { 
        https://doi.org/10.1080/07350015.2014.941467
},
eprint = { 
        https://doi.org/10.1080/07350015.2014.941467
}
}

@article{Forni2015,
title = {Dynamic factor models with infinite-dimensional factor spaces: One-sided representations},
journal = {Journal of Econometrics},
volume = {185},
number = {2},
pages = {359-371},
year = {2015},
issn = {0304-4076},
doi = {https://doi.org/10.1016/j.jeconom.2013.10.017},
url = {https://www.sciencedirect.com/science/article/pii/S0304407614002693},
author = {Mario Forni and Marc Hallin and Marco Lippi and Paolo Zaffaroni}
}

@article{bai_wang_2013,
title = {Identification theory for high dimensional static and dynamic factor models},
journal = {Journal of Econometrics},
volume = {178},
number = {2},
pages = {794-804},
year = {2014},
issn = {0304-4076},
doi = {https://doi.org/10.1016/j.jeconom.2013.11.001},
url = {https://www.sciencedirect.com/science/article/pii/S0304407613002315},
author = {Jushan Bai and Peng Wang}
}

@article{Moon_Weidner_2017, title={DYNAMIC LINEAR PANEL REGRESSION MODELS WITH INTERACTIVE FIXED EFFECTS}, volume={33}, DOI={10.1017/S0266466615000328}, number={1}, journal={Econometric Theory}, author={Moon, Hyungsik Roger and Weidner, Martin}, year={2017}, pages={158–195}}

@book{BrockwellDavis1991,
  author    = {Brockwell, Peter J. and Davis, Richard A.},
  title     = {Time Series: Theory and Methods},
  edition   = {2},
  year      = {1991},
  publisher = {Springer},
  address   = {New York},
  doi       = {10.1007/978-1-4419-0320-4}
}
\newpage

\section{Appendix}
Proof of Proposition \ref{alt_spec} and Proposition \ref{test_var} are in Supplemental Materials.

\subsection{Proof of Proposition \ref{consistency_1}}

\begin{proof}[Proof of Proposition \ref{consistency_1}]
Let
\[
C=G\Gamma'
\]
denote the true common-component matrix, so that $X=C+E$, and set
\[
a_{NT}
=
\frac{1}{\sqrt N}+\frac{1}{\sqrt T},
\qquad
x_{NT}
=
\frac{1}{\sqrt T}\fnorm{M_{\hat G}G}.
\]
Note that the minimization problem can be concentrated as
\begin{align*}
    \min_{(G,\Gamma)}\frac{1}{NT}\fnorm{X-G\Gamma'}^2 &= \min_{G}\frac{1}{NT}\fnorm{M_GX}^2.
\end{align*}
Because $(q,m)\in S_1$, Proposition \ref{alt_spec} implies that $C$ has an exact representation with structure $(q,m)$. Hence, by the optimality of $(\hat G,\hat\Gamma)$ and the definition of $P_{\hat G}$,
\[
\frac{1}{NT}\fnorm{M_{\hat G}X}^2
\leq
\frac{1}{NT}\fnorm{X-\hat G\hat\Gamma'}^2
\leq
\frac{1}{NT}\fnorm{E}^2.
\]
Using $X=G\Gamma'+E$ and the orthogonal decomposition $E=P_{\hat G}E+M_{\hat G}E$, we obtain
\begin{align}
0
\leq{}&
-\frac{1}{NT}\fnorm{M_{\hat G}G\Gamma'}^2
+\frac{1}{NT}\fnorm{P_{\hat G}E}^2 \notag\\
&-
\frac{2}{NT}
\operatorname{tr}
\left[
(M_{\hat G}G)'E\Gamma
\right].
\label{eq:basic_consistency}
\end{align}
For the first term,
\begin{align*}
\frac{1}{NT}\fnorm{M_{\hat G}G\Gamma'}^2
&=
\operatorname{tr}
\left[
\left(\frac{\Gamma'\Gamma}{N}\right)
\left(\frac{G'M_{\hat G}G}{T}\right)
\right]\\
&\geq
\lambda_{\min}
\left(
\frac{\Gamma'\Gamma}{N}
\right)
x_{NT}^2.
\end{align*}
By Assumption \ref{pd}, there exists a constant
$c>0$ such that
\[
\lambda_{\min}
\left(
\frac{\Gamma'\Gamma}{N}
\right)
\geq c
\]
with probability approaching one. Next, letting $r=qm$, we have $\operatorname{rank}(P_{\hat G})\leq r$. Because $(q,m)$ is fixed,
\begin{align*}
0
\leq
\frac{1}{NT}\fnorm{P_{\hat G}E}^2
&\leq
\frac{r}{NT}\|E\|^2\\
&=
O_p\left(\frac{1}{N}+\frac{1}{T}\right)\\
&=
O_p(a_{NT}^2),
\end{align*}
where $\|\cdot\|$ denotes the spectral norm and the second equality follows from Assumption \ref{error_norm}.

For the cross term, the Frobenius Cauchy--Schwarz inequality gives
\begin{align*}
\left|
\frac{2}{NT}
\operatorname{tr}
\left[
(M_{\hat G}G)'E\Gamma
\right]
\right|
&\leq
\frac{2}{NT}
\fnorm{M_{\hat G}G}
\fnorm{E\Gamma}\\
&\leq
\frac{2}{NT}
\fnorm{M_{\hat G}G}
\|E\|\fnorm{\Gamma}.
\end{align*}
Assumption \ref{pd} implies $\fnorm{\Gamma}=O_p(\sqrt N)$, while Assumption \ref{error_norm} implies $\|E\|=O_p(\sqrt{\max\{N,T\}})$. Therefore,
\begin{align*}
\left|
\frac{2}{NT}
\operatorname{tr}
\left[
(M_{\hat G}G)'E\Gamma
\right]
\right|
&=
O_p\left(
\frac{\sqrt{\max\{N,T\}}}{\sqrt{NT}}
\right)x_{NT}\\
&=
O_p(a_{NT})x_{NT}.
\end{align*}

Combining these bounds with \eqref{eq:basic_consistency}, there exist nonnegative random variables $A_{NT}=O_p(1)$ and $B_{NT}=O_p(1)$ such that, with probability approaching one,
\[
c x_{NT}^2
\leq
A_{NT}a_{NT}^2
+
B_{NT}a_{NT}x_{NT}.
\]
By Young's inequality,
\[
B_{NT}a_{NT}x_{NT}
\leq
\frac{c}{2}x_{NT}^2
+
\frac{B_{NT}^2}{2c}a_{NT}^2.
\]
It follows that
\[
x_{NT}^2
\leq
\left(
\frac{2A_{NT}}{c}
+
\frac{B_{NT}^2}{c^2}
\right)a_{NT}^2
=
O_p(a_{NT}^2).
\]
Consequently,
\[
\frac{1}{\sqrt T}\fnorm{M_{\hat G}G}
=
x_{NT}
=
O_p\left(
\frac{1}{\sqrt N}+\frac{1}{\sqrt T}
\right).
\]
\end{proof}

\subsection{Proof of Theorem \ref{thm:q_consistency} and \ref{thm:m_consistency}}

In the following lemmas, let
\[
C_0
=
\sum_{k=0}^{m_0-1}F_{-k}\lambda_k'
=
G\Gamma'
\]
denote the matrix of true common components, and let
\[
S
=
\{0,\ldots,q_{\max}\}
\times
\{0,\ldots,m_{\max}\},
\]
where $q_{\max}$ and $m_{\max}$ are fixed and $(q_0,m_0)\in S$. We represent the absence of a common component by
$(q_0,m_0)=(0,0)$.

For each $(q,m)\in S$ with $q,m\geq1$, let
\[
\left(
\hat F_{-k}^{(q,m)},
\hat\lambda_k^{(q,m)}
\right)_{k=0}^{m-1}
\]
be any global minimizer of the least-squares objective
\eqref{obj_fn}, where $\hat F_0^{(q,m)},\ldots,\hat F_{-(m-1)}^{(q,m)}$ are lag-stacked matrices generated by a common estimated $q$-dimensional factor sequence. Define
\begin{equation}
\label{common_hat}
\hat C_{q,m}
=
\sum_{k=0}^{m-1}
\hat F_{-k}^{(q,m)}
\hat\lambda_k^{(q,m)\prime}.
\end{equation}
For $q=0$ or $m=0$, set
\[
\hat C_{q,m}=0.
\]
We also define the spectral-norm residual
\[
\hat\delta_{q+1,m}(X)
=
\|X-\hat C_{q,m}\|.
\]
Because $\hat C_{q,m}\in\mathcal C_{q,m}$,
\[
\hat\delta_{q+1,m}(X)
\geq
\delta_{q+1,m}(X)
\]
in general; equality need not hold because $\hat C_{q,m}$ minimizes the Frobenius-norm criterion rather than the spectral-norm criterion.

Let $r_0=q_0m_0$. Consider the disjoint partition of $S$ given by
\begin{equation}
\begin{split}
\label{partition_set_finite}
S_1
&=
\left\{
(q,m)\in S:
q\geq q_0,\quad
qm\geq r_0
\right\},\\
S_2
&=
\left\{
(q,m)\in S:
qm<r_0
\right\},\\
S_3
&=
\left\{
(q,m)\in S:
q<q_0,\quad
qm\geq r_0
\right\}.
\end{split}
\end{equation}
Thus, $S_1$ contains the admissible structures, $S_2$ contains structures whose static dimension is smaller than the true static rank, and $S_3$ contains underfactored structures whose static dimension is at least $r_0$.

\begin{lemma}[Lower bound on the fitted residual]
\label{small_1_als}
For every $(q,m)\in S$,
\[
\hat\delta_{q+1,m}(X)
\geq
\delta_{r_0+qm+1,1}(E)
=
\sigma_{r_0+qm+1}(E).
\]
\end{lemma}

\begin{proof}
Since
\[
\operatorname{rank}(C_0)\leq r_0
\qquad\text{and}\qquad
\operatorname{rank}(\hat C_{q,m})\leq qm,
\]
we have
\[
\operatorname{rank}
\left(
\hat C_{q,m}-C_0
\right)
\leq
r_0+qm.
\]
Using $X=C_0+E$, it follows that
\begin{align*}
\hat\delta_{q+1,m}(X)
&=
\left\|
E-
\left(
\hat C_{q,m}-C_0
\right)
\right\|\\
&\geq
\inf_{\operatorname{rank}(A)\leq r_0+qm}
\|E-A\|\\
&=
\sigma_{r_0+qm+1}(E)\\
&=
\delta_{r_0+qm+1,1}(E),
\end{align*}
where the third equality follows from the Eckart--Young theorem.
\end{proof}

To apply Assumption \ref{lower_bound_E} to the preceding bound
uniformly over $S$, it is sufficient to choose
\[
k_{\max}
\geq
r_0+q_{\max}m_{\max}.
\]

\begin{lemma}[Consistency of the fitted common component]
\label{small_0_als}
Let Assumption \ref{error_norm} hold. Then
\[
\max_{(q,m)\in S_1}
\frac{1}{\sqrt{NT}}
\fnorm{
\hat C_{q,m}-C_0
}
=
O_p\left(
\sqrt{\frac{N+T}{NT}}
\right).
\]
\end{lemma}

\begin{proof}
Fix $(q,m)\in S_1$ and define
\[
\Delta_{q,m}
=
\hat C_{q,m}-C_0.
\]
By Proposition \ref{alt_spec}, $C_0$ admits an exact representation
with structure $(q,m)$ and is therefore feasible in the
least-squares problem defining $\hat C_{q,m}$. Global
optimality gives
\[
\fnorm{X-\hat C_{q,m}}^2
\leq
\fnorm{X-C_0}^2
=
\fnorm{E}^2.
\]
Since
\[
X-\hat C_{q,m}
=
E-\Delta_{q,m},
\]
we obtain
\[
\fnorm{\Delta_{q,m}}^2
\leq
2\left|
\operatorname{tr}
\left(
E'\Delta_{q,m}
\right)
\right|.
\]
By spectral--nuclear norm duality,
\[
\left|
\operatorname{tr}
\left(
E'\Delta_{q,m}
\right)
\right|
\leq
\|E\|
\|\Delta_{q,m}\|_*,
\]
where $\|\cdot\|_*$ denotes the nuclear norm. Moreover,
\[
\|\Delta_{q,m}\|_*
\leq
\sqrt{
\operatorname{rank}(\Delta_{q,m})
}
\fnorm{\Delta_{q,m}},
\]
and
\[
\operatorname{rank}(\Delta_{q,m})
\leq
\operatorname{rank}(\hat C_{q,m})
+
\operatorname{rank}(C_0)
\leq
qm+r_0.
\]
Therefore, if $\Delta_{q,m}\neq0$,
\[
\|\Delta_{q,m}\|_F
\leq
2\sqrt{qm+r_0}\,\|E\|.
\]
The same inequality holds trivially when $\Delta_{q,m}=0$.

Since
\[
qm+r_0
\leq
q_{\max}m_{\max}+r_0
\]
uniformly over $(q,m)\in S_1$, it follows that
\[
\max_{(q,m)\in S_1}
\|\hat C_{q,m}-C_0\|_F
\leq
2\sqrt{
q_{\max}m_{\max}+r_0
}\,\|E\|.
\]
Assumption \ref{error_norm} gives
\[
\|E\|
=
O_p\left(
\sqrt{\max\{N,T\}}
\right).
\]
Consequently,
\begin{align*}
\max_{(q,m)\in S_1}
\frac{1}{\sqrt{NT}}
\|\hat C_{q,m}-C_0\|_F
&=
O_p\left(
\sqrt{
\frac{\max\{N,T\}}{NT}
}
\right)\\
&=
O_p\left(
\sqrt{
\frac{N+T}{NT}
}
\right),
\end{align*}
which proves the result.
\end{proof}


\begin{lemma}[Upper bound on the fitted residual]\label{small_2_als}
Suppose Assumptions \ref{pd} and \ref{error_norm} hold. We have
\begin{align*}
\max_{(q,m)\in S_1}\frac{1}{\sqrt{NT}}\hat\delta_{q+1,m}(X)= O_p\left(\sqrt{\frac{N+T}{NT}}\right).
\end{align*}
\end{lemma}

\begin{proof}
Let
\[
C_0=\sum_{k=0}^{m_0-1}F_{-k}\lambda_k'.
\]
For every $(q,m)\in S_1$,
\begin{align*}
\hat\delta_{q+1,m}(X)
&=
\|X-\hat C_{q,m}\|\\
&\leq
\|E\|+\|\hat C_{q,m}-C_0\|\\
&\leq
\|E\|+\|\hat C_{q,m}-C_0\|_F.
\end{align*}
Assumption \ref{error_norm} gives
\[
\frac{\|E\|}{\sqrt{NT}}
=
O_p\left(
\sqrt{\frac{N+T}{NT}}
\right),
\]
while Lemma \ref{small_0_als} gives
\[
\max_{(q,m)\in S_1}
\frac{\|\hat C_{q,m}-C_0\|_F}{\sqrt{NT}}
=
O_p\left(
\sqrt{\frac{N+T}{NT}}
\right).
\]
The result follows.
\end{proof}

\begin{lemma}\label{small_3_als}
Suppose Assumptions \ref{error_norm} and Assumption \ref{lower_bound_E} hold with $q_0m_0+q_{max}m_{max}\leq k_{max}$. Let $(q,m)\in S_1$. Then, for all $\epsilon>0$, there exists $c_0,c_1>0$ such that
\begin{align*}
\mathbb{P}\left\{c_0\leq\frac{1}{\sqrt{N+T}}\hat\delta_{q+1,m}(X)\leq c_1\:\: \text{for all } (q,m)\in S_1\right\}\geq1-\epsilon
\end{align*}
for all $N, T$ large. 
\end{lemma}
\begin{proof}
By Lemma \ref{small_2_als},
\[
\max_{(q,m)\in S_1}
\frac{\hat\delta_{q+1,m}(X)}{\sqrt{N+T}}
=
O_p(1).
\]
Hence, for every $\epsilon>0$, there exists $c_U>0$ such that
\[
\mathbb P\left(
\max_{(q,m)\in S_1}
\frac{\hat\delta_{q+1,m}(X)}{\sqrt{N+T}}
>c_U
\right)
\leq\frac{\epsilon}{2}
\]
for all sufficiently large $N$ and $T$.

Moreover, Lemma \ref{small_1_als} and the ordering of singular values
give, uniformly over $(q,m)\in S_1$,
\[
\hat\delta_{q+1,m}(X)
\geq
\sigma_{r_0+qm+1}(E)
\geq
\sigma_{k_{\max}+1}(E).
\]
Assumption \ref{lower_bound_E} therefore implies that there exists
$c_L>0$ such that
\[
\mathbb P\left(
\min_{(q,m)\in S_1}
\frac{\hat\delta_{q+1,m}(X)}{\sqrt{N+T}}
<c_L
\right)
\leq\frac{\epsilon}{2}
\]
for all sufficiently large $N$ and $T$. Combining the two bounds
proves the result.
\end{proof}

\begin{lemma}\label{large_als}
Suppose Assumption \ref{pd}, \ref{error_norm}, and \ref{lower_bound_common} hold. For every $\epsilon>0$, there exist constants $c_L,c_U>0$ such that
\[
\mathbb P\left(
c_L
\leq
\frac{\hat\delta_{q+1,m}(X)}{\sqrt{NT}}
\leq
c_U
\text{ for every }(q,m)\in S_2\cup S_3
\right)
\geq1-\epsilon
\]
for all sufficiently large $N$ and $T$.
\end{lemma}
\begin{proof}
First, because $0$ is feasible in every candidate model and
$\hat C_{q,m}$ is a global least-squares minimizer,
\[
\hat\delta_{q+1,m}(X)
\leq
\|X-\hat C_{q,m}\|_F
\leq
\|X\|_F.
\]
Assumptions \ref{pd} and \ref{error_norm} imply
\[
\frac{\|X\|_F}{\sqrt{NT}}=O_p(1).
\]
Hence,
\[
\max_{(q,m)\in S_2\cup S_3}
\frac{\hat\delta_{q+1,m}(X)}{\sqrt{NT}}
=
O_p(1).
\]

For the lower bound, for every $(q,m)$,
\[
\frac{\hat\delta_{q+1,m}(X)}{\sqrt{NT}}
\geq
\frac{\delta_{q+1,m}(C_0)}{\sqrt{NT}}
-
\frac{\|E\|}{\sqrt{NT}}.
\]
Assumption \ref{error_norm} gives
\[
\frac{\|E\|}{\sqrt{NT}}=o_p(1).
\]

If $(q,m)\in S_2$, then $qm<r_0$, and hence
\[
\delta_{q+1,m}(C_0)
\geq
\sigma_{qm+1}(C_0)
\geq
\sigma_{r_0}(C_0).
\]
Moreover,
\[
\frac{\sigma_{r_0}^2(C_0)}{NT}
\geq
\lambda_{\min}\left(\frac{G'G}{T}\right)
\lambda_{\min}\left(\frac{\Gamma'\Gamma}{N}\right),
\]
which is bounded away from zero in probability by Assumption
\ref{pd}.

If $(q,m)\in S_3$, Assumption \ref{lower_bound_common} implies that
\[
\frac{\delta_{q+1,m}(C_0)}{\sqrt{NT}}
\]
is uniformly bounded away from zero in probability.

Since $S_2\cup S_3$ is finite, the preceding bounds imply that, for
every $\epsilon>0$, there exist constants $c_L,c_U>0$ such that
\[
\mathbb P\left(
c_L
\leq
\frac{\hat\delta_{q+1,m}(X)}{\sqrt{NT}}
\leq
c_U
\text{ for every }(q,m)\in S_2\cup S_3
\right)
\geq1-\epsilon
\]
for all sufficiently large $N$ and $T$.
\end{proof}

\textbf{Proof of Theorem \ref{thm:q_consistency}}
\begin{proof}[Proof of Theorem \ref{thm:q_consistency}]
Fix $m\in\{1,\ldots,m_{\max}\}$ and define
\[
q_m^\star
=
\max\left\{
q_0,
\left\lceil\frac{q_0m_0}{m}\right\rceil
\right\}.
\]
By the definition of $S_1$,
\[
(q,m)\in S_1
\quad\Longleftrightarrow\quad
q\geq q_m^\star.
\]
Thus, for the fixed candidate filter length $m$, structures with
$q<q_m^\star$ belong to $S_2\cup S_3$, whereas structures with
$q\geq q_m^\star$ belong to $S_1$.

Recall that
\[
DR_{q,m}(X)
=
\frac{\hat\delta_{q,m}(X)}
{\hat\delta_{q+1,m}(X)},
\qquad
q=1,\ldots,q_{\max},
\]
where $\hat\delta_{q,m}(X)$ is the residual associated with
$q-1$ fitted dynamic factors.

Because the candidate set is finite, Lemmas \ref{small_3_als} and
\ref{large_als}, together with a union bound, imply that, for every
$\epsilon>0$, there exist constants
\[
0<c_L<c_U<\infty,
\qquad
0<d_L<d_U<\infty,
\]
such that the event
\[
\begin{aligned}
\mathcal E_{NT}
={}&
\left\{\bigcap_{\ell=q_m^\star}^{q_{\max}}
\left\{
c_L
\leq
\frac{\hat\delta_{\ell+1,m}(X)}{\sqrt{N+T}}
\leq
c_U
\right\}\right\}\cap\left\{
\bigcap_{\ell=0}^{q_m^\star-1}
\left\{
d_L
\leq
\frac{\hat\delta_{\ell+1,m}(X)}{\sqrt{NT}}
\leq
d_U
\right\}\right\}
\end{aligned}
\]
satisfies
\[
\mathbb P(\mathcal E_{NT})\geq1-\epsilon
\]
for all sufficiently large $N$ and $T$.

We first consider the ratio at $q=q_m^\star$. Its numerator,
$\hat\delta_{q_m^\star,m}(X)$, corresponds to the structure
$(q_m^\star-1,m)\in S_2\cup S_3$, whereas its denominator,
$\hat\delta_{q_m^\star+1,m}(X)$, corresponds to
$(q_m^\star,m)\in S_1$. Hence, on $\mathcal E_{NT}$,
\begin{align*}
DR_{q_m^\star,m}(X)
&=
\frac{\hat\delta_{q_m^\star,m}(X)}
{\hat\delta_{q_m^\star+1,m}(X)}\\
&=
\sqrt{\frac{NT}{N+T}}\,
\frac{
\hat\delta_{q_m^\star,m}(X)/\sqrt{NT}
}{
\hat\delta_{q_m^\star+1,m}(X)/\sqrt{N+T}
}\\
&\geq
\frac{d_L}{c_U}
\sqrt{\frac{NT}{N+T}}.
\end{align*}
Because $N,T\to\infty$,
\[
\frac{NT}{N+T}
\geq
\frac{1}{2}\min\{N,T\}
\to\infty.
\]
Therefore,
\[
DR_{q_m^\star,m}(X)\to_p\infty.
\]

We next bound the ratios away from $q_m^\star$. If
$1\leq q<q_m^\star$, then both $(q-1,m)$ and $(q,m)$ belong to
$S_2\cup S_3$. Thus, on $\mathcal E_{NT}$,
\[
DR_{q,m}(X)
=
\frac{
\hat\delta_{q,m}(X)/\sqrt{NT}
}{
\hat\delta_{q+1,m}(X)/\sqrt{NT}
}
\leq
\frac{d_U}{d_L}.
\]
If $q_m^\star<q\leq q_{\max}$, then both $(q-1,m)$ and $(q,m)$
belong to $S_1$. Hence, on $\mathcal E_{NT}$,
\[
DR_{q,m}(X)
=
\frac{
\hat\delta_{q,m}(X)/\sqrt{N+T}
}{
\hat\delta_{q+1,m}(X)/\sqrt{N+T}
}
\leq
\frac{c_U}{c_L}.
\]
Consequently, on $\mathcal E_{NT}$,
\[
\max_{\substack{1\leq q\leq q_{\max}\\q\neq q_m^\star}}
DR_{q,m}(X)
\leq
M,
\qquad
M
=
\max\left\{
\frac{d_U}{d_L},
\frac{c_U}{c_L}
\right\}.
\]

Since
\[
\frac{d_L}{c_U}
\sqrt{\frac{NT}{N+T}}
\to\infty,
\]
for all sufficiently large $N$ and $T$,
\[
\frac{d_L}{c_U}
\sqrt{\frac{NT}{N+T}}
>
M.
\]
Therefore, on $\mathcal E_{NT}$ and for all sufficiently large
$N,T$,
\[
DR_{q_m^\star,m}(X)
>
\max_{\substack{1\leq q\leq q_{\max}\\q\neq q_m^\star}}
DR_{q,m}(X).
\]
It follows that the maximizer is uniquely attained at
$q_m^\star$ on $\mathcal E_{NT}$. Hence,
\[
\mathbb P\left\{\hat q=q_m^\star\right\}
\geq
\mathbb P(\mathcal E_{NT})
\geq
1-\epsilon
\]
for all sufficiently large $N$ and $T$. Since $\epsilon>0$ is
arbitrary,
\[
\mathbb P\left\{\hat q=q_m^\star\right\}\to1.
\]

Finally, if $m_0\leq m\leq m_{\max}$, then
\[
\frac{q_0m_0}{m}\leq q_0,
\]
and therefore
\[
q_m^\star=q_0.
\]
Thus,
\[
\mathbb P\{\hat q=q_0\}\to1.
\]

If instead $1\leq m<m_0$, then
\[
\frac{q_0m_0}{m}>q_0,
\]
so that
\[
q_m^\star
=
\left\lceil\frac{q_0m_0}{m}\right\rceil.
\]
Consequently,
\[
\mathbb P\left\{
\hat q
=
\left\lceil\frac{q_0m_0}{m}\right\rceil
\right\}
\to1.
\]
\end{proof}

\textbf{Proof of Theorem \ref{thm:m_consistency}}

\begin{proof}[Proof of Theorem \ref{thm:m_consistency}]
Fix $q\in\{q_0,\ldots,q_{\max}\}$ and let
\[
m_q^\star
=
\left\lceil\frac{q_0m_0}{q}\right\rceil.
\]
Because $q\geq q_0$, the definition of $S_1$ implies
\[
(q,m)\in S_1
\quad\Longleftrightarrow\quad
m\geq m_q^\star.
\]
Hence, $(q,m)\in S_2$ for $m<m_q^\star$ and
$(q,m)\in S_1$ for $m\geq m_q^\star$.

Since the candidate set is finite, Lemmas \ref{small_3_als} and
\ref{large_als} imply that, for every $\epsilon>0$, there exist
constants
\[
0<c_L<c_U<\infty,
\qquad
0<d_L<d_U<\infty,
\]
such that, with probability at least $1-\epsilon$ for all
sufficiently large $N$ and $T$,
\[
c_L
\leq
\frac{\hat\delta_{q+1,m}(X)}{\sqrt{N+T}}
\leq
c_U
\qquad
\text{for all }m=m_q^\star,\ldots,m_{\max},
\]
and
\[
d_L
\leq
\frac{\hat\delta_{q+1,m}(X)}{\sqrt{NT}}
\leq
d_U
\qquad
\text{for all }m=0,\ldots,m_q^\star-1.
\]

At $m=m_q^\star$, the numerator of $MR_{q,m}(X)$ corresponds to
the inadmissible structure $(q,m_q^\star-1)$, whereas the denominator
corresponds to the admissible structure $(q,m_q^\star)$. Therefore,
on the preceding event,
\begin{align*}
MR_{q,m_q^\star}(X)
&=
\frac{\hat\delta_{q+1,m_q^\star-1}(X)}
{\hat\delta_{q+1,m_q^\star}(X)}\\
&=
\sqrt{\frac{NT}{N+T}}\,
\frac{
\hat\delta_{q+1,m_q^\star-1}(X)/\sqrt{NT}
}{
\hat\delta_{q+1,m_q^\star}(X)/\sqrt{N+T}
}\\
&\geq
\frac{d_L}{c_U}
\sqrt{\frac{NT}{N+T}}.
\end{align*}
Since
\[
\frac{NT}{N+T}
\geq
\frac{1}{2}\min\{N,T\}
\to\infty,
\]
we have
\[
MR_{q,m_q^\star}(X)\to_p\infty.
\]

For $m<m_q^\star$, both $(q,m-1)$ and $(q,m)$ belong to $S_2$, so
\[
MR_{q,m}(X)
\leq
\frac{d_U}{d_L}.
\]
For $m>m_q^\star$, both $(q,m-1)$ and $(q,m)$ belong to $S_1$, so
\[
MR_{q,m}(X)
\leq
\frac{c_U}{c_L}.
\]
Thus,
\[
\max_{\substack{1\leq m\leq m_{\max}\\m\neq m_q^\star}}
MR_{q,m}(X)
=
O_p(1),
\]
whereas
\[
MR_{q,m_q^\star}(X)\to_p\infty.
\]
It follows that
\[
\mathbb P\{\hat m=m_q^\star\}\to1.
\]

Finally, when $q=q_0$,
\[
m_q^\star
=
\left\lceil\frac{q_0m_0}{q_0}\right\rceil
=
m_0,
\]
which gives
\[
\mathbb P\{\hat m=m_0\}\to1.
\]
For $q>q_0$, the stated conclusion follows directly from the
definition of $m_q^\star$.
\end{proof}

\subsection{Proof of Theorem \ref{test_IC}}

Let $V(q,m)$ denote the mean squared errors when the structure $(q,m)$ is used to estimate a dynamic factor model, i.e.,
\begin{align*}
V(q,m) 
= \frac{1}{NT}
\fnorm{X - \sum_{k=0}^{m-1} \hat F_{-k}^{(q)} \hat\lambda_k^{(q)\prime}}^2,
\end{align*}
where $(\hat F_{-k}^{(q,m)},\hat\lambda_k^{(q,m)})_{k=0}^{m-1}$ is the least squares estimator obtained under the structure $(q,m)$.

\begin{lemma}[Rates of the minimized mean squared errors]\label{V_rates}
Suppose that $q_0,m_0\geq0$ and that Assumptions \ref{pd},
\ref{error_norm}, and \ref{lower_bound_common} hold. Then
\[
\max_{(q,m)\in S_1}
\left\{
V(q_0,m_0)-V(q,m)
\right\}
=
O_p\left(
\frac{N+T}{NT}
\right).
\]
For every $\epsilon>0$, there exists a constant $c>0$ such that
\[
\mathbb P\left\{
\min_{(q,m)\in S_2\cup S_3}
\left[
V(q,m)-V(q_0,m_0)
\right]
\geq c
\right\}
\geq1-\epsilon
\]
for all sufficiently large $N$ and $T$.
\end{lemma}
\begin{proof}
For any given $(q,m)$, we have
    \begin{align*}
        V(q,m) &= \frac{1}{NT}\fnorm{\hat C_{q,m}-C_0}^2 + \frac{1}{NT}\fnorm{E}^2- \frac{2}{NT}\text{tr}\left(\left(\hat C_{q,m}-C_0\right)'E\right),
    \end{align*}
    so that
    \begin{equation}
        \begin{split}\label{V_diff}
            V(q_0,m_0)-V(q,m) &= \frac{1}{NT}\fnorm{\hat C_{q_0,m_0}-C_0}^2-\frac{1}{NT}\fnorm{\hat C_{q,m}-C_0}^2\\
        &-\frac{2}{NT}\text{tr}\left(\left(\hat C_{q_0,m_0}-C_0\right)'E\right)+\frac{2}{NT}\text{tr}\left(\left(\hat C_{q,m}-C_0\right)'E\right).
        \end{split}
    \end{equation}
    Also, we have
    \begin{align}\label{cross_diff}
        \frac{1}{NT}\text{tr}\left(\left(\hat C_{q,m}-C_0\right)'E\right)\leq \frac{\sqrt{mq+m_0q_0}}{NT}\fnorm{\hat C_{q,m}-C_0}\|E\|.
    \end{align}
    If $(q,m)\in S_1$, the RHS of equation \eqref{V_diff} is $O_p\left(\frac{N+T}{NT}\right)$ due to Lemma \ref{small_0_als}, equation \eqref{cross_diff}, and Assumption \ref{error_norm}.
    
    Now, if $(q,m)\in S_2\cup S_3$, we have $\frac{1}{NT}\fnorm{\hat C_{q,m}-C_0}^2$ strictly positive in the limit. To see this,
\begin{align*}
    \frac{1}{\sqrt{NT}}\fnorm{\hat C_{q,m}-C_0} &\geq \frac{1}{\sqrt{NT}}\left\|\hat C_{q,m}-C_0\right\|\geq\frac{1}{\sqrt{NT}}\delta_{q+1,m}\left(\sum_{k=0}^{m_0-1}F_{-k}\lambda_{-k}'\right),
\end{align*}
which is bounded away from zero in probability by Lemma \ref{large_als}. Thus, from \eqref{V_diff}, \eqref{cross_diff} and Lemma \ref{small_0_als}, we have
\begin{align*}
    V(q,m)-V(q_0,m_0) &= \frac{1}{NT}\fnorm{\hat C_{q,m}-C_0}^2+o_p(1)
\end{align*}
bounded away from zero for all $(q,m)\in S_2\cup S_3$.
\end{proof}

\textbf{Proof of Theorem \ref{test_IC}}

\begin{proof}[Proof of Theorem \ref{test_IC}]
    We show that 
    \[
    \mathbb{P}\{PC(q,m) < PC(q_0,m_0)\}\to 0
    \]
    as $N,T\to\infty$ for all $(q,m)\in S\setminus\{(q_0,m_0)\}$. 

    We first consider the case $(q,m)\in S_1\setminus\{(q_0,m_0)\}$. Using Lemma \ref{V_rates}, we have
    \[
    \max_{(q,m)\in S_1}
    \left\{
    V(q_0,m_0)-V(q,m)
    \right\}
    =
    O_p\left(
    \frac{N+T}{NT}
    \right).
    \]
    Note that we may write
    \[
    \mathbb{P}\{PC(q,m)<PC(q_0,m_0)\}=\mathbb{P}\{V(q_0,m_0)-V(q,m) >(r+q-(r_0+q_0))g(N,T)\},
    \]
    where the $V(q_0,m_0)-V(q,m)$ vanishes at the order $O_p\left(\frac{N+T}{NT}\right)$ and the $g(N,T)$ vanishes at a slower rate. Note that any $(q,m)\in S_1\setminus \{(q_0,m_0)\}$ implies $r+q>r_0+q_0$ so that $\mathbb{P}\{PC(q,m)<PC(q_0,m_0)\}\to 0$.

    Next, we consider the case $(q,m)\in S_2\cup S_3$. Using Lemma \ref{V_rates}, for every $\epsilon>0$, there exists a constant $c>0$ such that
    \[
    \mathbb P\left\{
    \min_{(q,m)\in S_2\cup S_3}
    \left[
    V(q,m)-V(q_0,m_0)
    \right]
    \geq c
    \right\}
    \geq1-\epsilon
    \]
    for all sufficiently large $N$ and $T$. Note that we have
    \[
    PC(q,m)-PC(q_0,m_0) = V(q,m)-V(q_0,m_0) -(r_0+q_0-(r+q))g(N,T),
    \]
    where the first term $V(q,m)-V(q_0,m_0)$ is bounded away from zero, whereas the second term $(r_0+q_0-(r+q))g(N,T)$ vanishes. Thus, we have $\mathbb{P}\{PC(q,m)<PC(q_0,m_0)\}\to 0$.

    For $DC(q,m)$, note that we have
    \[
    \frac{1}{NT}\hat\delta^2_{q+1,m}(X) = O_p\left(\frac{N+T}{NT}\right)
    \]
    when $(q,m)\in S_1$ due to Lemma \ref{small_2_als}. Also, for all $\epsilon>0$, there exists $c>0$ such that
    \begin{align*}
    \mathbb{P}\left\{\frac{1}{NT}\hat\delta^2_{q+1,m}(X) \geq c\right\}\geq1-\epsilon
    \end{align*}
    when $(q,m)\in S_2\cup S_3$ due to Lemma \ref{large_als}. Since $\frac{1}{NT}\hat\delta^2_{q_0+1,m_0}(X)=O_p\left(\frac{N+T}{NT}\right)$, we may use the same argument as in the case of $PC(q,m)$ to finish.
\end{proof}

\begin{proof}[Proof of Corollary \ref{IC}]
Let
\[
a_{NT}=\frac{N+T}{NT},
\qquad
K(q,m)=qm+q.
\]
By Lemma \ref{V_rates},
\[
\max_{(q,m)\in S_1}
\left|
V(q,m)-\frac{\|E\|_F^2}{NT}
\right|
=
O_p(a_{NT}).
\]
Assumption \ref{average_error_variance} therefore implies
\[
\max_{(q,m)\in S_1}
|V(q,m)-\sigma_\varepsilon^2|
\to_p0.
\]
The mean-value theorem then gives
\[
\max_{(q,m)\in S_1}
\left|
\log V(q,m)-\log V(q_0,m_0)
\right|
=
O_p(a_{NT}).
\]

For every
$(q,m)\in S_1\setminus\{(q_0,m_0)\}$,
\[
K(q,m)-K(q_0,m_0)>0.
\]
Since $g(N,T)/a_{NT}\to\infty$,
\[
\mathbb P\left\{
\min_{(q,m)\in
S_1\setminus\{(q_0,m_0)\}}
[
IC(q,m)-IC(q_0,m_0)
]>0
\right\}
\to1.
\]

For $(q,m)\in S_2\cup S_3$, Lemma \ref{V_rates} gives a constant
$c>0$ such that
\[
\mathbb P\left\{
\min_{(q,m)\in S_2\cup S_3}
[V(q,m)-V(q_0,m_0)]
\ge c
\right\}
\to1.
\]
Because $V(q_0,m_0)\to_p\sigma_\varepsilon^2>0$,
\[
\min_{(q,m)\in S_2\cup S_3}
[
\log V(q,m)-\log V(q_0,m_0)
]
\]
is bounded away from zero in probability. Since $g(N,T)\to0$ and
the candidate set is finite,
\[
\mathbb P\left\{
\min_{(q,m)\in S_2\cup S_3}
[
IC(q,m)-IC(q_0,m_0)
]>0
\right\}
\to1.
\]
Combining the two cases proves
\[
\mathbb P\{(\widehat q,\widehat m)=(q_0,m_0)\}\to1.
\]
\end{proof}

\clearpage
\appendix
\setcounter{page}{1}

\section{Supplemental Materials for "Determining the Structure of Dynamic Factor Models"}
\vspace{0.3cm}
by Sangmyung Ha
\vspace{0.3cm}
\subsection{Proof of Proposition \ref{alt_spec} and Corollary \ref{cor:lower_bounds}}

\begin{proof}[Proof of Proposition \ref{alt_spec}]
Let
\[
r_0=q_0m_0,\qquad r=qm.
\]
We first prove that every pair $(q,m)$ satisfying
\[
q\geq q_0,\qquad qm\geq q_0m_0
\]
admits an exact, possibly redundant, representation of the common component.

\emph{Sufficiency.} Suppose first that $m\geq m_0$. Let $p=q-q_0$ and define
\[
h_t=\begin{pmatrix}f_t\\0_{p\times1}\end{pmatrix},\qquad \widetilde\Gamma_k=\begin{pmatrix}\lambda_k&0_{N\times p}\end{pmatrix}\quad\text{for }k=0,\ldots,m_0-1,
\]
and set $\widetilde\Gamma_k=0$ for $k=m_0,\ldots,m-1$. Then
\[
\sum_{k=0}^{m-1}\widetilde\Gamma_kh_{t-k}=\sum_{k=0}^{m_0-1}\lambda_kf_{t-k}=\chi_t.
\]
Now suppose that $1\leq m<m_0$. Define
\[
p=q-q_0,\qquad d=m_0-m.
\]
Since $qm\geq q_0m_0$, we have
\[
pm=(q-q_0)m=qm-q_0m\geq q_0m_0-q_0m=q_0d.
\]
In particular, $p\geq1$.

Let $\Pi\in\mathbb R^{q_0\times q_0}$ be the cyclic permutation matrix satisfying $\Pi e_a=e_{a+1}$ for $a=1,\ldots,q_0$, where $e_a\in\mathbb{R}^{q_0},a=1,\ldots,q_0$ are $q_0$-dimensional elementary basis and $e_{q_0+1}=e_1$. Define
\[
T=
\begin{pmatrix}
0&0&\cdots&0&\Pi\\
I_{q_0}&0&\cdots&0&0\\
0&I_{q_0}&\ddots&0&0\\
\vdots&\ddots&\ddots&\vdots&\vdots\\
0&0&\cdots&I_{q_0}&0
\end{pmatrix}
\in\mathbb R^{q_0d\times q_0d},
\]
with the convention that $T=\Pi$ when $d=1$. The matrix $T$ is a permutation matrix consisting of a single cycle of length $q_0d$. Indeed, if its canonical basis vectors are indexed by $(a,j)$, where $a=1,\ldots,q_0$ and $j=1,\ldots,d$, then
\[
Te_{a,j}=e_{a,j+1}\quad\text{for }j<d,\qquad Te_{a,d}=e_{a+1,1},
\]
where the first index is interpreted cyclically.

Let
\[
s=\left\lceil\frac{q_0d}{m}\right\rceil.
\]
Because $pm\geq q_0d$, we have $s\leq p$. Fix a canonical row vector $e'\in\mathbb R^{1\times q_0d}$ and define $C\in\mathbb R^{p\times q_0d}$ so that its first $s$ rows are
\[
e',\quad e'T^m,\quad e'T^{2m},\quad\ldots,\quad e'T^{(s-1)m},
\]
and any remaining rows are zero. Consider
\[
\mathcal O_m=
\begin{pmatrix}
C\\
CT\\
\vdots\\
CT^{m-1}
\end{pmatrix}
\in\mathbb R^{pm\times q_0d}.
\]
For $j=0,\ldots,s-1$ and $k=0,\ldots,m-1$, the corresponding row of $\mathcal O_m$ is $e'T^{jm+k}$. Therefore, $\mathcal O_m$ contains the rows
\[
e',\quad e'T,\quad\ldots,\quad e'T^{sm-1}.
\]
Since $sm\geq q_0d$, it contains $e',e'T,\ldots,e'T^{q_0d-1}$. Because $T$ is a single-cycle permutation matrix, these are the $q_0d$ distinct canonical row vectors in some order. Hence
\[
\operatorname{rank}(\mathcal O_m)=q_0d.
\]
Partition $C$ into $d$ blocks of $q_0$ columns:
\[
C=\begin{pmatrix}B_1&B_2&\cdots&B_d\end{pmatrix},\qquad B_j\in\mathbb R^{p\times q_0}.
\]
Define the alternative factor process by
\[
h_t=
\begin{pmatrix}
f_t-\Pi f_{t-d}\\[1mm]
\displaystyle\sum_{j=1}^dB_jf_{t-j}
\end{pmatrix}
\in\mathbb R^{q_0+p}=\mathbb R^q.
\]
Let
\[
g_t=\begin{pmatrix}f_t'&f_{t-1}'&\cdots&f_{t-m_0+1}'\end{pmatrix}'\in\mathbb R^{r_0}
\]
and
\[
\tilde g_t=\begin{pmatrix}h_t'&h_{t-1}'&\cdots&h_{t-m+1}'\end{pmatrix}'\in\mathbb R^r.
\]
By construction, there exists a matrix $\mathcal H\in\mathbb R^{r\times r_0}$ such that
\[
\tilde g_t=\mathcal Hg_t.
\]
We show that $\mathcal H$ has full column rank.

Let
\[
y=\begin{pmatrix}y_0'&y_1'&\cdots&y_{m_0-1}'\end{pmatrix}',\qquad y_j\in\mathbb R^{q_0},
\]
and suppose that $\mathcal Hy=0$. For each $k=0,\ldots,m-1$, the first $q_0$ coordinates of the corresponding equation imply
\[
y_k-\Pi y_{k+d}=0,
\]
and the remaining $p$ coordinates imply
\[
\sum_{j=1}^dB_jy_{k+j}=0.
\]
For $k=0,\ldots,m$, define
\[
z_k=\begin{pmatrix}y_k'&y_{k+1}'&\cdots&y_{k+d-1}'\end{pmatrix}'\in\mathbb R^{q_0d}.
\]
The first set of equations gives
\[
z_k=Tz_{k+1},\qquad k=0,\ldots,m-1,
\]
while the second set gives
\[
Cz_{k+1}=0,\qquad k=0,\ldots,m-1.
\]
Iterating the recursion yields
\[
z_{k+1}=T^{m-k-1}z_m.
\]
It follows that
\[
Cz_m=0,\qquad CTz_m=0,\qquad\ldots,\qquad CT^{m-1}z_m=0,
\]
or equivalently,
\[
\mathcal O_mz_m=0.
\]
Since $\mathcal O_m$ has full column rank, $z_m=0$. The recursion $z_k=Tz_{k+1}$ then gives $z_{m-1}=\cdots=z_0=0$, and therefore
\[
y_0=y_1=\cdots=y_{m_0-1}=0.
\]
Thus,
\[
\ker(\mathcal H)=\{0\},\qquad \operatorname{rank}(\mathcal H)=r_0.
\]
Since $\mathcal H$ has full column rank, it has a left inverse. For example, define
\[
\mathcal H^\ell=(\mathcal H'\mathcal H)^{-1}\mathcal H',\qquad \mathcal H^\ell\mathcal H=I_{r_0}.
\]
Let
\[
\Gamma=\begin{pmatrix}\lambda_0&\lambda_1&\cdots&\lambda_{m_0-1}\end{pmatrix}\in\mathbb R^{N\times r_0}
\]
be the true stacked loading matrix, and define
\[
\widetilde\Gamma=\Gamma\mathcal H^\ell\in\mathbb R^{N\times r}.
\]
Partition $\widetilde\Gamma$ into $m$ blocks of $q$ columns:
\[
\widetilde\Gamma=\begin{pmatrix}\widetilde\Gamma_0&\widetilde\Gamma_1&\cdots&\widetilde\Gamma_{m-1}\end{pmatrix}.
\]
Then
\[
\sum_{k=0}^{m-1}\widetilde\Gamma_kg_{t-k}=\widetilde\Gamma\tilde g_t=\Gamma\mathcal H^\ell\mathcal Hg_t=\Gamma g_t=\sum_{j=0}^{m_0-1}\lambda_jf_{t-j}=\chi_t.
\]
This proves sufficiency.

\emph{Necessity.} By definition, $q_0$ is the smallest dynamic factor dimension among all exact representations of the common component. Therefore, every exact representation must satisfy $q\geq q_0.$ If a structure $(q,m)$ admits an exact representation of the common component, its static representation $(qm,1)$ is also an exact representation. By Assumption \ref{pd}, we have $qm\geq q_0m_0$ with probability approaching one. 
\end{proof}

\begin{proof}[Proof of Corollary \ref{cor:lower_bounds}]
By Proposition \ref{alt_spec}, admissible structures are
asymptotically characterized by
\[
q\geq q_0
\qquad\text{and}\qquad
qm\geq r_0.
\]
For fixed $m\geq1$, these conditions are equivalent to
\[
q\geq
\max\left\{
q_0,
\left\lceil\frac{r_0}{m}\right\rceil
\right\}.
\]
If $m<m_0$, then $r_0/m>q_0$, whereas if $m\geq m_0$, then
$\lceil r_0/m\rceil\leq q_0$, which gives the stated expression for
$q_{\min}(m)$. For fixed $q\geq q_0$, the condition $qm\geq r_0$
is equivalent to
\[
m\geq\left\lceil\frac{r_0}{q}\right\rceil.
\]
The sufficiency and asymptotic sharpness of these bounds follow
directly from Proposition \ref{alt_spec}.
\end{proof}

\subsection{Proof of Proposition \ref{consistency_2}}
\textbf{Proof of Proposition \ref{consistency_2}}
\begin{proof}
Let
\[
r=mq,
\qquad
a_{NT}
=
\frac{1}{\sqrt N}+\frac{1}{\sqrt T},
\]
and define the true and estimated common-component matrices by
\[
C=G\Gamma',
\qquad
\hat C=\hat G\hat\Gamma'.
\]

We divide the proof into four steps.

\medskip
\noindent
\emph{Step 1: Static rotated consistency.}

Since the true structure is used in estimation, $C$ is feasible in
the least-squares problem defining $\hat C$. Hence,
\[
\|X-\hat C\|_F^2
\leq
\|X-C\|_F^2
=
\|E\|_F^2.
\]
Let
\[
\Delta_C=\hat C-C.
\]
Because $X-\hat C=E-\Delta_C$, the preceding inequality implies
\[
\|\Delta_C\|_F^2
\leq
2
\left|
\operatorname{tr}
\left(
E'\Delta_C
\right)
\right|.
\]
Moreover,
\[
\operatorname{rank}(\Delta_C)
\leq
2r.
\]
By spectral--nuclear norm duality,
\[
\left|
\operatorname{tr}
\left(
E'\Delta_C
\right)
\right|
\leq
\|E\|
\|\Delta_C\|_*
\leq
\sqrt{2r}\,
\|E\|
\|\Delta_C\|_F.
\]
Therefore,
\[
\|\Delta_C\|_F
\leq
2\sqrt{2r}\,\|E\|
=
O_p\left(\sqrt{\max\{N,T\}}\right),
\]
and consequently
\[
\frac{\|\hat C-C\|}{\sqrt{NT}}
=
o_p(1).
\]

Assumption \ref{pd} implies that the $r$ nonzero singular values of
$C$ are of order $\sqrt{NT}$. More precisely,
\begin{align*}
\frac{\sigma_r^2(C)}{NT}
&=
\lambda_{\min}
\left[
\left(
\frac{\Gamma'\Gamma}{N}
\right)^{1/2}
\left(
\frac{G'G}{T}
\right)
\left(
\frac{\Gamma'\Gamma}{N}
\right)^{1/2}
\right]\\
&\geq
\lambda_{\min}
\left(
\frac{G'G}{T}
\right)
\lambda_{\min}
\left(
\frac{\Gamma'\Gamma}{N}
\right).
\end{align*}
Thus, for some $c>0$,
\[
\mathbb P\left\{
\sigma_r(C)\geq c\sqrt{NT}
\right\}
\to1.
\]
Weyl's inequality and
$\|\hat C-C\|=o_p(\sqrt{NT})$ then give
\[
\sigma_r(\hat C)\geq\frac{c}{2}\sqrt{NT}
\]
with probability approaching one.

The loading normalization implies
\[
\|\hat\Gamma\|_F^2
=
\sum_{k=0}^{m-1}
\|\hat\lambda_k\|_F^2
=
Nq,
\]
and hence
\[
\|\hat\Gamma\|
\leq
\sqrt{Nq}.
\]
Since
\[
\sigma_r(\hat C)
\leq
\sigma_r(\hat G)
\|\hat\Gamma\|,
\]
it follows that
\[
\sigma_r(\hat G)
\geq
c_G\sqrt T
\]
with probability approaching one for some $c_G>0$. In particular,
$\hat G$ has full column rank with probability approaching one.

Because the true and estimated static factor spaces have the same
dimension, Proposition \ref{consistency_1} implies
\[
\|P_{\hat G}-P_G\|_F
=
O_p(a_{NT}).
\]
Indeed,
\begin{align*}
\|P_{\hat G}-P_G\|_F^2
&=
2\operatorname{tr}(M_{\hat G}P_G)\\
&=
2\operatorname{tr}
\left[
\left(
\frac{G'G}{T}
\right)^{-1}
\left(
\frac{G'M_{\hat G}G}{T}
\right)
\right]\\
&\leq
2
\left\|
\left(
\frac{G'G}{T}
\right)^{-1}
\right\|
\frac{\|M_{\hat G}G\|_F^2}{T}\\
&=
O_p(a_{NT}^2).
\end{align*}

Define
\[
R_{NT}
=
(G'G)^{-1}G'\hat G
\]
and
\[
D_{NT}
=
\hat G-GR_{NT}
=
M_G\hat G.
\]
Since $P_{\hat G}\hat G=\hat G$,
\[
D_{NT}
=
(P_{\hat G}-P_G)\hat G.
\]
The compactness condition on $\hat G$ therefore gives
\[
\frac{1}{\sqrt T}
\|D_{NT}\|_F
=
O_p(a_{NT}).
\]
Assumption \ref{pd} and the compactness condition also imply
\[
\|R_{NT}\|=O_p(1).
\]

Furthermore,
\[
\hat G
=
GR_{NT}+D_{NT},
\qquad
G'D_{NT}=0.
\]
Since $\sigma_r(\hat G)/\sqrt T$ is bounded away from zero and
$\|D_{NT}\|/\sqrt T=o_p(1)$,
\[
\frac{\sigma_r(GR_{NT})}{\sqrt T}
\]
is bounded away from zero in probability. Because
$\|G\|/\sqrt T=O_p(1)$, it follows that
\[
\|R_{NT}^{-1}\|=O_p(1).
\]

We next derive the corresponding loading rate. The least-squares
first-order condition gives
\[
\hat\Gamma
=
X'\hat G
(\hat G'\hat G)^{-1}.
\]
Write
\[
A_{NT}
=
R_{NT}'G'GR_{NT},
\qquad
B_{NT}
=
D_{NT}'D_{NT}.
\]
Since $G'D_{NT}=0$,
\[
\hat G'\hat G=A_{NT}+B_{NT}
\]
and
\[
G'\hat G=G'GR_{NT}.
\]
Therefore,
\begin{align*}
\hat\Gamma
&=
\Gamma G'GR_{NT}
(A_{NT}+B_{NT})^{-1}
+
E'\hat G
(\hat G'\hat G)^{-1}\\
&=
\Gamma R_{NT}^{-T}
-
\Gamma R_{NT}^{-T}
B_{NT}
(A_{NT}+B_{NT})^{-1}\\
&\quad+
E'\hat G
(\hat G'\hat G)^{-1}.
\end{align*}
Because
\[
\frac{\|D_{NT}\|_F}{\sqrt T}
=
O_p(a_{NT}),
\]
we have
\[
\left\|
\frac{B_{NT}}{T}
\right\|
=
O_p(a_{NT}^2).
\]
Moreover,
\[
\left\|
T(A_{NT}+B_{NT})^{-1}
\right\|
=
O_p(1).
\]
It follows that
\[
\frac{1}{\sqrt N}
\left\|
\Gamma R_{NT}^{-T}
B_{NT}
(A_{NT}+B_{NT})^{-1}
\right\|_F
=
O_p(a_{NT}^2).
\]
For the error term,
\begin{align*}
\frac{1}{\sqrt N}
\left\|
E'\hat G
(\hat G'\hat G)^{-1}
\right\|_F
&\leq
\frac{\|E\|}{\sqrt N}
\left\|
\hat G
(\hat G'\hat G)^{-1}
\right\|_F\\
&=
O_p(a_{NT}),
\end{align*}
because
\[
\left\|
\hat G
(\hat G'\hat G)^{-1}
\right\|_F
=
O_p(T^{-1/2}).
\]
Consequently,
\begin{equation}
\label{eq:static_factor_loading_rotation}
\frac{1}{\sqrt T}
\|\hat G-GR_{NT}\|_F
=
O_p(a_{NT}),
\qquad
\frac{1}{\sqrt N}
\|\hat\Gamma-\Gamma R_{NT}^{-T}\|_F
=
O_p(a_{NT}),
\end{equation}
with
\[
\|R_{NT}\|+\|R_{NT}^{-1}\|=O_p(1).
\]

\medskip
\noindent
\emph{Step 2: The static rotation preserves the lag structure.}

Partition $R_{NT}$ into $q\times q$ blocks:
\[
R_{NT}
=
\begin{pmatrix}
R_{11}&\cdots&R_{1m}\\
\vdots&\ddots&\vdots\\
R_{m1}&\cdots&R_{mm}
\end{pmatrix}.
\]
Also partition
\[
D_{NT}
=
\begin{pmatrix}
D_1&\cdots&D_m
\end{pmatrix},
\qquad
D_b\in\mathbb R^{T\times q}.
\]
Then, for $b=1,\ldots,m$,
\[
\hat F_{-(b-1)}
=
\sum_{a=1}^m
F_{-(a-1)}R_{ab}
+
D_b.
\]

For $b=1,\ldots,m-1$, the estimated lag blocks satisfy the exact
shift identity
\[
J_0\hat F_{-(b-1)}
=
J_1\hat F_{-b}.
\]
Using
\[
J_0F_{-(a-1)}
=
J_1F_{-a},
\]
we obtain
\[
\overline G^{(m+1)}K_b
=
J_0D_b-J_1D_{b+1},
\]
where
\[
K_b
=
\begin{pmatrix}
R_{1,b+1}\\
R_{2,b+1}-R_{1b}\\
\vdots\\
R_{m,b+1}-R_{m-1,b}\\
-R_{mb}
\end{pmatrix}
\in\mathbb R^{q(m+1)\times q}.
\]
Equation \eqref{eq:static_factor_loading_rotation} implies
\[
\frac{1}{\sqrt T}
\|J_0D_b-J_1D_{b+1}\|_F
=
O_p(a_{NT}).
\]
Assumption \ref{pd_extra} therefore gives
\[
\|K_b\|_F
=
O_p(a_{NT}),
\qquad
b=1,\ldots,m-1.
\]

Hence,
\[
R_{1,b+1}
=
O_p(a_{NT}),
\qquad
R_{m,b}
=
O_p(a_{NT}),
\]
and
\[
R_{a,b+1}-R_{a-1,b}
=
O_p(a_{NT}),
\qquad
a=2,\ldots,m.
\]
Iterating these relations along the block diagonals yields
\[
R_{ab}
=
O_p(a_{NT})
\qquad
\text{for }a\neq b,
\]
and
\[
R_{aa}-R_{bb}
=
O_p(a_{NT})
\qquad
\text{for all }a,b.
\]
Let
\[
H_{NT}'=R_{11}.
\]
Since $m$ is fixed,
\begin{equation}
\label{eq:block_rotation}
\left\|
R_{NT}
-
I_m\otimes H_{NT}'
\right\|_F
=
O_p(a_{NT}).
\end{equation}
The boundedness of $R_{NT}$ and $R_{NT}^{-1}$, together with
\eqref{eq:block_rotation}, implies
\[
\|H_{NT}\|+\|H_{NT}^{-1}\|
=
O_p(1).
\]
The inverse perturbation identity then gives
\begin{equation}
\label{eq:block_inverse_rotation}
\left\|
R_{NT}^{-T}
-
I_m\otimes H_{NT}^{-1}
\right\|_F
=
O_p(a_{NT}).
\end{equation}

\medskip
\noindent
\emph{Step 3: Blockwise factor and loading consistency.}

Combining \eqref{eq:static_factor_loading_rotation} and
\eqref{eq:block_rotation},
\begin{align*}
\frac{1}{\sqrt T}
\left\|
\hat G
-
G(I_m\otimes H_{NT}')
\right\|_F
&\leq
\frac{1}{\sqrt T}
\|\hat G-GR_{NT}\|_F\\
&\quad+
\frac{\|G\|_F}{\sqrt T}
\left\|
R_{NT}-I_m\otimes H_{NT}'
\right\|\\
&=
O_p(a_{NT}).
\end{align*}
The first $q$-column block therefore satisfies
\begin{equation}
\label{eq:F_HNT_rate}
\frac{1}{\sqrt T}
\|\hat F-FH_{NT}'\|_F
=
O_p(a_{NT}).
\end{equation}

Similarly, equations
\eqref{eq:static_factor_loading_rotation} and
\eqref{eq:block_inverse_rotation} imply
\begin{equation}
\label{eq:Gamma_HNT_rate}
\frac{1}{\sqrt N}
\left\|
\hat\Gamma
-
\Gamma
(I_m\otimes H_{NT}^{-1})
\right\|_F
=
O_p(a_{NT}).
\end{equation}
Hence, for each $k=0,\ldots,m-1$,
\begin{equation}
\label{eq:lambda_HNT_rate}
\frac{1}{\sqrt N}
\left\|
\hat\lambda_k
-
\lambda_kH_{NT}^{-1}
\right\|_F
=
O_p(a_{NT}).
\end{equation}

\medskip
\noindent
\emph{Step 4: Identification of the explicit rotation \(H_f\).}

Using \eqref{eq:lambda_HNT_rate} and the true loading
normalization,
\begin{align*}
\frac{1}{N}
\sum_{k=0}^{m-1}
\hat\lambda_k'
\hat\lambda_k
&=
H_{NT}^{-T}
\left(
\frac{1}{N}
\sum_{k=0}^{m-1}
\lambda_k'\lambda_k
\right)
H_{NT}^{-1}
+
O_p(a_{NT})\\
&=
H_{NT}^{-T}H_{NT}^{-1}
+
O_p(a_{NT}).
\end{align*}
The estimated loading normalization therefore implies
\[
H_{NT}^{-T}H_{NT}^{-1}
=
I_q+O_p(a_{NT}),
\]
and hence
\[
H_{NT}H_{NT}'
=
I_q+O_p(a_{NT}).
\]
It follows that
\begin{equation}
\label{eq:HNT_orthogonal}
H_{NT}^{-1}
=
H_{NT}'
+
O_p(a_{NT}).
\end{equation}

Because
\[
\sum_{k=0}^{m-1}
\hat\lambda_k'
\hat\lambda_k
=
NI_q,
\]
the definition of $H_f$ becomes
\[
H_f'
=
\frac{1}{N}
\sum_{k=0}^{m-1}
\lambda_k'\hat\lambda_k.
\]
Using \eqref{eq:lambda_HNT_rate},
\begin{align*}
H_f'
&=
\left(
\frac{1}{N}
\sum_{k=0}^{m-1}
\lambda_k'\lambda_k
\right)
H_{NT}^{-1}
+
O_p(a_{NT})\\
&=
H_{NT}^{-1}
+
O_p(a_{NT})\\
&=
H_{NT}'
+
O_p(a_{NT}),
\end{align*}
where the last equality follows from
\eqref{eq:HNT_orthogonal}. Thus,
\[
\|H_f-H_{NT}\|
=
O_p(a_{NT}),
\]
and
\[
\|H_f\|+\|H_f^{-1}\|
=
O_p(1).
\]

Finally, equations \eqref{eq:F_HNT_rate} and
\eqref{eq:lambda_HNT_rate}, together with the preceding rotation
bound, give
\[
\frac{1}{\sqrt T}
\|\hat F-FH_f'\|_F
=
O_p(a_{NT})
\]
and, for each $k=0,\ldots,m-1$,
\[
\frac{1}{\sqrt N}
\|\hat\lambda_k-\lambda_kH_f^{-1}\|_F
=
O_p(a_{NT}).
\]
The stacked loading result follows by summing over the fixed number
$m$ of lag blocks.
\end{proof}

\subsection{Proof of Theorem \ref{test_var}}
\textbf{Proof of Theorem \ref{test_var}}
\begin{proof}[Proof of Theorem \ref{test_var}]
Let
\[
n=T-p_{\max},
\qquad
a_{NT}
=
\frac{N+T}{NT}
=
\frac{1}{N}+\frac{1}{T},
\]
Define the rotated true-factor matrices
\[
\widetilde Y=YH_f',
\qquad
\widetilde Z_\ell
=
Z_\ell(I_\ell\otimes H_f'),
\qquad
\ell\geq1.
\]
Because $H_f$ is nonsingular,
\[
\operatorname{col}(\widetilde Z_\ell)
=
\operatorname{col}(Z_\ell).
\]
Let
\[
\widetilde v_\ell
=
\frac{1}{n}
\min_B
\left\|
\widetilde Y-\widetilde Z_\ell B'
\right\|_F^2,
\]
with
\[
\widetilde v_0
=
\frac{1}{n}
\|\widetilde Y\|_F^2.
\]

Set
\[
\Delta_Y
=
\hat Y-\widetilde Y,
\qquad
\Delta_{Z,\ell}
=
\hat Z_\ell-\widetilde Z_\ell.
\]
Since $p_{\max}$ is fixed, Proposition \ref{consistency_2} implies
\begin{equation}
\label{eq:factor_VAR_rate}
\frac{1}{n}\|\Delta_Y\|_F^2
+
\max_{1\leq\ell\leq p_{\max}}
\frac{1}{n}\|\Delta_{Z,\ell}\|_F^2
=
O_p(a_{NT}).
\end{equation}

For a full-column-rank matrix $A$, let
\[
P_A=A(A'A)^{-1}A',
\qquad
M_A=I_n-P_A.
\]
Assumption \ref{var_order_regularity}, the boundedness of
$H_f^{-1}$, and the fact that $p_{\max}$ is fixed imply
\[
\min_{1\leq\ell\leq p_{\max}}
\frac{\sigma_{\min}(\widetilde Z_\ell)}{\sqrt n}
\]
is bounded away from zero in probability. Therefore, a standard
projection perturbation bound and \eqref{eq:factor_VAR_rate} give
\begin{equation}
\label{eq:VAR_projector_rate}
\max_{1\leq\ell\leq p_{\max}}
\left\|
P_{\hat Z_\ell}
-
P_{\widetilde Z_\ell}
\right\|
=
O_p(\sqrt{a_{NT}}).
\end{equation}

We prove separately that underfitted orders have a fixed fit loss and
that the improvement from overfitting is only $O_p(a_{NT})$.

\medskip
\noindent
\emph{Underfitted orders.}

Fix $\ell<p_0$. The residual matrices based on the estimated and
rotated true factors satisfy
\begin{align*}
&
\left\|
M_{\hat Z_\ell}\hat Y
-
M_{\widetilde Z_\ell}\widetilde Y
\right\|_F\\
&\qquad\leq
\|\Delta_Y\|_F
+
\left\|
P_{\hat Z_\ell}
-
P_{\widetilde Z_\ell}
\right\|
\|\widetilde Y\|_F\\
&\qquad=
O_p\left(\sqrt{na_{NT}}\right).
\end{align*}
Since both residual norms are $O_p(\sqrt n)$,
\begin{equation}
\label{eq:underfit_trace_approximation}
\max_{0\leq\ell\leq p_{\max}}
\left|
\hat v_\ell-\widetilde v_\ell
\right|
=
O_p(\sqrt{a_{NT}})
=
o_p(1).
\end{equation}

Let $\Sigma_\ell^\star$ denote the population residual covariance
matrix from the linear projection of $f_t$ on its first $\ell$ lags.
Minimality of the VAR order implies that, for every $\ell<p_0$,
\[
\Sigma_\ell^\star
=
\Sigma_u+D_\ell,
\qquad
D_\ell\succeq0,
\qquad
D_\ell\neq0.
\]
The usual VAR laws of large numbers therefore imply
\[
\widetilde v_\ell-\widetilde v_{p_0}
=
\operatorname{tr}
\left[
H_f
\left(
\hat\Sigma_\ell^F-
\hat\Sigma_{p_0}^F
\right)
H_f'
\right],
\]
where
\[
\hat\Sigma_\ell^F
=
\frac{1}{n}Y'M_{Z_\ell}Y.
\]
Since $\|H_f^{-1}\|=O_p(1)$ and
\[
\hat\Sigma_\ell^F-
\hat\Sigma_{p_0}^F
\to_p D_\ell,
\]
there exists a constant $d_\ell>0$ such that
\[
\mathbb P\left\{
\widetilde v_\ell-\widetilde v_{p_0}
\geq d_\ell
\right\}
\to1.
\]
Equation \eqref{eq:underfit_trace_approximation} then gives
\begin{equation}
\label{eq:underfit_trace_gap}
\mathbb P\left\{
\hat v_\ell-\hat v_{p_0}
\geq\frac{d_\ell}{2}
\right\}
\to1.
\end{equation}

Because $g(N,T)\to0$,
\begin{align*}
&
\operatorname{PIC}_{\mathrm{VAR}}(\ell)
-
\operatorname{PIC}_{\mathrm{VAR}}(p_0)\\
&\qquad=
\hat v_\ell-\hat v_{p_0}
-
\kappa q^2(p_0-\ell)g(N,T)
>0
\end{align*}
with probability approaching one.

\medskip
\noindent
\emph{Overfitted orders.}

Fix $\ell>p_0$. Let
\[
\hat X_\ell
=
\begin{pmatrix}
\hat F_{-(p_0+1)}&
\cdots&
\hat F_{-\ell}
\end{pmatrix}
\]
collect the additional estimated lags, and define
\[
\hat W_\ell
=
M_{\hat Z_{p_0}}\hat X_\ell.
\]
Define $\widetilde X_\ell$ and
\[
\widetilde W_\ell
=
M_{\widetilde Z_{p_0}}\widetilde X_\ell
\]
analogously using the rotated true factors.

The additional-lag covariance condition in Assumption
\ref{var_order_regularity} and the same perturbation argument used in
\eqref{eq:VAR_projector_rate} imply
\begin{equation}
\label{eq:extra_lag_projector_rate}
\left\|
P_{\hat W_\ell}
-
P_{\widetilde W_\ell}
\right\|
=
O_p(\sqrt{a_{NT}}).
\end{equation}

Let
\[
\hat R_{p_0}
=
M_{\hat Z_{p_0}}\hat Y,
\qquad
\widetilde R_{p_0}
=
M_{\widetilde Z_{p_0}}\widetilde Y.
\]
Equation \eqref{eq:factor_VAR_rate} and
\eqref{eq:VAR_projector_rate} give
\begin{equation}
\label{eq:VAR_residual_rate}
\left\|
\hat R_{p_0}
-
\widetilde R_{p_0}
\right\|_F
=
O_p\left(\sqrt{na_{NT}}\right).
\end{equation}

Under the true VAR$(p_0)$ model,
\[
Y=Z_{p_0}B_{p_0}'+U,
\]
where $U$ collects the innovations. Therefore,
\[
\widetilde R_{p_0}
=
M_{Z_{p_0}}UH_f'.
\]
Moreover, the martingale-difference property of $u_t$ and the fixed
dimension of the additional-lag regressors imply
\begin{equation}
\label{eq:true_overfit_projection}
\left\|
P_{\widetilde W_\ell}
\widetilde R_{p_0}
\right\|_F^2
=
O_p(1).
\end{equation}

By the Frisch--Waugh--Lovell theorem,
\[
n
\left(
\hat v_{p_0}-\hat v_\ell
\right)
=
\left\|
P_{\hat W_\ell}
\hat R_{p_0}
\right\|_F^2.
\]
Using \eqref{eq:extra_lag_projector_rate},
\eqref{eq:VAR_residual_rate}, and
\eqref{eq:true_overfit_projection},
\begin{align*}
\left\|
P_{\hat W_\ell}
\hat R_{p_0}
\right\|_F
&\leq
\left\|
P_{\widetilde W_\ell}
\widetilde R_{p_0}
\right\|_F\\
&\quad+
\left\|
P_{\hat W_\ell}
-
P_{\widetilde W_\ell}
\right\|
\|\widetilde R_{p_0}\|_F\\
&\quad+
\left\|
\hat R_{p_0}
-
\widetilde R_{p_0}
\right\|_F\\
&=
O_p(1)
+
O_p\left(\sqrt{na_{NT}}\right).
\end{align*}
Consequently,
\begin{equation}
\label{eq:overfit_trace_rate}
0
\leq
\hat v_{p_0}-\hat v_\ell
=
O_p\left(
\frac{1}{n}+a_{NT}
\right)
=
O_p(a_{NT}),
\end{equation}
because $n/T\to1$ and $T^{-1}\leq a_{NT}$.

It follows that
\begin{align*}
&
\operatorname{PIC}_{\mathrm{VAR}}(\ell)
-
\operatorname{PIC}_{\mathrm{VAR}}(p_0)\\
&\qquad=
-
\left(
\hat v_{p_0}-\hat v_\ell
\right)
+
\kappa q^2(\ell-p_0)g(N,T).
\end{align*}
Dividing by $g(N,T)$ and using
\[
\frac{a_{NT}}{g(N,T)}\to0,
\]
we obtain
\[
\frac{
\operatorname{PIC}_{\mathrm{VAR}}(\ell)
-
\operatorname{PIC}_{\mathrm{VAR}}(p_0)
}{
g(N,T)
}
=
\kappa q^2(\ell-p_0)+o_p(1)>0
\]
with probability approaching one.

The candidate set
\[
\{0,\ldots,p_{\max}\}
\]
is fixed and finite. Hence, the preceding comparisons hold
simultaneously for every $\ell\neq p_0$, and therefore
\[
\mathbb P\{\hat p=p_0\}\to1.
\]
\end{proof}

\subsection{Sufficient Conditions for Assumptions
\ref{pd} and \ref{lower_bound_common}}
\label{apx:discussion_assumption}

This subsection provides sufficient conditions for the factor-side restriction in Assumption \ref{pd} and for the uniform separation condition in Assumption \ref{lower_bound_common}. The loading-side pervasiveness condition in Assumption \ref{pd} is maintained separately.

\paragraph{Factor-side nondegeneracy.}

Recall that
\[
g_t=
\begin{pmatrix}
f_t'&
f_{t-1}'&
\cdots&
f_{t-m_0+1}'
\end{pmatrix}'
\in\mathbb R^{r_0},
\qquad
r_0=q_0m_0,
\]
and that $G$ is the $T\times r_0$ matrix whose $t$-th row is $g_t'$. Suppose that $(f_t)$ is covariance stationary with spectral density matrix $S_f(\omega)$ and that
\[
\frac{1}{T}\sum_{t=1}^Tg_tg_t'
\to_p
\mathbb E(g_tg_t').
\]
We show that
\[
\mathbb E(g_tg_t')>0
\]
whenever $S_f(\omega)>0$ almost everywhere on a measurable set of frequencies with positive Lebesgue measure. Let
\[
a=
\begin{pmatrix}
a_0'&
a_1'&
\cdots&
a_{m_0-1}'
\end{pmatrix}'
\in\mathbb R^{r_0}
\]
be nonzero, where $a_k\in\mathbb R^{q_0}$. Then
\[
a'g_t
=
\sum_{k=0}^{m_0-1}a_k'f_{t-k}.
\]
Moreover,
\[
a'\mathbb E(g_tg_t')a
=
\mathbb E\!\left[(a'g_t)^2\right]
\geq
\operatorname{Var}(a'g_t).
\]
It therefore suffices to show that
\[
\operatorname{Var}(a'g_t)>0
\qquad
\text{for every }a\neq0.
\]
Define the vector-valued trigonometric polynomial
\[
A(\omega)
=
\sum_{k=0}^{m_0-1}a_ke^{ik\omega}
\in\mathbb C^{q_0}.
\]
By the spectral representation of a finite linear filter,
\[
\operatorname{Var}(a'g_t)
=
\frac{1}{2\pi}
\int_{-\pi}^{\pi}
A(\omega)^*S_f(\omega)A(\omega)\,d\omega,
\]
where $^*$ denotes conjugate transpose.

Let
\[
\mathcal B
=
\left\{
\omega\in[-\pi,\pi]:
S_f(\omega)>0
\right\},
\]
and suppose that $\mathcal B$ has positive Lebesgue measure. Because $a\neq0$, at least one component of $A(\omega)$ is a nonzero trigonometric polynomial. Hence,
\[
\mathcal Z
=
\left\{
\omega\in[-\pi,\pi]:
A(\omega)=0
\right\}
\]
has Lebesgue measure zero. It follows that $\mathcal B\setminus\mathcal Z$ has positive measure. For every $\omega\in\mathcal B\setminus\mathcal Z$,
\[
A(\omega)^*S_f(\omega)A(\omega)>0.
\]
Since $S_f(\omega)$ is positive semidefinite almost everywhere, the integrand is nonnegative almost everywhere and strictly positive on a set of positive measure. Therefore,
\[
\operatorname{Var}(a'g_t)>0.
\]
Thus,
\[
a'\mathbb E(g_tg_t')a>0
\qquad
\text{for every }a\neq0,
\]
which establishes
\[
\mathbb E(g_tg_t')>0.
\]
The assumed law of large numbers then gives
\[
\frac{1}{T}G'G
=
\frac{1}{T}\sum_{t=1}^Tg_tg_t'
\to_p
\mathbb E(g_tg_t')
\succ0.
\]
Hence, the factor-side restriction in Assumption \ref{pd} holds.

\paragraph{Uniform separation under underfactoring.}

Let
\[
C_0
=
\sum_{k=0}^{m_0-1}F_{-k}\lambda_k'
\in\mathbb R^{T\times N}
\]
denote the sample matrix of true common components. For
$q=0,\ldots,q_0-1$ and $m=1,\ldots,m_{\max}$, define the optimized sample mean-squared approximation error
\[
Q_{NT}(q,m)
=
\inf_{C\in\mathcal C_{q,m}}
\frac{1}{NT}\|C_0-C\|_F^2,
\]
where $\mathcal C_{q,m}$ is the sample dynamic approximation class defined in Section \ref{sec:test}. In particular, $\mathcal C_{0,m}=\{0\}$.

Let $\mathscr C_{q,m}$ denote the corresponding population class of second-order stationary approximations of the form
\[
\widetilde\chi_t
=
\sum_{k=0}^{m-1}B_kh_{t-k},
\qquad
h_t\in\mathbb R^q,
\]
where $(h_t)$ is jointly second-order stationary with $(\chi_t)$. Define the optimized population mean-squared approximation error by
\[
Q_N(q,m)
=
\inf_{\widetilde\chi\in\mathscr C_{q,m}}
\frac{1}{N}
\mathbb E
\left[
\|\chi_t-\widetilde\chi_t\|^2
\right].
\]

The following high-level condition imposes the required convergence of the optimized sample criterion to its population counterpart.

\begin{assumption}[Convergence of optimized approximation criteria]
\label{optimized_criterion_convergence}
Let $q_0$ and $m_{\max}$ be fixed. Then
\[
\max_{\substack{0\leq q<q_0\\1\leq m\leq m_{\max}}}
\left|
Q_{NT}(q,m)-Q_N(q,m)
\right|
\to_p0.
\]
\end{assumption}

Because the candidate set is finite, Assumption \ref{optimized_criterion_convergence} is equivalent to requiring
\[
Q_{NT}(q,m)-Q_N(q,m)\to_p0
\]
for each fixed pair
\[
q=0,\ldots,q_0-1,
\qquad
m=1,\ldots,m_{\max}.
\]
The assumption directly imposes convergence of the optimized criterion and therefore does not require an interchange of an infimum and a pointwise ergodic limit.

Let $S_{\chi,N}(\omega)$ denote the $N\times N$ spectral density matrix of the true common component, and let
\[
\mu_{1,N}(\omega)
\geq
\mu_{2,N}(\omega)
\geq
\cdots
\geq
\mu_{N,N}(\omega)
\geq0
\]
denote its ordered eigenvalues. Under the true dynamic factor structure, at most the first $q_0$ eigenvalues are nonzero at each frequency.

\begin{proposition}[A sufficient condition for uniform separation]
\label{prop:sufficient_lower_bound_common}
Suppose that Assumption \ref{optimized_criterion_convergence} holds and that $m_{\max}\geq m_0$ is fixed. Suppose further that there exists a constant $c_\mu>0$ such that
\[
\mathbb P\left(
\frac{1}{2\pi N}
\int_{-\pi}^{\pi}
\mu_{q_0,N}(\omega)\,d\omega
<c_\mu
\right)
\to0.
\]
Then Assumption \ref{lower_bound_common} holds.
\end{proposition}

\begin{proof}
Every process in $\mathscr C_{q,m}$ has dynamic rank at most $q$. By the frequency-domain best rank-$q$ approximation bound,
\[
Q_N(q,m)
\geq
\frac{1}{2\pi N}
\int_{-\pi}^{\pi}
\sum_{j=q+1}^{q_0}
\mu_{j,N}(\omega)\,d\omega.
\]
Since $q<q_0$,
\[
Q_N(q,m)
\geq
\frac{1}{2\pi N}
\int_{-\pi}^{\pi}
\mu_{q_0,N}(\omega)\,d\omega.
\]
Therefore,
\[
\mathbb P\left(
\min_{\substack{0\leq q<q_0\\1\leq m\leq m_{\max}}}
Q_N(q,m)
<c_\mu
\right)
\to0.
\]

Assumption \ref{optimized_criterion_convergence} then implies
\[
\mathbb P\left(
\min_{\substack{0\leq q<q_0\\1\leq m\leq m_{\max}}}
Q_{NT}(q,m)
<\frac{c_\mu}{2}
\right)
\to0.
\]

For any $C\in\mathcal C_{q,m}$,
\[
\operatorname{rank}(C_0-C)
\leq
\operatorname{rank}(C_0)+\operatorname{rank}(C)
\leq
r_0+qm.
\]
Consequently,
\[
\|C_0-C\|_F^2
\leq
(r_0+qm)\|C_0-C\|^2.
\]
Taking the infimum over $C\in\mathcal C_{q,m}$ gives
\[
\frac{1}{NT}
\delta_{q+1,m}^2(C_0)
\geq
\frac{Q_{NT}(q,m)}{r_0+qm}.
\]

Let
\[
R_{\max}
=
r_0+(q_0-1)m_{\max}.
\]
For every
\[
0\leq q<q_0,
\qquad
1\leq m\leq m_{\max},
\]
we have
\[
r_0+qm\leq R_{\max}.
\]
It follows that
\[
\mathbb P\left(
\min_{\substack{0\leq q<q_0\\1\leq m\leq m_{\max}}}
\frac{1}{\sqrt{NT}}
\delta_{q+1,m}(C_0)
<
\sqrt{\frac{c_\mu}{2R_{\max}}}
\right)
\to0.
\]
Thus, for every $\epsilon>0$, one may choose any
\[
0<c_2<
\sqrt{\frac{c_\mu}{2R_{\max}}}
\]
so that
\[
\mathbb P\left(
\frac{1}{\sqrt{NT}}
\delta_{q+1,m}(C_0)
<c_2
\right)
\leq\epsilon
\]
for all
\[
q=0,\ldots,q_0-1,
\qquad
m=1,\ldots,m_{\max},
\]
and all sufficiently large $N$ and $T$. This is Assumption
\ref{lower_bound_common}.
\end{proof}

The spectral condition in Proposition
\ref{prop:sufficient_lower_bound_common} follows from the factor and loading conditions used elsewhere in the paper. Define the frequency-domain loading matrix
\[
\Lambda_N(\omega)
=
\sum_{k=0}^{m_0-1}
\lambda_ke^{-ik\omega}
\in\mathbb C^{N\times q_0}.
\]
Then
\[
S_{\chi,N}(\omega)
=
\Lambda_N(\omega)
S_f(\omega)
\Lambda_N(\omega)^*.
\]
Let
\[
V(\omega)
=
\begin{pmatrix}
I_{q_0}\\
e^{-i\omega}I_{q_0}\\
\vdots\\
e^{-i(m_0-1)\omega}I_{q_0}
\end{pmatrix}
\in\mathbb C^{r_0\times q_0}.
\]
Since
\[
\Gamma
=
\begin{pmatrix}
\lambda_0&
\lambda_1&
\cdots&
\lambda_{m_0-1}
\end{pmatrix},
\]
we have
\[
\Lambda_N(\omega)=\Gamma V(\omega)
\]
and
\[
V(\omega)^*V(\omega)=m_0I_{q_0}.
\]
It follows that
\begin{align*}
\lambda_{\min}
\left(
\frac{1}{N}
\Lambda_N(\omega)^*
\Lambda_N(\omega)
\right)
&=
\lambda_{\min}
\left[
V(\omega)^*
\left(
\frac{\Gamma'\Gamma}{N}
\right)
V(\omega)
\right]\\
&\geq
m_0
\lambda_{\min}
\left(
\frac{\Gamma'\Gamma}{N}
\right).
\end{align*}

The nonzero eigenvalues of $S_{\chi,N}(\omega)$ are the eigenvalues of
\[
S_f(\omega)^{1/2}
\Lambda_N(\omega)^*
\Lambda_N(\omega)
S_f(\omega)^{1/2}.
\]
Therefore,
\[
\frac{\mu_{q_0,N}(\omega)}{N}
\geq
m_0
\lambda_{\min}
\left(
\frac{\Gamma'\Gamma}{N}
\right)
\lambda_{\min}\bigl(S_f(\omega)\bigr).
\]

Suppose that $S_f(\omega)\succ0$ almost everywhere on a measurable set $\mathcal B$ of positive Lebesgue measure. Then
\[
c_f
=
\frac{1}{2\pi}
\int_{\mathcal B}
\lambda_{\min}\bigl(S_f(\omega)\bigr)\,d\omega
>0.
\]
Consequently,
\[
\frac{1}{2\pi N}
\int_{-\pi}^{\pi}
\mu_{q_0,N}(\omega)\,d\omega
\geq
m_0c_f
\lambda_{\min}
\left(
\frac{\Gamma'\Gamma}{N}
\right).
\]
By the loading-side condition in Assumption \ref{pd}, the right-hand side is bounded away from zero in probability. Hence, the spectral condition in Proposition
\ref{prop:sufficient_lower_bound_common} holds.

In particular, a stable and invertible VARMA factor process driven by innovations with positive-definite covariance has
\[
S_f(\omega)\succ0
\qquad
\text{for every }\omega\in[-\pi,\pi].
\]
Together with the loading-side pervasiveness condition in Assumption \ref{pd} and Assumption
\ref{optimized_criterion_convergence}, such a process therefore satisfies Assumption \ref{lower_bound_common}.

\subsection{Estimation using Discrete (Fast) Fourier Transform}\label{sec:DFF}
The objective function is given by
\begin{align*}
    \frac{1}{NT} \fnorm{X - \sum_{k=0}^{m-1} F_{-k} \lambda_k'}^2 = \frac{1}{NT}\sum_{t=1}^T\left(x_t-\sum_{k=0}^{m-1} \lambda_k f_{t-k}\right)'\left(x_t-\sum_{k=0}^{m-1} \lambda_k f_{t-k}\right)
\end{align*}
The first order conditions for $f_t$ are given by
\begin{align*}
    \sum_{j=\max(0,1-t)}^{\min(m-1,T-t)}\lambda_j'\left(x_{t+j}-\sum_{k=0}^{m-1}\lambda_kf_{t+j-k}\right) = 0
\end{align*}
for $t=2-m,\ldots,T$. 

Let $(T+m-1)\times N$ dimensional matrix $\tilde X$ and $(T+m-1)\times q$ matrix $\tilde F_{-k}$ be given as
\begin{align*}
    \tilde X = \begin{pmatrix}
        0_{(m-1)\times N}\\
        X
    \end{pmatrix}, \tilde F_{-k} = \begin{pmatrix}
        0_{(m-1)\times q}\\
        F_{-k}
    \end{pmatrix},
\end{align*}
i.e., matrices $X$ and $F_{-k}$ that have been augmented with $m-1$ zero vectors on top. 

Given $a = (a_1,\ldots,a_{T+m-1})'\in\mathbb{R}^{T+m-1}$, we define $P$ as the one-step upward cyclic permutation that acts as
\begin{align*}
    (Pa)_t = a_{t+1(\text{mod }T+m-1)},
\end{align*}
i.e., $P$ sends $(a_1,\ldots,a_{T+m-1})$ to $(a_2,\ldots,a_{T+m-2},a_1)$. Let $P_j=P^{j}$, i.e., $j$-step upward cyclic permutation. Then, the first order conditions for $(f_t)$ may be summarized as
\begin{align*}
    \sum_{j=0}^{m-1}P_j\left(\sum_{k=0}^{m-1}\hat{\tilde F}_{-k}\hat \lambda_k'\right)\hat\lambda_j &= \sum_{j=0}^{m-1}P_j\tilde X\hat\lambda_j,
\end{align*}
where $\hat{\tilde F}_{-k}$ is defined using $\hat F_{-k}$. We restrict the boundary values of $f_t$ to zero, i.e., $f_t = 0$ for $t=2-m,\ldots,0$ and $t=T-m+2,\ldots,T$, with the convention that when $m=1$, $t=1,\ldots,0$ and $t=T+1,\ldots,T$ define an empty set. Then, we have
\begin{align*}
    \sum_{k,j=0}^{m-1}P_{j-k}\hat{\tilde F}\hat \lambda_k'\hat\lambda_j &= \sum_{j=0}^{m-1}P_j\tilde X\hat\lambda_j.
\end{align*}
Estimation proceeds by iterating between
\begin{align}\label{foc_lambda2}
    \sum_{k=0}^{m-1} \hat \lambda_k \hat F_{-k}' \hat F_{-j} = X' \hat F_{-j}, \quad j=0,\ldots,m-1.
\end{align}
and
\begin{align}\label{foc_f2}
    \sum_{k,j=0}^{m-1} P_{j-k} \hat{\tilde F} \hat \lambda_k' \hat \lambda_j = \sum_{j=0}^{m-1} P_j \tilde X \hat \lambda_j.
\end{align}
Given $\hat G$, the equation \ref{foc_lambda2} updates $\hat\Gamma$ using
\begin{align*}
    \hat\Gamma = X'\hat G(\hat G'\hat G)^{-1}
\end{align*}
Given $\hat \Gamma$, the equation \ref{foc_f2} updates $\hat{\tilde{F}}$ by
\begin{align*}
\vecr(\hat{\tilde F})
 = \left(\sum_{k,j=0}^{m-1}\left(P_{j-k} \otimes \hat \lambda_j' \hat \lambda_k\right)\right)^{-1}\sum_{j=0}^{m-1}\left(P_j\otimes \hat \lambda_j'\right)\vecr(\tilde X)
\end{align*}
so that naive update step for $\hat{\tilde F}$ involves inverse of $(T+m_1+m_2)q\times (T+m_1+m_2)q$ dimensional matrix. However, the structure in the Toeplitz matrix $\left(\sum_{k,j=0}^{m-1}\left(P_{j-k} \otimes \hat \lambda_j' \hat \lambda_k\right)\right)$ reduces this step into computing inverse of $q\times q$ dimensional matrices $T+m-1$ times.

We define the matrix $S$ as
\begin{align*}
    S = \sum_{j,k=0}^{m-1}\left(P_{j-k}\otimes \frac{1}{N}\hat \lambda_j'\hat \lambda_k\right)=\sum_{h=-(m-1)}^{m-1}(P_h\otimes B_h).
\end{align*}
where
\begin{align*}
    B_h = \frac{1}{N}\sum_{k=0}^{m-1}\hat \lambda_{k+h}'\hat \lambda_k, \quad h=j-k,
\end{align*}
    with the convention that $\hat\lambda_k = 0$ for $k<0$ and $k>m-1$. Let $a=(a_1',\ldots,a_{T+m-1}')'$ be $(T+m-1)q$ dimensional vector with $q$ dimensional blocks given by $(a_t)$ and $b=(b_{1}',\ldots,b_{T+m-1}')'$ be $(T+m-1)q$ dimensional vector given by
    \begin{align*}
        b = Sa
    \end{align*}
    We may define the blockwise discrete Fourier transform for frequencies $\omega_\ell = 2\pi \ell/(T+m-1), \ell=1,\ldots,T+m-1$ to both input and output sequence
    \begin{align*}
        a(\omega_\ell) = \sum_{t=1}^{T+m-1}a_t e^{-i\omega_\ell t},\quad        b(\omega_\ell) = \sum_{t=1}^{T+m-1}b_t e^{-i\omega_\ell t}
    \end{align*}
    which give a $q$-vector for each $\ell=1,\ldots,T+m-1$. We have
    \begin{align*}
        b(\omega_\ell) = B(e^{i\omega_\ell})a(\omega_\ell)
    \end{align*}
    with 
    \begin{align*}
        B(e^{i\omega}) = \sum_{h=-(m-1)}^{m-1}B_he^{i\omega h}.
    \end{align*}
    Thus, solving for $a$ given $S$ and $b$ in $b=Sa$ reduces to computing
    \begin{align*}
        a(\omega_\ell) = B(e^{i\omega_\ell})^{-1}b(\omega_\ell)
    \end{align*}
    for each $\ell=1,\ldots,T+m-1$. After we obtain $a(\omega_\ell)$, we may obtain $(a_t)$ using inverse Fourier transform.

\subsection{Monetary Policy Shock Identification using Dynamic Factor Model}\label{sec:apd-mp}

Let $(x_t)$ denote an $N$-dimensional time series of macroeconomic variables generated by the dynamic factor model
\begin{align}\label{dfm-shock}
    x_t = \sum_{k=0}^{m-1}\lambda_k f_{t-k} + \varepsilon_t,
\end{align}
where $q$-dimensional factor process $(f_t)$ follows a VAR($p$)
\[
f_t = \sum_{l=1}^p A_l f_{t-l} + u_t,
\qquad 
u_t = f_t - \mathbb{E}[f_t \mid (f_s)_{s<t}],
\]
with innovation covariance matrix $\Sigma_u = \mathbb{E}(u_t u_t')$. Let $\Sigma_u = QQ'$ be a Cholesky decomposition. Then, for any invertible $q\times q$ matrix $H$, the factor dynamics can be written as
\[
f_t = \sum_{l=1}^p A_l f_{t-l} + QH e_t,
\qquad 
e_t = H^{-1} Q^{-1} u_t .
\]
Substituting into the measurement equation yields the structural moving-average representation
\[
x_t
=
\sum_{k=0}^{m-1}
\lambda_k L^k
\left(I - \sum_{l=1}^p A_l L^l \right)^{-1}
QH e_t
+ \varepsilon_t ,
\]
where $L$ is a lag operator so that the impulse response of $x_t$ to the lagged structural shocks $e_{t-h}$ is given by
\[
\sum_{k=0}^{m-1}
\lambda_k L^k
\left(I - \sum_{l=1}^p A_l L^l \right)^{-1}
QH .
\]
The matrix $H$ can be identified using standard identification schemes from the SVAR literature.

For comparison, we also estimate the model as in \cite{Forni-Gambetti-2010}. In this approach, we first estimate the static representation
\[
x_t = \Gamma g_t + \varepsilon_t ,
\]
where $g_t\in\mathbb{R}^{qm}$ and then fit a VAR($\tilde p$) to the static factors, 
\[
g_t = \sum_{l=1}^{\tilde p} \tilde A_l g_{t-l} + \tilde u_t,
\qquad 
\tilde u_t = g_t - \mathbb{E}[g_t \mid (g_s)_{s<t}].
\]
Given the true model in \eqref{dfm-shock}, the projection residual $\tilde u_t$ is given by $\tilde u_t = \begin{pmatrix}
    u_t',0,\ldots,0
\end{pmatrix}'$ and the innovation covariance matrix $\tilde \Sigma_u = \mathbb{E}(\tilde u_t\tilde u_t')$ is $r$ dimensional with rank $q$. Let $R$ be $r\times q$ dimensional matrix satisfying $\tilde \Sigma_u = RR'$. As before, for any invertible $q\times q$ matrix $H$, the factor dynamics can be written as
\[
g_t = \sum_{l=1}^{\tilde p} \tilde A_l g_{t-l} + RH e_t,
\qquad 
e_t = H^{-1} (R'R)^{-1}R'\tilde u_t .
\]
Structural shocks $e_t$ are again identified through restrictions on $H$. The low rank representation $R$ of $\tilde \Sigma_u$ is estimated using the leading eigenvectors of the residual sample covariance matrix as in \cite{Forni-Gambetti-2010}.

We estimate the dynamic effects of monetary policy shocks using both approaches, where identification follows \cite{Eichenbaum-Evans-1995}. In particular, industrial production and CPI are assumed not to respond contemporaneously to a monetary policy shock. 
All other variables, including exchange rates, are allowed to respond contemporaneously to innovations in the federal funds rate\footnote{Additional zero restrictions, other than on contemporaneous response of industrial production and CPI, can be imposed. In that case, the restrictions become over-identifying and must be enforced during estimation. This is straightforward in our alternating least squares framework. For example, in the factor loading update step, one may impose zero restrictions via constrained least squares on $(\hat f_t)$.}.



We use the FRED--MD data set of \cite{McCracken-Ng-2015}. We perform estimation using the sample period from March 1973 to November 2007. The sample excludes the fixed exchange rate regime and the global financial crisis, following \cite{Forni-Gambetti-2010}. The sample contains $T=417$ observations on $N=124$ macroeconomic variables. Stationary transformations are kept minimal: prices enter in log differences, while interest rates are in levels.

The estimated dynamic structure $(q,m)$ equals $(6,2)$, $(4,2)$, and $(2,2)$ under $DC_2$, $PC_2$, and $IC_2$, respectively. We set $(q,m)=(4,2)$ for the subsequent analysis. VAR lag orders are selected by BIC for both the estimated dynamic factors $(\hat f_t)$ and static factors $(\hat g_t)$, yielding $(p,\tilde p)=(3,1)$.

The fraction of variance explained by the model with $(q,m)=(4,2)$ is $0.56$. This exceeds the $0.46$ obtained from a static factor model with the same number of factors, $(q,m)=(4,1)$, but is below the $0.59$ achieved by the static representation of the original model with $(q,m)=(8,1)$. This is to be expected because the model with $(qm,1)$ is the most flexible, while the model with $(q,1)$ is the least flexible among the three. The estimation of the VAR$(\tilde p)$ for $(g_t)$ requires a second-step dimension reduction of the residuals $(\tilde u_t)$, reducing the dimension from $r=8$ to $q=4$. The explained variation in this second step is $0.85$. Accounting for this loss, the effective explained variation of the model $(q,m)=(8,1)$ is $0.50$, which is substantially lower than the explained variation $0.59$ of the model $(q,m) = (4,2)$.

\begin{figure}[H]
    \centering
    \includegraphics[width=1\linewidth]{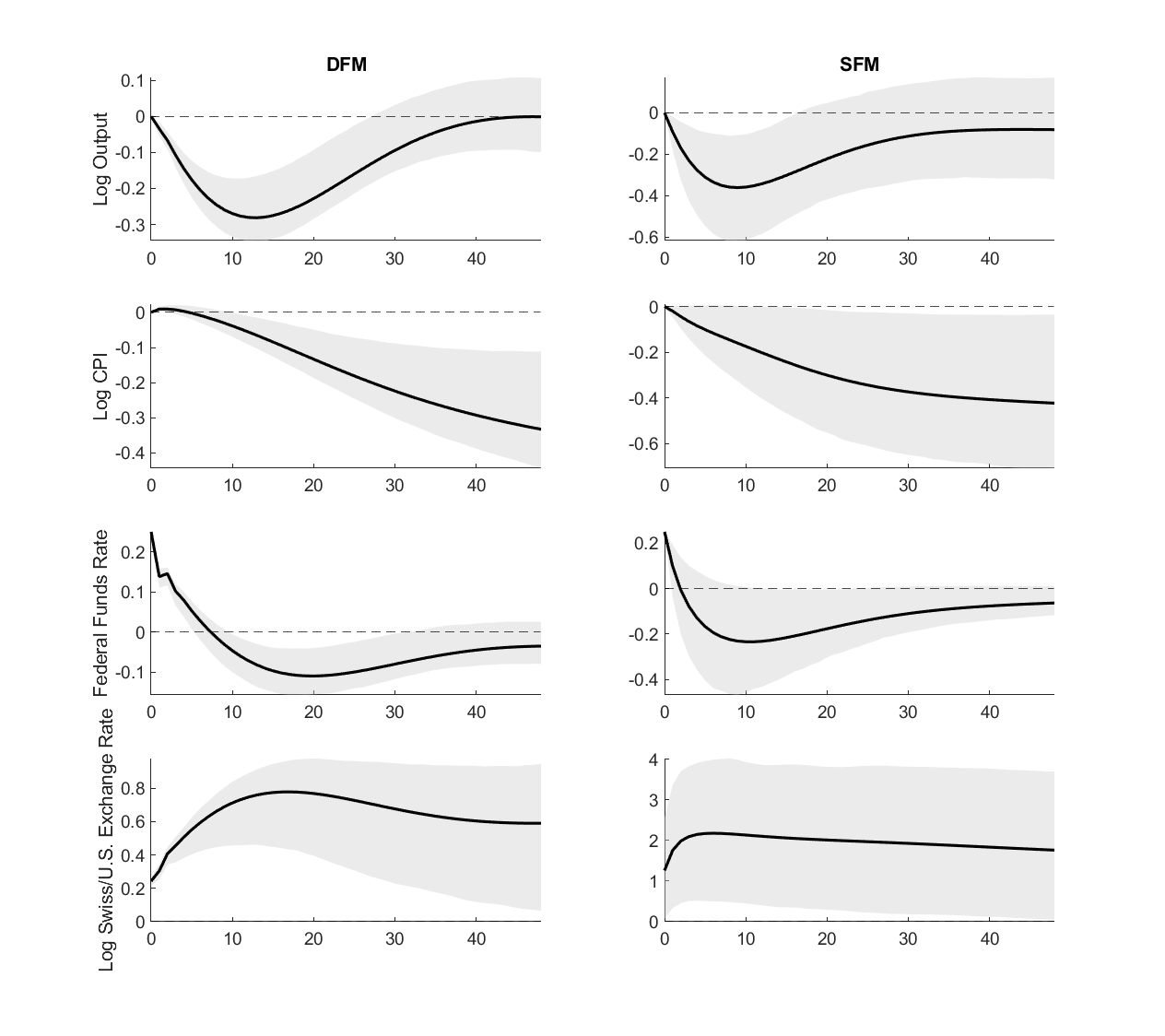}
    \caption{Estimated impulse responses of log output, log CPI, the federal funds rate, and the log exchange rate to a monetary policy shock, based on the dynamic factor model (left) and the static factor model (right).}
    \label{fig:irf_q3}
\end{figure}

Figure~\ref{fig:irf_q3} reports impulse responses of log output, log inflation, the federal funds rate, and the log exchange rate. Confidence bands are constructed using a residual bootstrap at the 90\% level. The point estimates are qualitatively similar across the two methods. However, confidence bands based on the dynamic factor model (left) are narrower, indicating efficiency gains from direct estimation of dynamic factors.

\end{document}